\documentclass[journal,onecolumn]{IEEEtran}

\usepackage{ifpdf}

\usepackage{cite}

\ifCLASSINFOpdf
\usepackage[pdftex]{graphicx}
\else
\usepackage[dvips]{graphicx}
\fi
\usepackage{amsmath}
\usepackage{amssymb}
\usepackage{mathtools}
\usepackage{amsfonts}
\usepackage{dsfont}
\usepackage{bbm}

\usepackage{caption}
\usepackage[ruled,linesnumbered]{algorithm2e}
\usepackage{setspace}
\SetNlSty{textbf}{}{:}
\newcommand{\FuncName}[1]{\mbox{\normalfont\textsc{#1}}}
\SetKwFor{ForEach}{for each}{do}{end}

\SetCommentSty{mycommfont}
\SetKwComment{Comment}{$\triangleright$\ }{}

\SetKwFunction{column}{\FuncName{DecodeColumnBottom}}
\SetKwFunction{decode}{\FuncName{DecodeSegment}}
\SetKwFunction{extract}{\FuncName{GetDeltaS}}

\makeatletter
\newenvironment{algocontinue}[2][]
  {\renewcommand{\algorithmcfname}{Algorithm #2}%
   \begin{algorithm*}[#1]
   \long\def\@caption##1[##2]##3{%
     \par
     \begingroup\@parboxrestore
     \if@minipage\@setminipage\fi
     \normalsize \@makecaption{\AlCapSty{\AlCapFnt\algorithmcfname}}{\ignorespaces ##3}%
     \par\endgroup
   }}
  {\end{algorithm*}}
\makeatother

\usepackage{array}

\usepackage{arydshln}
\usepackage{multirow}

\ifCLASSOPTIONcompsoc
\usepackage[caption=false,font=normalsize,labelfont=sf,textfont=sf]{subfig}
\else
\usepackage[caption=false,font=footnotesize]{subfig}
\usepackage{float}

\usepackage{url}
\usepackage{amsthm}

\usepackage[dvipsnames]{xcolor}
\definecolor{darkgreen}{rgb}{0.2, 0.7, 0.2}
\definecolor{darkblue}{rgb}{0.2,0.2, 0.8}

\usepackage{tikz}
\usetikzlibrary{matrix,positioning,fit,calc,patterns}
\usepackage{pgfplots}

\usepackage{hyperref}
\hypersetup{
	colorlinks,
	citecolor=darkgreen,
	filecolor=darkblue,
	linkcolor=darkblue,
	urlcolor=darkblue
}

\usepackage[normalem]{ulem}
\usepackage{mathrsfs}
\usepackage{tikz}
\usetikzlibrary{matrix,positioning,fit,calc,decorations.pathreplacing}

\usepackage{extarrows}

\usepackage{enumitem}

\newcolumntype{L}[1]{>{\raggedright\arraybackslash}p{#1}}
\newcolumntype{C}[1]{>{\centering\arraybackslash}p{#1}}
\newcolumntype{R}[1]{>{\raggedleft\arraybackslash}p{#1}}

\newcommand{\mcal}[1]{\mathcal{#1}}
\newcommand{\intv}[1]{\left[#1\right]}
\newcommand {\sen}[1]{\mcal {#1}}
\newcommand {\seq}[1]{\set{#1}}
\newcommand{\set}[1]{\left\lbrace #1 \right\rbrace}
\newcommand{\size}[1]{ \left| #1 \right|}

\newcommand{\ind}[2]{\mathsf{ind}_{#1} (#2)}
\newcommand{\1}[1]{\mathds{1}_{\{#1\}}}
\DeclarePairedDelimiter{\ceil}{\lceil}{\rceil}

\newtheorem{lm}{Lemma}

\newtheorem{rem}{Remark}
\newtheorem{defi}{Definition}
\newtheorem{prop}{Proposition}
\newtheorem{cor}{Corollary}
\newtheorem{thm}{Theorem}

\newtheorem{examplex}{Example}
\newenvironment{ex}
  {\pushQED{\qed}\examplex}
  {\popQED\endexamplex}

\newcommand{\bA}{\mathbf{A}}

\newcommand{\bB}{\mathbf{B}}

\newcommand{\bC}{\mathbf{C}}

\newcommand{\bD}{\mathbf{D}}

\newcommand{\bI}{\mathbf{I}}

\newcommand{\bM}{\mathbf{M}}

\newcommand{\bO}{\mathbf{O}}

\newcommand{\bX}{\mathbf{X}}

\newcommand{\cC}{\mathcal{C}}

\newcommand{\cG}{\mathcal{G}}
\newcommand{\cH}{\mathcal{H}}

\newcommand{\cZ}{\mathcal{Z}}

\usepackage{pdfrender}
\definecolor{CharFillColor}{rgb}{0,0,0}
\newcommand{\e}[1]{\textpdfrender{
		TextRenderingMode=FillStroke,
		LineWidth=.3pt,
		FillColor=white,
	}{\mathbf{#1}}}
\newcommand{\f}[1]{\textpdfrender{
		TextRenderingMode=FillStroke,
		LineWidth=.3pt,
		FillColor=black,
	}{\mathbf{#1}}}

\usepackage{mathrsfs}

\newcommand{\eSMat}[4]{{\e{#1}}_{\xleftarrow{#3} {#4}}}
\newcommand{\eSMatt}[5]{{\e{#1}_{#2}}_{\xleftarrow{#4} {#5}}}
\newcommand{\fSMat}[4]{{\f{#1}}_{\xleftarrow{#3} {#4}}}

\newcommand{\tOMat}[3]{\boldsymbol{\Delta}_{\xleftarrow{#2} {#3}}}

\newcommand{\md}[1]{\mathsf{mode}\left(#1\right)}

\newcommand{\term}{\mathsf{Term}}

\newcommand{\repMat}[2]{\mathbf{\Xi}^{#1,(#2)}}

\newcommand{\nd}[1]{\left<#1\right>}

\newcommand{\up}[1]{\overline{#1}} 
\newcommand{\down}[1]{\underline{#1}} 

\newcommand{\enc}{\mathbf{\Psi}}

\newcommand{\encdc}{\enc[\sen K, :]}
\newcommand{\gdc}{\mathbf{\Gamma}[\sen K, :]}
\newcommand{\ddc}{\mathbf{\Upsilon}[\sen K, :]}
\newcommand{\pdc}{\mathbf{\Psi}[\sen K, :]}

\newcommand{\TM}{\mathbf{T}}
\newcommand{\PM}{\mathbf{P}}

\newcommand{\MM}{\mathbf{M}}

\newcommand{\rep}[2]{\mathbf{R}^{#1}(#2)}
\newcommand{\erep}[2]{\mathbf{R}^{#1}(#2)}

\newcommand{\s}[2]{\sigma_{#1}{(#2)}}

\newcommand{\mdcnt}{t}
\newcommand{\dgp}[1]{\cG_1(\down{\e{#1}})}
\newcommand{\ngp}[1]{\cG_2(\down{\e{#1}})}
\newcommand{\pgp}[1]{\cG_3(\down{\e{#1}})}

\newcommand{\dgpi}[2]{\cG_1(\down{\e{#1}_{#2}})}
\newcommand{\ngpi}[2]{\cG_2(\down{\e{#1}_{#2}})}
\newcommand{\pgpi}[2]{\cG_3(\down{\e{#1}_{#2}})}

\newcommand{\seg}{S}
\newcommand{\eSeg}{\e{\seg}}
\newcommand{\fSeg}{\f{\seg}}

\makeatletter
\def\namedlabel#1#2{\begingroup
	#2%
	\def\@currentlabel{#2}%
	\phantomsection\label{#1}\endgroup
}

\begin{document}
	\title{Cascade Codes For Distributed Storage Systems}
	
	\author{Mehran Elyasi and Soheil~Mohajer,~\IEEEmembership{Member,~IEEE}
		\thanks{M. Elyasi and 
			S. Mohajer are with the Department of Electrical and Computer Engineering, University of Minnesota, Twin Cities, MN 55455, USA, (email: \{melyasi, soheil\}@umn.edu).}
		\thanks{This work was supported in part by the National Science Foundation under Grant CCF-1617884.} 	
		\thanks{This paper is presented in part at the IEEE International Symposium on Information Theory (ISIT), 2018 \cite{elyasi2018cascade}.}}

	\maketitle

	\begin{abstract}
		A novel coding scheme for exact repair-regenerating codes is presented in this paper. 
		The codes proposed in this work can trade between the repair bandwidth of nodes (number of downloaded symbols from each surviving node in a repair process) and the required storage overhead of the system. These codes work for general system parameters $(n,k,d)$, which are the total number of nodes, the number of nodes suffice for data recovery, and the number of helper nodes in a repair process, respectively. The proposed construction offers a unified scheme to develop exact-repair regenerating codes for the entire trade-off, including the MBR and MSR points. We conjecture that the new storage-vs.-bandwidth trade-off achieved by the proposed codes is optimum. Some other key features of this code include:  the construction is linear; the required field size is only $\Theta(n)$; and the code parameters and in particular sub-packetization level is at most $(d-k+1)^k$; which is independent of the number of the parity nodes. Moreover, the proposed repair mechanism is \emph{helper-independent}, that is the data sent from each helper only depends on the identity of the helper and failed nodes, but independent of the identity of other helper nodes participating in the repair process.
	\end{abstract}

	\section{introduction}
	The dynamic, large and disparate volume of data garnered from social media, Internet-driven technologies, financial records, and clinical research has arisen an increasing demand for reliable and scalable storage technologies. 
	Distributed storage systems are widely being used in modern data centers, such as Google File System \cite{ghemawat2003google}, Facebook Distributed File System \cite{sathiamoorthy2013xoring}, Microsoft Azure \cite{huang2012erasure} and also peer-to-peer storage settings, such as DHash++ \cite{dabek2004designing}, OceanStore \cite{rhea2001maintenance} and Total Recall \cite{bhagwan2004total}.
	In distributed storage systems individual storage nodes are  unreliable due to various hardware and software failures. Hence, redundancy is introduced to improve the system's reliability in the presence of node failures. The simplest  form of redundancy is the replication of the data in multiple storage nodes. Even though it is the most common form of redundancy, replication is very inefficient in terms of the offered reliability gain  per cost of the extra storage units required to store the redundancy. In this context, coding techniques have provably achieved orders of magnitude more reliability for the same redundancy compared to replication. 
	
	Besides the reliability offered by storing the redundant data, in order to be durable, it is necessary for a storage system to repair the failed nodes. The repair process consists of  downloading (part of) the content of a number of surviving nodes to reconstruct the missing content of the failed nodes. The conventional erasure codes suffer from high repair-bandwidth, the total size of data to be downloaded for the repair of each failed node. Regenerating code is a class of erasure codes which have gained popularity in this context, due to their low repair-bandwidth, while providing the same level of fault tolerance as erasure codes. 
	
	\subsection{Problem Formulation and System Model}
	In an $(n,k,d)$ regenerating code~\cite{dimakis2010network}, a file comprised of $F$ data symbols, each from a finite field $\mathbb{F}_q$, is encoded into $n$ pieces, and each piece will be stored in one storage node of capacity $\alpha$ symbols.
	The stored data in the nodes should maintain two main properties: 
	\begin{enumerate}
		\item  \textbf{Data Recovery:} By accessing any set of $k$ nodes, the data collector must be able to recover the original stored file.
		\item \textbf{Node Repair:} In the event of node failure, the content of the failed node can be regenerated by connecting to any subset of $\sen H$ nodes of size $\size {\sen{H}} = d$, and downloading $\beta$ symbols from each of the $d$ nodes. The set $\sen H$ is called the set of \emph{helper nodes}.
	\end{enumerate}
	
	In \cite{dimakis2010network} it is shown that there is a trade-off between the per-node storage capacity $\alpha$ and the per-node repair-bandwidth $\beta$ in a storage system that can guarantee the above main properties. While it is desired to minimize both $\alpha$ and $\beta$, one can be reduced only at the cost of increasing the other. 
	
	There are two types of node repairs: (i) \emph{functional-repair}, where a failed node will be replaced by a new node such that the resulting system continues to satisfy the data recovery and node repair properties.  An alternative to function repair is (ii) \emph{exact-repair}, under which the replaced node stores precisely the same content as the failed node. Hence, exact-repair is a more demanding criterion, and it is expected to require more repair bandwidth in comparison to functional repair, for a given storage size. However, from the practical stand, the exact repair is preferred, since it does not need the extra overhead of updating the meta-data (the relationship between the contents) and the system configuration. 
	
	\subsection{An Overview of the Related Works}
	The regenerating codes were introduced in the seminal work of Dimakis et al. \cite{dimakis2010network}, wherein $(n,k,d)$ distributed storage systems were studied using  \emph{information flow graphs}. Moreover, using the cut-set bound, it was shown that 
	the per-node storage capacity $\alpha$, the per-node repair bandwidth $\beta$, and the file size  $F$ should satisfy
	\begin{align}
	F \leq \sum_{i=1}^{k} \min(\alpha,(d-i+1)\beta),
	\label{eq:func:tradeoff}
	\end{align}
	for a storage system that maintains data recovery and node repair properties. This bound implies a trade-off between $\alpha$ and $\beta$ for a given $F$. This trade-off was shown to be achievable for the functional repair using random codes introduced in  \cite{ho2006random, wu2010existence}. 
	An important follow-up question was whether the same trade-off is achievable with the exact-repair property. First, in \cite{rashmi2011optimal} it was shown that exact-repair regenerating codes can be constructed for the extreme points of the trade-off, namely, the minimum bandwidth regeneration (MBR) referring to the minimum $\beta$ satisfying \eqref{eq:func:tradeoff}, and the minimum storage regeneration\footnote{The first MSR code construction in \cite{rashmi2011optimal} only holds for some range of system parameters.} (MSR), referring to the minimum $\alpha$ for which \eqref{eq:func:tradeoff} can be satisfied for a given $F$. Later, in \cite{shah2012distributed} it was shown that some of interior (between the two extreme) points of the trade-off are not achievable under the exact repair criterion. While the proof of \cite{shah2012distributed} did not rule out the possibility of approaching the non-achievable trade-off points with an arbitrary small gap, 
	the next question was whether there is a non-vanishing gap between the trade-off of exact-repair and functional-repair codes. This question was first answered in  \cite{tian2014characterizing}, where using a computer-aided approach of Yeung for information-theoretic inequalities \cite{yeung1997framework}, Tian completely characterized the trade-off for an $(n,k,d)=(4,3,3)$ system. Note that $(4,3,3)$ is the smallest system parameter for which there is a non-vanishing gap between the functional and the exact repair trade-off. 
	
	Thereafter, the attention of the data storage community has shifted to characterizing the optimum storage-bandwidth trade-off for the exact-repair regenerating codes. A trade-off characterization consists of two-fold: (i) designing code constructions satisfying data recovery and exact node repair properties, to achieve pairs of $(\alpha, \beta)$, and  (ii) proving information-theoretic arguments that provide lower bounds for the achievable pairs of $(\alpha, \beta)$.  The focus of this paper is on the code construction part, and hence, we provide a brief review of the existing code constructions in the literature. To this end, we divide the existing codes into three main categories based on the achievable trade-offs, as follows. 
	
	\noindent \textbf{1. The MBR point:} This point was fully solved for general $(n,k,d)$ in \cite{rashmi2011optimal}, where it was shown that the functional-repair trade-off is also achievable under the exact-repair criterion.  
	
	\noindent \textbf{2. The MSR point:} The trade-off offered by~\eqref{eq:func:tradeoff} for the MSR point is achievable by exact repair regenerating codes, and most of the existing code constructions are dedicated to this point. In \cite{cadambe2010distributed, suh2010existence}, it was shown that the exact-repair MSR code (for both low rate with $k/n \leq 1/2$ and high rate regimes where $k/n > 1/2$) is achievable in the asymptotic sense, that is when the file size is growing unboundedly. However, the proof was existential and no explicit code construction was provided.

	One of the earliest work in this area was  \cite{cullina2009searching}, where a computer search was carried out to find an $(n, k, d)=(5, 3, 4)$ MSR code.  In \cite{rashmi2011optimal} a code construction for parameters satisfying $n-1\geq d \geq 2k-2$ was presented. For the codes with $n-1>d>2k-2$ the code construction includes two steps: first a code is developed with $d'=2k'-2$, and then converted  to a code with desired $(k,d)$ parameters. Later, in \cite{lin2015unified} the code construction was unified for all parameters $n-1\geq d\geq 2k-2$.
	
	The explicit code construction for the codes with parameters $d<2k-2$ and $d \leq n-1$ was an open problem and several papers published to improve the state-of-the-art.  The code constructions in \cite{goparaju2017minimum,tamo2013zigzag,raviv2017constructions,wang2016explicit,li2015framework} were limited to the repair of only \emph{systematic} nodes. Another category of code constructions was dedicated to codes with a limited number of parity nodes. In particular, explicit MDS storage codes with only two parities ($n=k+2$) are constructed based on  Hadamard matrices in~\cite{papailiopoulos2013repair} and permutation matrices in~\cite{li2016optimal,wang2011codes}. Both constructions  offer an optimum  repair-bandwidth for the repair of any single (systematic or parity) node failure.

	Sub-packetization level, referring to the unnormalized value of $\alpha$ in terms of the number of symbols, is a practically important parameter of any code construction. While it is preferred to minimize the sub-packetization level, it cannot be lower than a certain lower bound provided in \cite{balaji2018tight}. 
	A class of MSR codes was introduced in \cite{sasidharan2015high} which requires only polynomial sub-packetization in $k$, but a very large field size. Nevertheless, the proposed construction of \cite{sasidharan2015high}  was not fully explicit, and it was limited to parameters satisfying $n=d-1$. The latter restriction was later relaxed in  \cite{rawat2016progress}, where the same result was shown for an arbitrary number of nodes, $n$. However, the code construction in \cite{rawat2016progress} still requires a large field size. 
	
	Several MSR code constructions for arbitrary parameters $(n,k,d)$ were recently proposed  \cite{ye2017explicitnearly,li2018generic, li2018alternative,sasidharan2016explicit}. They were all optimum in terms of the storage vs. repair-bandwidth trade-off, and all  achieve the optimum sub-packetization, i.e, matching the bound introduced in \cite{balaji2018tight}. The codes proposed  in \cite{ye2017explicit} offer a \emph{dynamic repair}, where the  number of helpers is flexible to be varied between $k$ and $n-1$.

	Our proposed codes in this work cover the entire trade-off, including the MSR point. For the resulting MSR codes from proposed construction, the code parameters and especially the  sub-packetization level do not depend on the total number of nodes $n$. Also, another advantage of this construction is its flexibility of system expansion by adding new parity nodes. In contrast to other existing codes in the literature, where the entire code needs to be redesigned in order to add new parity nodes, the proposed construction allows adding new parity nodes to the system, without changing the content of the other nodes. A comparison between the MSR code obtained from this construction and the existing codes, in terms of required field size and sub-packetization level, is presented in Table~\ref{tab:comparison}.
	
	\begin{table}
		\resizebox{\textwidth}{!}{
			\begin{tabular}{|l|l|l|l|}
				\hline
				&Sub-packetization level $\alpha$ & Field size $q$	 & Code Parameters \\ \hline \hline
				M.Ye et al.  \cite{ye2017explicit} Construction 1 & $(n-k) ^ {n}$ & $q \geq (n-k) n$ & $ n=d+1$\\ \hline
				M.Ye et al.  \cite{ye2017explicit} Construction 2 & $(n-k) ^ {n-1}$ & $q \geq n$ & $n=d+1$\\ \hline
				M.Ye et al.  \cite{ye2017explicit} Constructions 1,2 & $(d-k+1)^n$ & $q \geq s n$ for $s=\mathsf{lcm}(1,2,\cdots,n-k)$ & $d<n-1$\\ \hline
				M.Ye et al.   \cite{ye2017explicitnearly} & $(n-k) ^ {\ceil{\frac{n}{n-k}}}$ & $q \geq (n-k) \ceil{\frac{n}{n-k}}$ & $n=d+1$\\ \hline
				B. Sasidharan et al. \cite{sasidharan2016explicit}& $(n-k) ^ {\ceil{\frac{n}{n-k}}}$ & $q \geq (n-k) \ceil{\frac{n}{n-k}}$ &  $n=d+1$ \\ \hline
				Jie Li et al.  \cite{li2018generic}& $(n-k) ^ {\ceil{\frac{n}{n-k}}}$ & $q \geq n$ & $n=d+1$   \\  \hline
				Current paper & $(d-k+1)^k$ & $q \geq n$ & All \\ 
				\hline
			\end{tabular}}
			\vspace{2mm}	
			\caption{Comparison between code constructions proposed for the MSR point.}
			\label{tab:comparison}
	\end{table}
	
	\noindent \textbf{3. Interior points:} The construction for the interior points (trade-off points except for MBR and MSR) was restricted to the specific system parameters. In \cite{tian2015layered} a code construction for $(n=d+1, k=d, d)$ was presented. The trade-off achieved by this construction was shown to be optimum under the assumption of the linearity of the code  \cite{elyasi2015linear, prakash2015storage, duursma2015shortened}. However, it wasn't clear that if the same trade-off is achievable for $n>d+1$. Most of the follow-up efforts to increase the number of parity nodes resulted in compromising the system capacity to construct a code for larger values of the $n$, and hence their trade-off was diverging from the lower bound of \cite{elyasi2015linear, prakash2015storage, duursma2015shortened}, and $n$ increases. The first $n$-independent achievable trade-off for interior points was provided in \cite{elyasi2015probabilistic}, where it was shown that the first corner point on the trade-off next to the MSR point can be achieved for any $(n, k=d, d)$ system. However, the proof of \cite{elyasi2015probabilistic} was just an existence proof, where a random ensemble of codes were introduced, and it was shown that for any $n$ and large enough field size, there exists at least one code in the ensemble that satisfies both data recovery and node repair properties. 
	
	In \cite{elyasi2016determinant} the above-mentioned restriction for $n$ was uplifted, where an \emph{explicit} code construction for the entire trade-off of an $(n,k=d,d)$ storage system was introduced. The proposed \emph{determinant codes} are optimal subject to the linearity of the code and achieve the lower bound in \cite{elyasi2015linear, prakash2015storage, duursma2015shortened}, regardless of the total number of nodes. However, the repair process of the determinant codes in \cite{elyasi2016determinant} requires heavy computation, and more importantly, the repair data sent from a helper node to a failed node depends on the identity of all the helpers participating in the repair process. Later in  \cite{elyasi2018newndd} this issue was resolved by introducing a new repair mechanism for determinant codes.
	
	The next set of works focused on breaking the last constraint, i.e., $k=d$. An explicit code construction was introduced in \cite{tian2015layered} for an $(n,k\leq d,d)$ system. The resulting trade-off was improved by the code construction of \cite{goparaju2014new}. 
	In \cite{senthoor2015improved} a class of improved layered codes was introduced. However, it turned out that the trade-off achieved by \cite{senthoor2015improved} is optimum only for the corner point next to the MBR, and only for an $(n=d+1, k , d)$ system. These limitations were slightly relaxed in	\cite{elyasi2017scalable}, where the constraint of $n=d+1$ was lifted. However, the construction in \cite{elyasi2017scalable} was dedicated to the trade-off point next to the MBR, implying a low repair bandwidth. Later in \cite{elyasi2017exact}, this construction was extended for the entire trade-off but only for an $(n,k=d-1, d)$ system. In this paper, we propose a code construction for general $(n,k,d)$ parameters for the entire trade-off. We conjecture that the resulting trade-off is optimum. 
	\subsection{Notation and Terminologies}
	\label{subsec:notation}
	We will use lowercase letters to denote scalars (e.g., integers $k$ and $d$) or data file symbols (e.g. $v$ and $w$ that take value from a finite field). We use bold capital letters to denote matrices. For positive integers $a$ and $b$, we denote set $\{a,a+1,\dots, b-1,b\}$ by $\intv{a:b}$, and set $\{1,2,\dots, b\}$ by $\intv{b}$. Also,  for $a>b$ we have $\intv{a:b}=\varnothing$. Calligraphic  letters (e.g., $\sen I$ and $\sen J$) are used to denote a set of integer numbers. Moreover, $x\in \sen I$ indicates that the scalar $x$ belongs to the set $\sen I$. The largest entry of a set $\sen I$ is denoted by $\max \sen I$, and  the maximum of an empty set is defined as $\max \varnothing = -\infty$, for consistency. In this work, $\size{\sen{I}}$ denotes the size of the set $\sen I$. For two sets $\sen I$ and $\sen J$ with $\size{\sen I} = \size{\sen J}$, we write $\sen I \prec \sen J$ to indicate the lexicographical order between $\sen I$ and $\sen J$, e.g., $\{1,2,3\} \prec \set{1,2,4} \prec \set{1,2,5} \prec \set{1,3,4} \prec \set{1,3,5}$. 
	For an integer $x$ and a set $\sen I$, we define
	\begin{align}
	\ind{\sen I}{x} = \left|\{y\in \sen I: y\leq x \}\right|.
	\label{eq:def:ind}
	\end{align}
	For a matrix $\f{P}$, we may label its rows and columns by integers (e.g. $i$ and $j$) or by sets (e.g. $\sen I$ and $\sen J$), which will be specified accordingly. Then $\f{P}_{i,j}$ (or $\f{P}_{\sen I, \sen J}$) refers to an entry of matrix $\f{P}$ at the row indexed by $i$ (or $\sen I$) and the column labeled by $j$ (or $\sen J$). Moreover, we  denote the  $i$-th row of $\f{P}$ by $\f{P}_{i,:}$. We also use $\f{P}_{:,j}$ to refer to the $j$-th column of $\f{P}$. Lastly, we may use $\f{P}[\sen X, \sen Y]$ to refer to a submatrix of $\f{P}$ obtained by a family (with family elements being integer numbers or sets) of rows $\sen X$ and a family of columns $\sen Y$. Accordingly, $\f{P}[\sen X, :]$ is a submatrix of $\f{P}$ formed by stacking all rows with labels belonging to $\sen X$.
	
	Note that $k$ and $d$ (with $k\leq d$) are the main system parameters throughout the paper. For a matrix $\f{P}$  with $d$ rows, we define $\up{\f{P}}=\f{P}[\set{1,2,\cdots,k},:]$ and $\down{\f{P}}=\f{P}[\set{k+1,\cdots,d},:]$ to be sub-matrices of $\PM$ obtained from stacking the top $k$ and the bottom $(d-k)$ rows of a matrix $\f{P}$, respectively.
	Throughout the analysis,  we frequently need to concatenate several matrices, that is, merging a number of matrices with the same number of rows side-by-side, to form a fat matrix with the same number of rows. 
	
	Finally, the binomial coefficient is defined as $\binom{\ell}{m}=\frac{\ell!}{m!(\ell-m)!}$ and we set it to zero for $m<0$ and $m>\ell$, for the sake of consistency.
	\subsection{Paper Organization} 
		The rest of this paper is organized as follows. We first present the main result of this work, which is the trade-off achieved by the proposed codes, in Section~\ref{sec:main}. Appendix~\ref{app:proof:cor} includes the proof to show that the parameters of this code can achieve MBR, MSR, and another point on the cut-set bound. Then in Section~\ref{sec:ndd:review}, we review the (signed) determinant codes~\cite{elyasi2016determinant, elyasi2018newndd}, their construction and main properties. The proof of node repairability for singed determinant codes is given in Appendix~\ref{app:prf:prop:ndd:repair} for the sake of completeness. The core idea of our code construction, which is the multiplication of a fixed encoder matrix to a cascade message matrix, is presented in Section~\ref{sec:code}. Appendix~\ref{app:semi-sys} discusses the conversion of the proposed code with an arbitrary encoder matrix to a (semi)-systematic code. Section~\ref{sec:supermessage}, explains the details of cascading message matrices and includes a running example to facilitate understanding the notation and the concept of injection. The main properties of node exact repair and data recovery are proved in Sections~\ref{sec:noderepair} and~\ref{sec:datarec}, respectively.
		The parameters of the proposed code construction are evaluated in Section~\ref{sec:parameters}. An explicit evaluation of the parameters requires solving an implicit recursive equation, which is performed using the $\cZ$-transform, as discussed in Appendix~\ref{app:Z}. 
		Finally, the paper is concluded in Section~\ref{sec:conclusion}, with a comparison between the proposed code and the product-matrix code, a discussion on our conjecture regarding the optimality of the proposed codes, and a number of related open questions for future works.
	
	\section{Main Result}
	\label{sec:main}
	The main contribution of this paper is a novel construction for exact-repair regenerating codes, with arbitrary parameters $(n,k,d)$. The following theorem characterizes the achievable storage vs. repair-bandwidth trade-off of the proposed code construction. 
	\begin{thm}
		\label{thm:main}
		For a distributed storage system with parameters $(n,k,d)$ satisfying $k \leq d < n$, the triple $(\alpha,\beta,F)$ defined as
		\begin{align}
		\begin{split}
		\alpha(k,d;\mu) &= \sum_{m=0} ^{\mu} (d-k)^{\mu-m} \binom{k}{m}\\
		\beta(k,d;\mu) &= \sum_{m=0} ^{\mu} (d-k)^{\mu-m} \binom{k-1}{m-1}\\
		F(k,d;\mu) &= \sum_{m=0} ^{\mu} k(d-k)^{\mu-m} \binom{k}{m}-\binom{k}{\mu+1}\\
		\end{split}
		\label{eq:params}
		\end{align}
		can be achieved by the cascade codes proposed in this paper for $\mu\in\{1,2,\dots, k\}$.
	\end{thm}
	The trade-off achieved by this construction is shown in Fig.~\ref{fig:tradeoff}, and is compared against that of other existing codes. 
	\begin{figure}[!t]
    \centering
    \begin{tikzpicture}[scale=0.40]
    \node[anchor=north east,inner sep=0] (image) at (0,0) {\includegraphics[width=0.5\linewidth]{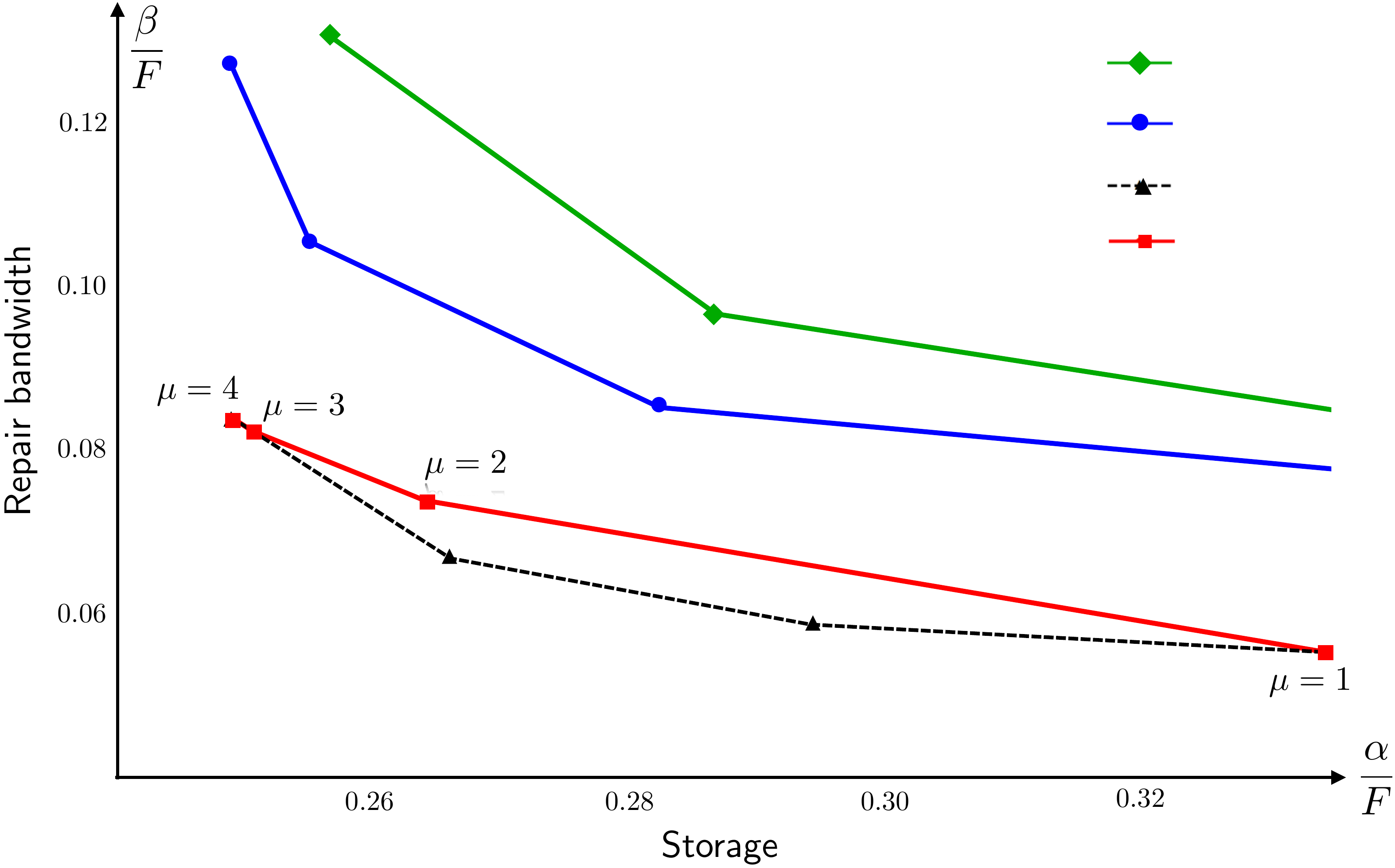}};
    \begin{scope}[x={(image.north west)},y={(image.south east)}]
    \node[fill = white, outer sep=0pt, inner sep=0pt,align = left] (rect) at (0.35,0.2) { 
        {\tiny Code construction in \cite{tian2015layered}}  \\[-1.5pt]
        {\tiny Code construction in \cite{senthoor2015improved}} \\[-1.3pt]
        {\tiny Funcational repair trade-off \cite{dimakis2010network}} \\[-1.3pt]
        {\tiny The proposed code} \\[-4.3pt]
    };
    \end{scope}
    \end{tikzpicture}
  		\caption{The achievable trade-off of the proposed code for $(n=8,k=4,d=6)$, and its comparison to other known codes with these parameters. The proposed trade-off has $k=4$ corner points, enumerated with $\mu=1,2,3,4$. }
    \label{fig:tradeoff}	
    \end{figure}

	Note that for a given system parameters $(n,k,d)$, a total of $k$ triples of $(\alpha, \beta, F)$  can be obtained from \eqref{eq:params}, by varying the parameter $\mu$, which we refer to as the \emph{mode} of the code. Therefore, while parameters $(n,k,d)$ are the system parameters, we append the parameter $\mu$ to this notation as $(n,k,d;\mu)$, to refer to the specific point on the trade-off corresponding to $\mu$.
	The following corollary shows that MBR and MSR points are subsumed as special cases of the proposed trade-off. The proof of the following corollaries are straightforward algebraic manipulations and given in Appendix~\ref{app:proof:cor}.
	\begin{cor}
		The cascade code with mode $\mu=1$ is an MBR code with parameters 
		\begin{align*}
		(\alpha_{\mathsf{MBR}}, \beta_{\mathsf{MBR}}, F_{\mathsf{MBR}}) = \left(d,1, \frac{k(2d-k+1)}{2}\right).
		\end{align*}
		Similarly, the cascade code with mode $\mu=k$ is an MSR code with parameters
		\begin{align*}
		(\alpha_{\mathsf{MSR}},	&\beta_{\mathsf{MSR}}, F_{\mathsf{MSR}}) = \left( (d-k+1)^k, (d-k+1)^{k-1}, k(d-k+1)^k\right).
		\end{align*}
		\label{cor:MSR} 
	\end{cor}
	\begin{cor}
		The proposed code at mode $\mu=k-1$ achieves the cut-set bound, and hence it is optimum.  
		\label{cor:mu=k-1}  
	\end{cor}
Fig.~\ref{fig:compare} compares the \emph{scalability} aspect of the code construction of this paper with an existing  construction. The term scalability refers to the property that the number of nodes in a distributed storage system can be increased (for a sufficiently large field size) without compromising the system performance, and its overall capacity.

The rest of this paper is dedicated to a comprehensive proof of Theorem~\ref{thm:main}. After reviewing the (signed) determinant codes~\cite{elyasi2016determinant, elyasi2018newndd} in Section~\ref{sec:ndd:review}, we introduce the cascade codes and present their construction in Sections~\ref{sec:code} and~\ref{sec:supermessage}. Then, we prove that the proposed codes maintain the exact-repair property (Proposition~\ref{prop:nkd:repair}) and data recovery property (Proposition~\ref{prop:nkd:recovery}), in Sections~\ref{sec:noderepair} and~\ref{sec:datarec}, respectively. Finally, in Section~\ref{sec:parameters} we show that the parameters of the proposed code match those claimed in Theorem~\ref{thm:main}. 
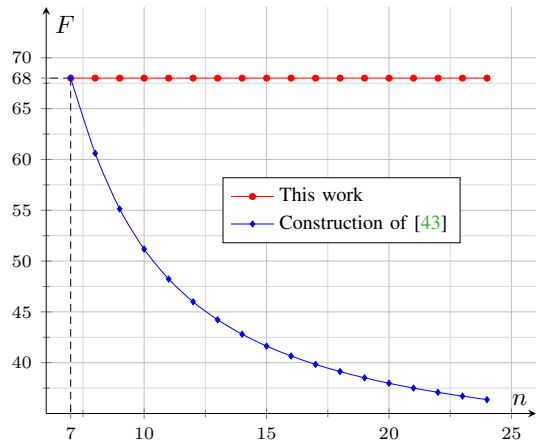
\begin{figure}
\begin{center}
	\scalebox{0.95}{
\begin{tikzpicture}
\begin{axis}[
legend style={at={(0.36,0.5)},anchor=west},
xmin=6,xmax=26,
ymin=35,ymax=75,
grid=both,
grid style={line width=.1pt, draw=gray!30},
major grid style={line width=.2pt,draw=gray!50},
axis lines=middle,
minor tick num=1,
mark options={scale=0.6},
legend cell align={left},
legend style={font=\footnotesize},
ticklabel style={font=\scriptsize,fill=white},
xlabel=$n$,
ylabel=$F$,
extra tick style={
	xmajorgrids=false,
	ymajorgrids=false,
},
extra x ticks={7},
extra y ticks={68},
ytick = {35,40, ..., 70},
xtick = {5,10, ..., 30},
]

\addplot[smooth,color=red,mark=*]
plot coordinates {
(7,68)
(8,68)
(9,68)
(10,68)
(11,68)
(12,68)
(13,68)
(14,68)
(15,68)
(16,68)
(17,68)
(18,68)
(19,68)
(20,68)
(21,68)
(22,68)
(23,68)
(24,68)
};
\addlegendentry{This work}

\addplot[smooth,color=blue,mark=diamond*]
plot coordinates {
(7,68)
(8,13152/217)
(9,20067/364)
(10,1791/35)
(11,8393/174)
(12,81448/1771)
(13,26269/594)
(14,863024/20163)
(15,18944/455)
(16,1668752/41041)
(17,310077/7784)
(18,112413/2873)
(19,531647/13804)
(20,35512072/935085)
(21,50593/1349)
(22,3269992/88179)
(23,135671/3696)
(24,37937472/1043119)
};
\addlegendentry{Construction of \cite{senthoor2015improved}}
\draw [densely dashed] (axis cs:7,68) -- (axis cs:7,0)  ;
\draw [densely dashed] (axis cs:7,68) -- (axis cs:0,68)  ;
\end{axis}
\end{tikzpicture}
}
\end{center}	
\caption{Comparison of the maximum files size ($F$) of  two $(n \geq 7,k=4,d=6)$ exact-regenerating codes, and code parameters $(\alpha,\beta)= (18,5)$. When the distributed storage system has only $n=d+1=7$ nodes, both codes can store $F=68$ units of data. However, for a sufficiently large field size, the storage capacity decays as a function of $n$ for code introduced in \cite{senthoor2015improved}, while the storage capacity is preserved for the cascade code introduced in this paper. Note that in the cascade code construction, the pair $(\alpha,\beta)= (18,5)$ corresponds to a code construction with $\mu=2$.}
\label{fig:compare}
\end{figure}

	\section{A Review of $(n,k=d,d;m)$ Signed Determinant Codes}
	\label{sec:ndd:review}
	The code construction presented in this paper uses \emph{signed determinant codes}  as the main building blocks. The family of determinant codes, first introduced in \cite{elyasi2016determinant, elyasi2018newndd}
	is a class of codes that achieve the optimum (linear) storage-bandwidth trade-off of the regenerating codes \cite{elyasi2015linear, prakash2015storage, duursma2015shortened}, when the number nodes participating in data recovery equals the number of helpers contributing in a repair process, i.e., $k=d$. They are linear codes, which can be constructed by multiplying an encoder matrix by a message matrix.

	The modification here (that converts a determinant code to a signed determinant code) is due to an arbitrary assignment of ($+/-$) signs to the rows of the message matrix, which affect all the entries in the corresponding row. This modification will be applied by a signature vector defined in Definition~\ref{defi:sign:vect}. As we will see in this section, the above modification preserves all properties of determinant codes, while it is helpful towards our next step, which is the construction of $(n,k,d;m)$ codes. 
	
	In the following, we review the construction of signed determinant codes and their properties.     First, note that the mode for the $d=k$ signed determinant codes is defined in the same way that is defined for $k<d$ cascade code defined in Section~\ref{sec:main}. This means that the family of signed determinant codes for an $(n,k=d,d)$ system consists of $k$ distinct codes, enumerated by a parameter $m\in \set{1,2,\dots, k}$, which is called \emph{mode} of the code.\footnote{The reason that we use different letters $m$ and $\mu$ for the mode of the  cascade codes and determinant codes is that in the construction of $(n,k,d;\mu)$ cascade codes of this paper, several $(n,k=d,d;m)$ signed determinant code with different values of $m$ is used.}
	
	For any mode $m\in\intv{k}$, the parameters of the determinant code corresponding to the $m$-th corner point on the trade-off are given by
	\begin{align*}
	\left(\alpha_{m} ,\beta_{m} ,F_{m} \right) =\left(\binom{d}{m},\binom{d-1}{m-1},m\binom{d+1}{m+1} \right).
	\end{align*}
	Here $m=1$ corresponds to MBR code, while $m=k$ results in the parameters of an MSR code. We also define the operator $\md{\e{D}}$ which returns parameter $m$ (the mode) of the determinant code $\e{D}$.
	
	\subsection{Code Construction of $(n,k=d,d;m)$ codes}
	A signed determinant code with parameters $(n,k=d,d;m)$ is represented by a matrix $\bC_{n \times \alpha_m}$ whose  $i$-th row includes the coded content of the $i$-th node. In general, $\bC_{n \times \alpha_m}$ is obtained by multiplying an \emph{encoder matrix} $\enc_{n \times d}$ whose entries are from a finite field $\mathbb{F}_q$ and a \emph{message matrix} $\e{D}_{d \times \alpha_m}$. The encoder matrix $\enc$ is chosen such that any collection of $d$ of its rows are linearly independent. Examples of such matrices include Vandermonde \cite{horn1990matrix} and Cauchy matrices \cite{schechter1959inversion}. The message matrix $\e{D}$ has $d$ rows and $\alpha_m=\binom{d}{m}$ columns, and its entries are filled with the symbols from the file to be stored in the storage system. To do this, we first split $F_m=m\binom{d+1}{m+1}$ raw file symbols into two groups, namely,  $\sen V$ and $\sen W$, and label them as
	\begin{align}
	\begin{split}
	\sen V  &= \left\{v_{x,\sen X}:{\sen X} \subseteq \intv d,\size {\sen  X} = m , x\in \sen X \right\},\\
	\sen W &=\left\{w_{x,\sen Y}:{\sen Y} \subseteq \intv d,\size {\sen  Y} =m+1, x\in {\sen Y}\setminus \set{\max \sen Y} \right\}.
	\end{split}
	\label{eq:vw-symbols}
	\end{align}
	Note that the symbols of $\sen V$ are indexed by a set ${\sen  X} \subseteq \intv d$ of size $m$ and an element $x\in {\sen X}$, implying $\size{\sen V}=m\binom{d}{m}$. Similarly, the symbols in $\sen W$ are indexed by a pair $(x,\sen Y)$, where $x$ can be any element of $\sen Y$, except the largest element. Hence, there are $\size {\sen W} = m\binom{d}{m+1}$ symbols in set $\sen W$. Note that 
	\begin{align*}
	\size{\sen V} + \size{\sen W} = m\binom{d}{m} + m\binom{d}{m+1} = m\binom{d+1}{m+1} = F_m.
	\end{align*}
	In the definition of $\sen W$ symbols in~\eqref{eq:vw-symbols}, the symbol corresponding to $x = \max \sen Y$ was excluded. In the following we define this symbol as the parity symbol.	
	\begin{defi}
		\label{def:parity}
		 Let $\sen Y \subseteq \intv{d}, \size{\sen Y}=m+1$, we define a parity symbol $w_{\max \sen Y,\sen Y}$, such that the parity equation
		\begin{align}
		\sum_{y \in \sen Y} (-1)^{\ind{\sen Y}{y}} w_{y, \sen Y} = 0
		\label{eq:parityeq}
		\end{align}
		is satisfied.\footnote{This parity equation is the same as original determinant codes introduces in \cite{elyasi2016determinant, elyasi2018newndd}.} We refer to the symbols in the set $\set{w_{y,\sen Y}\!:\hspace{-1pt}{\sen Y} \subseteq \intv d,\size {\sen  Y}  \hspace{-1pt}=m+1, y\hspace{-1pt}\in\hspace{-1pt} {\sen Y}}$ as the \emph{$w$-group} of~$\sen Y$. 
	\end{defi}
	
	\begin{defi}
		\label{defi:sign:vect}
		For a signed determinant codes, a vector $\sigma_{\e{D}}$, called \emph{signature} vector, is given as the input of the code construction.\footnote{In the original determinant codes introduced in \cite{elyasi2016determinant, elyasi2018newndd}  $\sigma_{\e{D}}$ was fixed to $\s{\e{D}}{x}=0$ for all $x\in \intv{d}$.} This vector is of length $d$ with integer entries. The entry at position $x\in \intv{d}$ of this vector is denoted by $\sigma_{\e{D}}(x)$. Using this signature vector a plus or minus sign will be assigned to each integer $x\in \intv{d}$, that is $(-1)^{\s{\e{D}}{x}}$. 
		\end{defi}
		\begin{rem}
		\label{rem:signature}
		Any choice of the signature vector yields a valid exact regenerating code satisfying data recovery and exact repair properties. In this paper, whenever we generate an instance of a signed determinant code, we specify the signature vector of that code first. \hfill $\diamond$
		\end{rem}

	To fill the message matrix $\e{D}$, we label its rows by integer numbers from $\intv{d}$ and its columns by subsets $\sen I \subseteq \intv{d}$ with $\size{\sen I} =m$, according to the lexicographical order. Then the entry at row $x$ and column $\sen I$ is given by 
	\begin{align}
	\e{D}_{x,\sen I}=\left\{
	\begin{array}{l l}      
	(-1)^{\s{\e{D}}{x}} v_{x,\sen I} & \textrm{if $x \in \sen I$}, \\
	(-1)^{\s{\e{D}}{x}} w_{x,\sen{I}\cup \set {x}} & \textrm{if $x \notin \sen I$}.
	\end{array}\right.
	\label{eq:def:S}
	\end{align}
	Once the message matrix is formed, the coded content of node $i$ is given by $\enc_{i,:} \cdot \e{D}$.
	For the sake of completeness, we define an $(n,k=d,d;m=0)$ determinant code at mode $m=0$ to be a trivial code with $(\alpha=1,\beta=0,F=0)$, whose message matrix is a $d\times 1$ all-zero matrix.

	In \cite[Section~IV]{elyasi2018newndd}, an illustrative example for a determinant code with parameters
		$(n,k,d;m)=(8,4,4;2)$ is given. The corresponding signature of this code is $(0,0,0,0)$; however, each row of the message matrix can be multiplied by an arbitrary sign to change the signature vector of the code. This example also explains the main idea of the code construction for determinant codes, as well as data recovery and node repair properties of the code.
	
	Additionally, some instances of message matrices of $(n,d,d;m)$ determinant codes are provided in this paper. For instance, the matrix in~\eqref{eq:D0} of Example~\ref{ex:root}, is an instance of the message matrix of determinant code for $d=6,m=4$ with an all-zero signature vector. In this example, there is a separator line between rows $4$ and $5$, and also some symbols are highlighted and marked by frames, which will be explained later. Regardless, if these details are ignored, the underlying matrix is an $(n,d,d;m)=(n,6,6;4)$ message matrix. Also, two matrices are given in~\eqref{eq:D2} and~\eqref{eq:D5} in Example~\ref{ex:modified:child}, which are message matrices for an $(n,k,d)=(n,6,6)$ determinant code of modes $m=2$ and $m=1$ respectively. The signature vector of these codes are $\sigma_{\e{T}_2}=(2,2,2,2,2,2)$ and $\sigma_{\e{T}_5}=(2,2,2,2,2,3)$.
	
		\begin{rem}
			\label{rem:raw:parity}
			From the definition of message matrix in~\eqref{eq:def:S}, the entries $\e{D}_{x,\sen I}$ with $x > \max \sen I$ are parity symbols and entries with $x \leq \max \sen I$ are raw data file symbols. This is because if $x > \max \sen I$, then $x \notin \sen I$ and thus $\e{D}_{x,\sen I}=(-1)^{\s{\e{D}}{x}} w_{x,\sen{I}\cup \set {x}}$. Also, since $x$ is the largest element in the $w$-group of $\sen Y = \sen I \cup \set{x}$, $w_{x,\sen{I}\cup \set {x}}=w_{\max \sen Y,\sen Y}$ will correspond to a parity symbol as defined in~\eqref{eq:parityeq}. Similarly, one can argue that for $x \leq \max \sen I$, the symbol at position $(x, \sen I)$ of $\e{D}$ is either a $v$-symbol or a non-parity $w$-symbol.
	\end{rem}
	
	\begin{rem}
		The construction of signed determinant codes  \emph{decouples} the parameters of the code.  The encoder matrix only depends on $n$ and $d$ and remains the same for all modes. On the other hand, the message matrix $\e{D}$ is fully determined by parameters $(d,m)$, and does not depend on $n$, the total number of nodes in the system. Thus, we refer to the code defined above as a \emph{$(d;m)$ signed determinant code} and to matrix $\e{D}$ as a \emph{$(d;m)$ message matrix}. 
		\label{rem:decouple}
	\end{rem}
	Next, we review data recovery and exact repair properties for signed determinant codes.
	\subsection{Data Recovery of $(n,d=k,d;m)$ Determinant Codes}
	\begin{prop}
		In a $(d;m)$ signed determinant code, all the data symbols can be recovered from the content of any $k=d$ nodes. 
		\label{prop:ndd:recovery}
	\end{prop}
	The proof of this proposition is similar to that of Proposition~1 in \cite{elyasi2018newndd}, and hence omitted for the sake of brevity. 
	
	\subsection{Node Repair of $(n,k=d,d;m)$ Determinant Codes}
	Consider the repair process of a failed node $f\in\intv{n}$ from an arbitrary set $\mathcal{H}$ of $\size{\sen H} = d$ helpers. The repair-encoder matrix for a $(d;m)$ signed determinant code is defined below. 
	\begin{defi}
	\label{defi:repair-encoder}
		For a $(d;m)$ signed determinant code with signature $\sigma_{\e{D}}$, and a failed node $f\in \intv{n}$, the \emph{repair-encoder matrix} $\repMat{f}{m}$ is defined as a $\binom{d}{m} \times \binom{d}{m-1}$ matrix, whose rows are labeled by $m$-element subsets of $\intv{d}$ and columns are labeled by $(m-1)$-element subsets of $\intv{d}$. The element in  row $\sen I$ and column $\sen J$ of this matrix is given by
		\begin{align}
		\repMat{f}{m}_{\sen I, \sen J}=\left\{
		\begin{array}{l l}      
		(-1)^{\s{\e{D}}{y}+\ind{\sen I}{y}}\psi_{f,y} & \textrm{if $\sen I \setminus \sen J=\set{y}$},\\
		0& \textrm{otherwise}.
		\end{array}\right.
		\label{eq:rep:enc}
		\end{align}
		where $\psi_{f,y}$ is the entry of the encoder matrix $\enc$ at row $f$ and column $y$. Also, $\s{\e{D}}{y}$'s are the same signature values used in \eqref{eq:def:S}. Note that these repair-encoder matrices can be fully generated from the encoder matrix, and do not depend on the node contents.  
	\end{defi}
	
	In order to repair node $f$, each helper node $h\in \sen H$  multiplies its content $\enc_{h,:}  \cdot \e{D}$ by the repair-encoder matrix of node $f$ to obtain $\enc_{h,:} \cdot \e{D} \cdot \repMat{f}{m}$, and sends the result to node $f$. The required repair-bandwidth of this repair scheme is given in the following proposition. 
		\begin{prop}
			The matrix $\repMat{f}{m}$ defined in \eqref{eq:rep:enc} has rank $\beta_{m} =\binom{d-1}{m-1}$. Therefore, even though the number of entries in vector $\enc_{h,:}\cdot \e{D} \cdot \repMat{f}{m}$ is  $\binom{d}{m-1}$, it can be fully delivered to the failed node  by communicating only $\beta_{m} = \binom{d-1}{m-1}$ symbols in $\mathbb{F}_q$ from node $h$ to node $f$. 
			\label{prop:beta}
		\end{prop}
		We refer to \cite[Proposition~3]{elyasi2018newndd}  for the proof of Proposition~\ref{prop:beta}.
	
	\noindent Upon receiving $d$ vectors $\{\enc_{h,:} \cdot \e{D} \cdot \repMat{f}{m}: h\in \sen H\}$ and stacking them into a matrix, the failed node obtains \[\enc[\cH,:] \cdot \e{D} \cdot \repMat{f}{m},\] where $\enc[\cH,:]$ is a $d\times d$ full-rank sub-matrix of $\enc$, obtained from stacking rows indexed by $h\in \sen H$. Multiplying both sides by $\enc[\cH,:]^{-1}$, the failed node retrieves
	\begin{align}
		\erep{f}{\e{D}} = \e{D} \cdot \repMat{f}{m},
		\label{eq:eRepSpace}
    \end{align}
	which is called the \emph{repair space of node $f$}. One instance of the repair space for a $(d;m)=(6,4)$ code is presented in~\eqref{eq:RT0} in Example~\ref{ex:repair}. Recall that the content of node $f$ is represented by the row vector $\enc_{f,:} \cdot \e{D}$ which its entries have the same labeling as columns of $\e{D}$, i.e $m$ element subsets of $\intv{d}$. All the coded symbols in node $f$ can be recovered from its repair space as described in the following proposition. 
	\begin{prop}
		The coded symbol at index $\sen I$ of node $f$ can be recovered from $\erep{f}{\e{D}}$ defined in \eqref{eq:eRepSpace} using
		\begin{align}
		\left[\enc_{f,:} \cdot \e{D}\right]_{\sen I} = \sum_{i\in \sen I} (-1)^{\s{\e{D}}{i}+\ind{\sen I}{i}} \left[\erep{f}{\e{D}}\right]_{i,\sen I \setminus \seq{i}}.
		\label{eq:ndd:repair}
		\end{align}
		\label{prop:ndd:repair}
	\end{prop}
	This proposition is very similar to Proposition~2 in \cite{elyasi2018newndd}. However, due to the modification introduced by the signature vector here, we present the proof of the current proposition in Appendix~\ref{app:prf:prop:ndd:repair} for the sake of completeness. 
	\begin{rem}
		Note that a signed determinant code can be defined over any Galois field. In particular, for a code designed over $\mathsf{GF}(2^s)$ with characteristic $2$, we have $-1=+1$, and hence, all the positive and negative signs disappear. Especially, the signs in~\eqref{eq:def:S} can be removed and the parity equation in  \eqref{eq:parityeq} will simply reduce to $\sum_{y \in \sen I} w_{y, \sen I} = 0$. Also, the non-zero entries of the repair encoder matrix in \eqref{eq:rep:enc} will be $\psi_{f,x}$, and the repair equation in \eqref{eq:ndd:repair} will be simplified to $\left[\enc_{f,:} \cdot \e{D} \right]_{\sen I} = \sum_{x\in \sen I} \left[\rep{f}{\e{D}}\right]_{x,\sen I \setminus \seq{x}}$.
	\end{rem}
	
	\section{$(n,k,d;\mu)$ Cascade Code Construction}
	\label{sec:code}
		In this section, we describe the construction of the cascade codes. For a fixed set of system parameters $(n,k,d;\mu)$, the code parameters $(\alpha,\beta,F)$ of this construction are given in~\eqref{eq:params}. Similar to the construction of $(n,k=d,d;m)$ determinant codes, an $(n,k,d;\mu)$ cascade code is also presented by a matrix $\bC_{n \times \alpha}$ whose  $i$-th row includes the coded content of the $i$-th node. This $\bC_{n \times \alpha}$ matrix is obtained by multiplying an \emph{encoder matrix} $\enc_{n \times d}$ by a \emph{super message matrix} $\bM_{d \times \alpha}$. We first explain the construction of the encoder matrix and then give an overview of the cascade structure of the super message matrix. The details of $(n,k,d)$ super message matrix are discussed in Section~\ref{sec:supermessage}. A summary of the symbols and notations used throughout this construction is given in Table~\ref{table:notations}, as a reference to facilitate following the code construction. 
	\begin{table*}
	\begin{center}
    \resizebox{\textwidth}{!}{\begin{tabular}{|c||l|c|}
      \hline
     Symbol & Definition & Reference \\
     \hline 
     \hline
     $\ind{\sen I}{x}$ & number of elements in set $\sen I$ not exceeding $x$ & \eqref{eq:def:ind}\\
     \hline
     $\e{D}$ & message matrix of a determinant code  & \eqref{eq:def:S}\\
     \hline
     $\md{\e{D}}$ & mode of the determinant code with message matrix $\e{D}$ & \\
     \hline 
     $\up{\e{D}}$ & submatrix of $\e{D}$ consisting of its top $k$ rows & \\ 
     \hline
     $\down{\e{D}}$ & submatrix of $\e{D}$ consisting of its bottom $(d-k)$ rows & \\ 
     \hline
     $\sigma_{\e{D}}$ & signature vector of signed determinant code & Definition~\ref{defi:sign:vect}\\
     \hline
     $\repMat{f}{m}$ & repair encoder matrix for failed node $f$ & Definition~\ref{defi:repair-encoder}\\
     \hline
     $\erep{f}{\e{D}}$ & repair space of node $f$ & \eqref{eq:eRepSpace}\\
     \hline
     $\enc$ & encoder matrix & Definition~\ref{def:encoder}\\
     \hline
     $\dgp{P}, \ngp{P}, \pgp{P}$ & classification of symbols in $\down{\e{P}}$ & Definition~\ref{def:groups}\\
     \hline
     $(x,\sen B)$ & injection pair & Remark~\ref{rem:injpair}\\
     \hline
     $\eSMat{Q}{m}{x,\sen B}{\e{P}}$ & child matrix of parent matrix $\e{P}$ associate with injection pair $(x,\sen B)$ & \eqref{eq:mode:relation} and \eqref{eq:inj:sign}\\
     \hline
     $\tOMat{m}{x,\sen B}{\e{P}}$ & injection matrix from parent matrix $\e{P}$ to a child matrix with injection pair $(x,\sen B)$ & \eqref{eq:inj:mat}\\
     \hline
     $\fSMat{Q}{m}{x,\sen B}{\e{P}}$ & child code $\eSMat{Q}{m}{x,\sen B}{\e{P}}$ after injection & \eqref{eq:Qdef}\\
     \hline
     $\MM$ & super message matrix of cascade code  & \\
     \hline
    \end{tabular}}
    \end{center}
    \caption{A summary of the symbols frequently used in the code construction.}
    \label{table:notations}
    \end{table*}
	\subsection{The Encoder Matrix}
	\begin{defi}
		\label{def:encoder}
		The encoder matrix $\enc$ for a code with parameters $(n,k,d)$ is defined as an $n\times d$ matrix 
		\[\enc_{n\times d} = \left[\mathbf{\Gamma}_{n\times k} | \mathbf{\Upsilon}_{n\times (d-k)} \right],
		\]
		such that
		\begin{enumerate}[label=\bf{(E\arabic*)}, ref=\bf{(E\arabic*)}] 
		    \item any $k$ rows of $\mathbf{\Gamma}$ are linearly independent; and \label{cond:G}
		    \item any $d$ rows of $\enc$ are linearly independent. \label{cond:P}
		\end{enumerate}
	\end{defi}
	Note that Vandermonde matrices satisfy both properties. Similar to determinant codes, the super-message matrix of the cascade codes will be multiplied (from left) by an encoder matrix to generate the coded content of the nodes. Also, the condition in \ref{cond:G} is an additional requirement in comparison to the encoder matrix of ($k=d$)-signed determinant codes.
	
	In the following, we define a specific class of codes, called semi-systematic codes, and their corresponding encoder matrix.
	\begin{defi}
		\label{def:semi-sys}
		We call a signed determinant (or cascade) code \emph{semi-systematic code} if the contents of the first $k$ nodes are identical to the $k$ rows of the message matrix (or super message matrix), i.e. if $\bC_{i,:} = \e{D}_{i,:}$ (or $\bC_{i,:} = \MM_{i,:}$) for every $i\in \intv{k}$. The encoder matrix $\enc$ of a semi-systematic code should consist of a $k\times k$ identity matrix at its upper-left corner and a $k\times (d-k)$ zero matrix in its upper-right corner. We also refer to such an encoder matrix as  \emph{semi-systematic encoder}. \footnote{Note that the symbols of the message matrix are not necessarily raw symbols from the file, due to the parity equation in \eqref{eq:parityeq} as well as injections. Hence, we rather call these codes semi-systematic to distinguish them from the standard notion of systematic codes.}
	\end{defi}	
	In general, any matrix satisfying conditions \ref{cond:G} and \ref{cond:P} can be converted to a semi-systematic encoder matrix as it is discussed in Appendix~\ref{app:semi-sys}. For the code construction of this paper, the encoder matrix doesn't need to be semi-systematic, and we prove the properties of the code for the general encoder matrices defined in~\eqref{def:encoder}. However, to demonstrate the main idea of the code construction, we use semi-systematic encoders.  
	
	\subsection{An Overview: Cascading Message Matrices of Determinant Codes}
	\label{sec:cascade}
	The general structure of super message matrix $\MM$ for an $(n,k,d;\mu)$ cascade code is presented in Fig.~\ref{fig:cascade}. This matrix is obtained by concatenating the message matrices of multiple $(n,d,d;m)$ signed determinant codes with different modes, ranging from $m=0$ to $m=\mu$. The number of required message matrices for each mode $m$ is denoted by $\mdcnt_m$. Then, this super-message matrix will be multiplied (from left) by an encoder matrix $\enc$ to generate the node contents. Therefore, the codewords (the content of the nodes) will also be a collection of the codewords of the signed determinant codes used as building blocks. The following definition is what we use to refer to each of these message matrices
	\begin{defi}
			We refer to the determinant code message matrices that are concatenated to form  a super-message matrix as \emph{code segments}. Similarly, the ultimate codewords of a cascade code comprise of multiple  \emph{codeword segments}, each corresponding to the multiplication of the encoder matrix by one code segment.
	\end{defi}
	The construction starts with the message matrix of an $(n,d,d)$ system with mode $m=\mu$ called the \emph{root} of the cascade code. There is only one message matrix (code segment) of this mode ($\mdcnt_\mu=1$), and all other message matrices have modes less than $\mu$. These message matrices need to be modified such that the code generated from multiplication of  the encoder to the ultimate super message matrix can provide data recovery from any $k$ nodes and also enables the exact-repair property from any $d$ helper nodes.
	
	\begin{figure*}[!t]
		\centering
		\includegraphics[width=\textwidth,page=1]{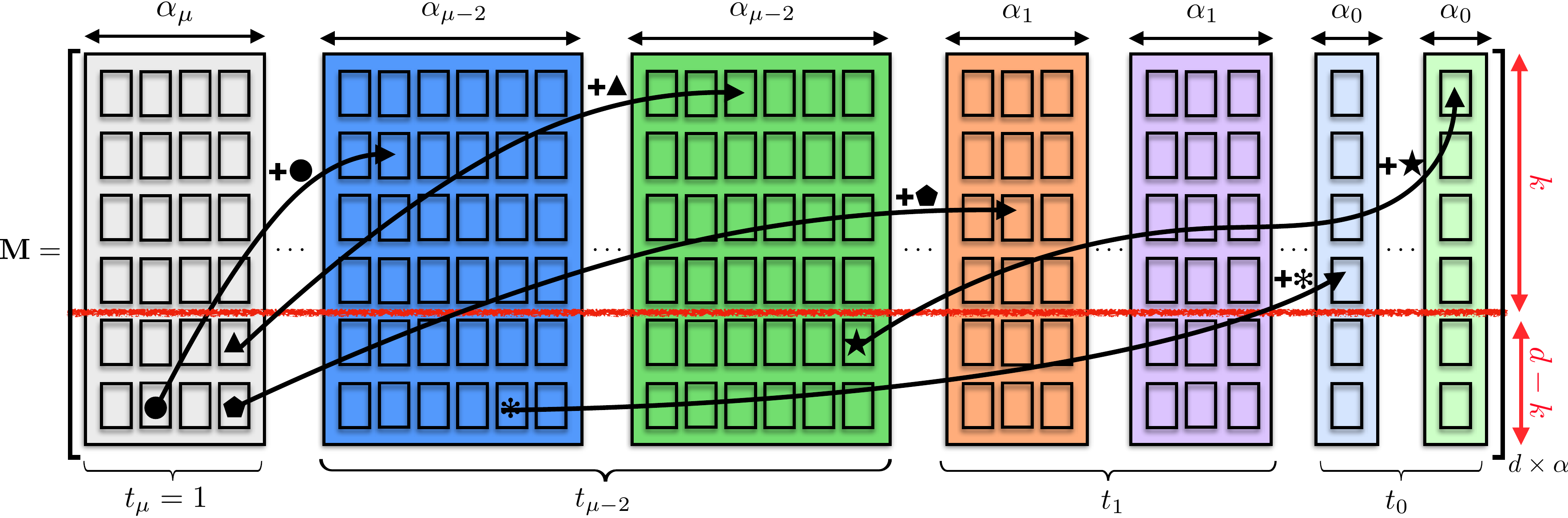}
		\caption{Cascading of determinant codes. In this figure, every rectangle represents the message matrix of one determinant code, and $\mdcnt_m$ denotes the number of message matrices of mode $m$, for $m\in \set{0,1,\dots, \mu}$. The message matrix at mode $m$ is of size $d\times \alpha_m$, where $\alpha_m = \binom{d}{m}$. These message matrices are placed from the highest mode to the lowest mode. The leftmost message matrix corresponds to the root of the cascade code with mode $m=\mu$. The rightmost message matrices have either a mode of $m=1$ or $m=0$. In the message matrices $\MM$, some of the symbols in the bottom $d-k$ rows will be  missing in data recovery. These missing symbols are backed up by begin injected (adding with a sign) into the parity symbols located at the top $k$ rows of other determinant codes with lower modes. Injections are demonstrated by arrows from a symbol in the parent matrix to a parity symbol in the child matrix.}
		\label{fig:cascade}
	\end{figure*} 
	
	To demonstrate the main idea, we can assume the encoder matrix $\enc_{n\times d}$ is semi-systematic, and thus the content of node $i\in\intv{k}$ is the same as the $i$-th row of the matrix $\MM$. Now, consider an attempt for data recovery from the first $k$ nodes, and assume $\e{P}$ is a code segment of $\MM$ with mode $m_1$. It is clear that without making any modification, some symbols in rows $[k+1:d]$ of $\e{P}$ cannot be recovered. This is due to the fact that we are performing data recovery only from the first $k<d$ nodes, and these symbols do not appear (coded or uncoded) in the content of the top $k$ nodes. We refer to such symbols as \emph{missing symbols}.
	
	In order to overcome the data recovery challenge, one can leverage the parity symbols located in the top $k$ rows of a message matrix  $\e{Q}$, associated with another signed determinant code and mode $m_2$, where  $m_2<m_1$. Recall from Remark~\ref{rem:raw:parity}  that a symbol $w_{\max \sen Y,\sen Y}$ in $\e{Q}$ with $\size{\sen Y}=m_2+1$ is a parity symbol located in position  $\left(\max \sen Y, \sen Y \setminus \set{\max \sen Y}\right)$ of matrix $\e{Q}$.  If $\sen Y \subseteq \intv{k}$, all the symbols of the $w$-group $\sen Y$ (defined in Definition~\ref{def:parity}), including the parity symbol, will be located in the upper $k$ rows of $\e{Q}$, and hence, appear in the content of the first $k$ nodes used for data recovery. This will make the symbol $w_{\max \sen Y,\sen Y}$ redundant in the data recovery, as it can be also obtained from the other $w$-group symbols of $\sen Y$, i.e. $\{w_{i,\sen Y}: i\in \sen Y, i<\max \sen Y\}$, and using the parity equation in~\eqref{eq:parityeq}. Therefore, by adding (a possibly signed copy of)  a missing symbol $z$ from the bottom $\intv{k+1:d}$ rows of the massage matrix $\e{P}$ to symbol $w_{\max \sen Y,\sen Y}$ of $\e{Q}$, we can provide a backup for this missing symbol. With this, the entry at position $\left(\max \sen Y, \sen Y \setminus \set{\max \sen Y}\right)$ of matrix $\e{Q}$ will be $w_{\max \sen Y, \sen Y} \pm z$. As the parity part $w_{\max \sen Y, \sen Y}$ can be independently recovered from the parity equation, one can remove it from the entry and obtain a copy of the missing symbol $z$. Based on this description for providing backups for missing symbols, we formally define some  terms in the following definition, that will be used in the rest of this paper.
	
	\begin{defi}
			We define the following terms, which will be used in the construction of cascade codes:
			\begin{itemize}
				\item \textbf{Symbol injection} is referred to the process of providing backup copy(s) for missing symbols,   which is simply adding a signed version of them to the parity symbols of (some of) other message matrices. 
				\item \textbf{Parent matrix/code (of an injection)} is referred to the message matrix of the signed determinant code whose symbols in the lower part need protection, and is hence injected. 
				\item \textbf{Child matrix/code (of an injection)} is the message matrix into which symbols are injected.
			\end{itemize}
	\end{defi}
		
	For the injection of a specific missing symbol of a parent matrix $\e{P}$, we have to carefully choose the child code  and  the parity symbol of the child code into which the injection occurs.  The main challenge is to introduce such backup copies such that they do not demolish the repair process. Note that from Proposition~\ref{prop:ndd:repair}, each of the code segments used in the super message matrix $\MM$ is repairable before injection (modification). However, after modification they are no longer repairable by the standard repair process of $(n,d,d;m)$ codes. Therefore, a careful design for the injection and a revised repair process are needed that can simultaneously guarantee data recovery and node repairability.
	
	In order to protect the missing symbols of the root, some child matrices need to be introduced. The missing symbols of the newly introduced child code should also be backed up by injecting into other determinant codes with lower modes. This leads to a \emph{cascade} of determinant codes. This  construction process continues until it reaches signed determinant codes that do not need a further injection.

	Although the main idea of injection is inspired based on a semi-systematic encoder and data recovery from the top-$k$ nodes, we will prove that, with no further modification, the same super message matrix (after injection) leads to a cascade code that provides data recovery for any choice of $k$ nodes.
	
	The details of the construction of the cascade super message matrix are discussed in the following section.
	
	\section{Super message matrix of $(n,k,d)$ codes}
	\label{sec:supermessage}
    In this section, the details of the cascade structure of the super message matrix of an $(n,k,d)$ code are given in several steps. In order to understand the details and ideas of code construction, we present a running example for an $(n, k=4, d=6;\mu=4)$ code with parameters $(\alpha,\beta,F)=(81,27,324)$, as indicated by \eqref{eq:params}. The construction of this code instance is also broken into several steps, according to the details of the general code construction. Note that a code with parameters  $(\alpha,\beta,F)=(81,27,324)$ is indeed an MSR code, since $F=k\alpha$ and $\beta = \alpha/(d-k+1)$, for which several code constructions are known (e.g.  \cite{ye2017explicit,ye2017explicitnearly,sasidharan2016explicit,li2018generic,li2018alternative,rawat2016progress}). Nevertheless, the 

	(unnormalized) code parameters, such as sub-packetization level, of the existing codes will grow by increasing the number of nodes, $n$, in the system. However, in the code construction of this paper, the sub-packetization level is independent of $n$. Also, cascade codes enable the flexibility of system expansion by adding new parity nodes, without changing the contents of already existing nodes.\footnote{The only limitation for adding new nodes to the system is the field size $q$, that has to satisfy $q>n$.}
	
	In a nutshell, the code construction is a recursive algorithm that creates a rooted tree similar to that of Fig.~\ref{fig:tree}. In this tree, nodes represent message matrices of different determinant codes (code segments) that are concatenated to form the super message matrix. The choice of the root of this tree will be discussed in Subsection~\ref{subsec:root}. In Subsection~\ref{subsec:grouping}, we explain that which symbols of a signed determinant code (a parent matrix $\e{P}$ associated to a node in the tree) need to be protected by injecting into another code (the child matrix $\e{Q}$, corresponding to another node in the tree, which lies under\footnote{We may exchangeably use child/parent \emph{matrix} or \emph{node} when we discuss the hierarchical tree} $\e{P}$). Next, in  Subsection~\ref{subsec:injection} all child nodes (matrices) of the parent node (matrix) will be identified in Remark~\ref{rem:injpair}, and then mode, signature vector, and modifications of each child's message matrix  are explained. Finally, Subsection~\ref{subsec:cascade:struct} completes the details of the construction of the super message matrix. To conclude this section, we present Algorithm~\ref{alg:supermsg} that summarizes and incorporates all steps of the code construction.
	
	\subsection{Root of the Cascade Code}
	\label{subsec:root}
	In order to construct the message matrix of an $(n,k,d;\mu)$ code, we first generate a massage matrix for an $(n,k'=d,d; m=\mu)$ determinant code of mode $m=\mu$ and we choose an all-zero vector (i.e., $\s{\e{T}_0}{x}=0$ for every $x\in \intv{d}$) as the signature vector of this code.
	\begin{ex}
		\label{ex:root}
		Our goal is to generate an  $(n, k=4, d=6;\mu=4)$ cascade code, with parameters $(\alpha,\beta,F)=(81,27,324)$. To this end, we start with a determinant code $(d=6; m_0=4)$ and denote its message matrix by $\e{T}_0$. We choose the signature vector $\sigma_{\e{T}_0}=(0,0,0,0,0,0)$. The size of this message matrix is $d\times \alpha_{m_0} = d\times \binom{d}{m_0} = 6\times \binom{6}{4} = 6 \times 15$. The message matrix $\e{T}_0$ of this determinant code is given in~\eqref{eq:D0} in the next page. We will need to define matrices $\e{T}_1, \e{T}_2,\cdots, \e{T}_{14}$ later to complete the code construction.  We distinguish the entries of matrix $\e{T}_i$ by the superscript $\nd{i}$.
		Note that horizontal line in the matrix separates the top $k=4$ rows from the bottom $(d-k)=6-4=2$ rows. We will refer to the top $4\times 15$ sub-matrix  by $\up{\e{T}_0}$, and to the bottom $2\times 15$ sub-matrix  by $\down{\e{T}_0}$. In the representation of this matrix, some symbols are designated by a frame around them, with different background colors, which will be explained in the next subsections.
		\clearpage
		\begin{align}
	    	\includegraphics[width=0.9\textwidth,page=1]{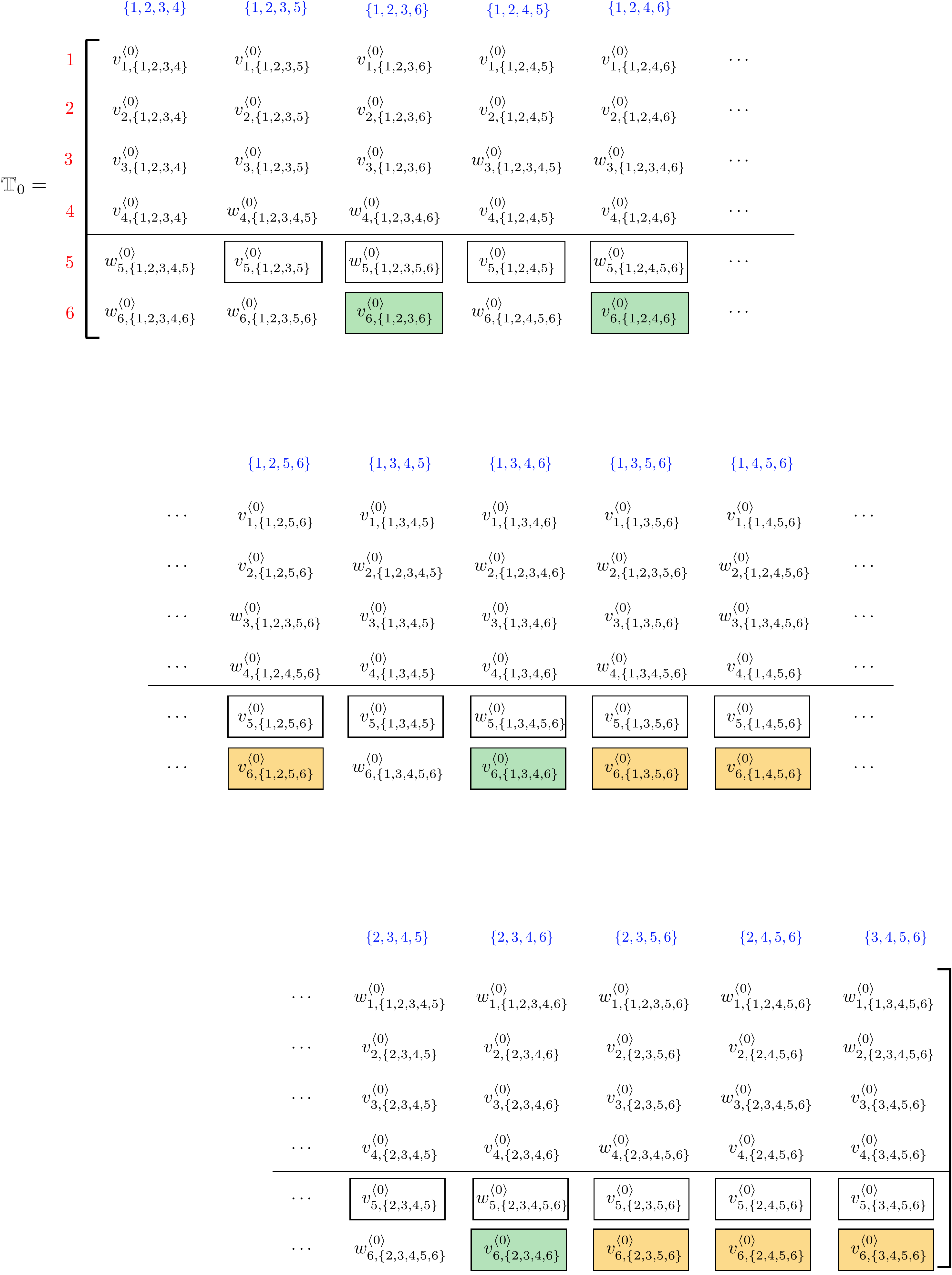}
    	    \label{eq:D0}
	    \end{align}
	\end{ex}

	\subsection{Grouping of Symbols}
	
	\label{subsec:grouping}
	Next, we need to group the missing symbols at the lower $(d - k)$ rows of a determinant code. This is due to the fact that the symbols in each group will be treated differently. To further elaborate on this, we need the following definition.	
	\begin{defi}
		Consider an $(n,d,d;j)$ determinant code at mode $j$ with message matrix $\e{P}$. Recall from~\eqref{eq:def:S} that, $\e{P}_{x,\sen J}$ denotes the entry in the row $x$ and column $\sen J$ with $\size{\sen J}=\md{\e{P}}=j$. For an entry $\e{P}_{x,\sen J}$ in the lower $(d-k)$ rows of $\e{P}$, we have $x\in \intv{k+1:d}$, and set $\sen J$ can be partitioned into two disjoint sets $\sen A = \sen J \cap \intv{k}$ and $\sen B = \sen J \cap \intv{k+1:d}$, where $\size{\sen A} + \size{\sen B} = \size{\sen J} = \md{\e{P}}$. The entries in the lower $(d-k)$ rows will be classified  into three disjoint groups as:
			\begin{itemize}
				\item $\dgp{P} = \set{\e{P}_{x,\sen A \cup \sen B}: x \leq \max \sen B, \sen A \neq \varnothing}$,
				\item $\ngp{P}=  \set{\e{P}_{x,\sen A \cup \sen B}: x \leq \max \sen B, \sen A = \varnothing}$,
				\item $\pgp{P}=  \set{\e{P}_{x,\sen A \cup \sen B}: x> \max \sen B}$.
			\end{itemize}
		\label{def:groups}
	\end{defi}
	We treat symbols in the above-mentioned three groups as follows:
	\begin{itemize}
			\item Symbols in $\dgp{P}$ are raw data  symbols that need to be protected  by injecting into a redundant parity symbol located in the top $k$ rows of child message matrices with lower modes.
			\item Symbols in $\ngp{P}$ will be set to zero (nulled). This yields a reduction of $N_j=\size{\ngp{P}}$ in $F_j$, the number of raw data symbols in the message matrix.
			\item Symbols in $\pgp{P}$ are essentially parity symbols. This is because for $\e{P}_{x,\sen A \cup \sen B} \in \pgp{P}$, we have $x> \max \sen B= \max (\sen A \cup \sen B) = \max \sen J$, and from Remark~\ref{rem:raw:parity}, this is a parity symbol. This symbol can be recovered using the parity equation in~\eqref{eq:parityeq}. More precisely, we have
			\begin{align*}
			    	\e{P}_{x,\sen A \cup \sen B } &= \e{P}_{x,\sen J} = (-1)^{\s{\e{p}}{x}} w_{x,\set{x} \cup \sen J} = (-1)^{\s{\e{p}}{x}+j} \sum_{y \in \sen J}(-1)^{\ind{\sen J \cup \set{x}}{y}}w_{y,\set{x} \cup \sen J},
			\end{align*}
			which can be evaluated after once  symbols in the summation are recovered. Therefore, there is no need to provide backups for the symbols in $\pgp{P}$.
	\end{itemize}
	In summary, every symbol $\e{P}_{x,\sen A \cup \sen B}$ located at the bottom $\intv{k+1:d}$ rows of  $\e{P}$ is injected into another determinant code if and only if it is not a parity symbol and $\sen A \neq \varnothing$.
	
	\begin{ex}
		\label{ex:groups}
		According to Definition~\ref{def:groups}, the symbols in $\down{\e{T}_0}$ can be grouped as follows:
		\begin{align*}
		\dgpi{T}{0} & \hspace{-0.5mm}=\hspace{-0.5mm} \left\{ \hspace{-1.5mm} \begin{array}{l}
		v_{5,\set{1,2,3,5}}^{\nd{0}},v_{5,\set{1,2,4,5}}^{\nd{0}},v_{5,\set{1,2,5,6}}^{\nd{0}},v_{5,\set{1,3,4,5}}^{\nd{0}}, 	v_{5,\set{1,3,5,6}}^{\nd{0}},\\
		v_{5,\set{1,4,5,6}}^{\nd{0}}, v_{5,\set{2,3,4,5}}^{\nd{0}},
		v_{5,\set{2,3,5,6}}^{\nd{0}}, v_{5,\set{2,4,5,6}}^{\nd{0}}, 
		v_{5,\set{3,4,5,6}}^{\nd{0}},\\ v_{6,\set{1,2,3,6}}^{\nd{0}},v_{6,\set{1,2,4,6}}^{\nd{0}}, v_{6,\set{1,2,5,6}}^{\nd{0}}, v_{6,\set{1,3,4,6}}^{\nd{0}},
		v_{6,\set{1,3,5,6}}^{\nd{0}} ,\\
		v_{6,\set{1,4,5,6}}^{\nd{0}}, v_{6,\set{2,3,4,6}}^{\nd{0}},v_{6,\set{2,3,5,6}}^{\nd{0}}, v_{6,\set{2,4,5,6}}^{\nd{0}},v_{6,\set{3,4,5,6}}^{\nd{0}},\\
		w_{5,\set{1,2,3,5,6}}^{\nd{0}},w_{5,\set{1,2,4,5,6}}^{\nd{0}},w_{5,\set{1,3,4,5,6}}^{\nd{0}},w_{5,\set{2,3,4,5,6}}^{\nd{0}}
		\end{array}\hspace{-1.5mm} \hspace{-1pt}\right\}\hspace{-1pt},\\
		\ngpi{T}{0} &\hspace{-0.5mm}=\hspace{-0.5mm}\varnothing,\\
		\pgpi{T}{0} &\hspace{-0.5mm}=\hspace{-0.5mm} \left\{
		\begin{array}{l}  w_{5,\set{1,2,3,4,5}}^{\nd{0}},w_{6,\set{1,2,3,4,6}}^{\nd{0}},w_{6,\set{1,2,3,5,6}}^{\nd{0}},w_{6,\set{1,2,4,5,6}}^{\nd{0}},w_{6,\set{1,3,4,5,6}}^{\nd{0}},w_{6,\set{2,3,4,5,6}}^{\nd{0}} \end{array} 
		\right\}.
		\end{align*}
		
		The symbols in $\dgpi{T}{0}$ are marked in boxes in \eqref{eq:D0}, to indicate that they need to be injected into other code segments with lower modes. Since $\ngpi{T}{0} = \varnothing$, no symbol will be set to zero. Symbols in $\pgpi{T}{0}$ can be recovered from the parity equations, and won't be injected. These symbols located below the horizontal line, without a surrounding frame in \eqref{eq:D0}.
	\end{ex}
	The details of the injection process are given in the following subsection.
	\subsection{The Injection Process} 
	\label{subsec:injection}
	In the previous part, it was discussed that only symbols in $\dgp{P}$ need to be protected by being injected into the message matrices of child codes of $\e{P}$. Recall that the term injection refers to the addition of a signed copy of a missing symbol of the parent matrix to a parity (redundant) symbol of the child matrix. \emph{It is important to note that a single symbol in $\dgp{P}$ might be injected into several child matrices}. One of such injections is called \emph{primary}, and the rest are called \emph{secondary}.  In fact, in a primary injection, a symbol is injected into one of the parities located in the upper $k$ rows of a child node while secondary injection(s) will be performed into parity symbol located in the lower $(d-k)$ rows of a child matrix. The goal of primary injections is to preserve the data recovery property, while secondary injections are performed in order to preserve the node repair property. In this subsection, we first introduce child nodes of a parent node and then explain the details of primary and secondary injections into a given child node.
	
	To fully cover the details of injection process, this subsection consists of four parts: 1) child matrices of a given parent matrix, 2) mode of a child matrix, 3) signature vector of a child matrix, and 4) modified child matrix after the injection. These parts are explained in the following. 
	
	\noindent \underline{\textbf{Child matrices of a given parent matrix:}}\\
	For a given parent matrix $\e{P}$, the group $\dgp{P}$  will be further partitioned into several subgroups, depending on their pair of $(x,\sen B)$. Recall that for a symbol $\e{P}_{x,\sen J}$, we identified $\sen A = \sen J \cap \intv{k}$ and $\sen B=\sen J \cap \intv{k+1:d}$. All the symbols of the form $\e{P}_{x,\sen A \cup \sen B}$ in $\dgp{P}$ with the same pair of $(x,\sen B)$ will be \emph{primarily} injected into parity symbols in the top $k$ rows of the same child determinant code. We refer to $(x,\sen B)$ as the \emph{injection pair} and we denote the message matrix of the child code of $\e{P}$ into which these symbols are injected by $\eSMat{Q}{m}{x,\sen {B}}{\e{P}}$ (or simply $\e{Q}$, whenever its parent matrix and injection pair are clear from the context). The following remark gives a more precise description of child nodes of a parent matrix.
	\begin{rem}
		\label{rem:injpair}
		A parent code segment with message matrix $\e{P}$, has several child nodes of the form  $\eSMat{Q}{m}{x,\sen {B}}{\e{P}}$, for all pairs $(x,\sen B)$ satisfying
			\begin{description}
				\item[\namedlabel{cond:IP-i}{(i)}] $\sen B \subseteq [k+1:d]$,
				\item[\namedlabel{cond:IP-ii}{(ii)}] $\size{\sen B} <  \md{\e{P}}$,
				\item[\namedlabel{cond:IP-iii}{(iii)}] $x\in [k+1:d]$, 
				\item[\namedlabel{cond:IP-iv}{(iv)}] $x\leq \max \sen B$.
		\end{description}
		The above conditions come from the description of $\dgp{P}$ in Definition~\ref{def:groups}. Especially, condition \ref{cond:IP-ii} is equivalent to $\sen B$ being a proper subset of $\sen J$, and thus $\size{\sen B}<\size{\sen A \cup \sen B} = \md{\e{P}}$ guarantees that for the injecting symbol $\e{P}_{x,\sen A \cup \sen B}\in \dgp{P}$ we have  $\sen A \neq \varnothing$. Condition~\ref{cond:IP-iii} ensures that symbol $\e{P}_{x,\sen A \cup \sen B}$ lies in the lower submatrix $\down{\e{P}}$, and \ref{cond:IP-iv} rules out the parity symbols. 
		Note that we defined $\max \varnothing=-\infty$, and therefore conditions \ref{cond:IP-iii} and \ref{cond:IP-iv} rule out the  choice of $\sen B=\varnothing$.
	\end{rem}
	In Section~\ref{sec:parameters}, the number of child matrices for a given parent matrix is evaluated based on Remark~\ref{rem:injpair}.

	\begin{ex}
		\label{ex:inj:pair}
		In Table~\ref{table:T0:children}, all the injection pairs $(x,\sen B)$ satisfying conditions of Remark~\ref{rem:injpair} for $\e{T}_0$ of Example~\ref{ex:root} are enumerated. Note that each injection pair in the table corresponds to one child matrix. 
		For the ease of notation, we use $\e{T}_1 = \eSMatt{T}{1}{m}{5, \{5\}}{\e{T}_0}$ to refer to the child matrix of $\e{T}_0$ associated with the injection pair $(5,\{5\})$. Similarly, child matrices  $\e{T}_2,\e{T}_3, \e{T}_4 ,\e{T}_5$ are defined in Table~\ref{table:T0:children}. For a given $(x,\sen B)$, the corresponding child matrix will primarily host symbols of the form $\left[\e{T}_0 \right]_{x,\sen A \cup \sen B}\in \dgp{T}$. Recall from  Definition~\ref{def:groups} that we have  $\sen A \subseteq \intv{4}$ and $\size{\sen A}+\size{\sen B}= \md{\e{T}_0}=4$.
    	\begin{table*}[!h]
			\centering
			\begin{align*}
			\begin{array}{|c|c|c|c|}
			\hline
			\text{Pair } (x,\sen B) &\text{Symbol Format}& \text{Primarily Injected Symbols}  & \text{Child Matrix}\\ \hline
			(5,\set{5})& \left[\e{T}_0 \right]_{5,\sen A \cup \set{5}}  & v_{5,\set{1,2,3,5}}^{\nd{0}},v_{5,\set{1,2,4,5}}^{\nd{0}},v_{5,\set{1,3,4,5}}^{\nd{0}},v_{5,\set{2,3,4,5}}^{\nd{0}} & \e{T}_1 = \eSMatt{T}{1}{m}{5, \{5\}}{\e{T}_0} \\ \hline
			(6,\set{6}) & \left[\e{T}_0 \right]_{6,\sen A \cup \set{6}}  & v_{6,\set{1,2,3,6}}^{\nd{0}},v_{6,\set{1,2,4,6}}^{\nd{0}},v_{6,\set{1,3,4,6}}^{\nd{0}},v_{6,\set{2,3,4,6}}^{\nd{0}} & \e{T}_2 = \eSMatt{T}{1}{m}{6, \{6\}}{\e{T}_0} \\ \hline
			(5,\set{6}) & \left[\e{T}_0 \right]_{5,\sen A \cup \set{6}}& w_{5,\set{1,2,3,5,6}}^{\nd{0}},w_{5,\set{1,2,4,5,6}}^{\nd{0}},w_{5,\set{1,3,4,5,6}}^{\nd{0}},w_{5,\set{2,3,4,5,6}}^{\nd{0}} & \e{T}_3 = \eSMatt{T}{1}{m}{5, \{6\}}{\e{T}_0} \\ \hline
			(5,\set{5,6})& \left[\e{T}_0 \right]_{5,\sen A \cup \set{5,6}} & v_{5,\set{1,2,5,6}}^{\nd{0}},v_{5,\set{1,3,5,6}}^{\nd{0}},v_{5,\set{1,4,5,6}}^{\nd{0}},v_{5,\set{2,3,5,6}}^{\nd{0}},v_{5,\set{2,4,5,6}}^{\nd{0}},v_{5,\set{3,4,5,6}}^{\nd{0}} & \e{T}_4 = \eSMatt{T}{4}{m}{5, \{5,6\}}{\e{T}_0} \\ \hline
			(6,\set{5,6})& \left[\e{T}_0 \right]_{6,\sen A \cup \set{5,6}} &  v_{6,\set{1,2,5,6}}^{\nd{0}},v_{6,\set{1,3,5,6}}^{\nd{0}},v_{6,\set{1,4,5,6}}^{\nd{0}},v_{6,\set{2,3,5,6}}^{\nd{0}},v_{6,\set{2,4,5,6}}^{\nd{0}},v_{6,\set{3,4,5,6}}^{\nd{0}} & \e{T}_5 = \eSMatt{T}{5}{m}{6, \{5,6\}}{\e{T}_0}\\ \hline
			\end{array}
			\end{align*}
			\caption{Injection pairs for the root matrix of the running example.}
			\label{table:T0:children}
		\end{table*} 
		
		Note that all symbols of $\dgp{T}$ are marked in a box in \eqref{eq:D0}. Moreover, those  symbols that will be injected into $\e{T}_2$ and $\e{T}_5$ are further distinguished by different background colors (green and orange for $\e{T}_2$ and $\e{T}_5$, respectively.).
	\end{ex}
	
	\noindent \underline{\textbf{Mode of a child matrix:}}\\
	Recall that the child matrix $\eSMat{Q}{m}{x,\sen {B}}{\e{P}}$ is introduced for the primary injection of a symbol $\e{P}_{x,\sen A \cup \sen B}\in \dgp{P}$. In the primary injection, a (signed copy of a) missing symbol $\e{P}_{x,\sen A \cup \sen B}$ will be injected into the parity symbol $\pm w_{\max \sen A, \sen A}$ located in the row $i=\max \sen A$ and column $\sen I = \sen A \setminus \set{\max \sen A}$ of the message matrix $\eSMat{Q}{m}{x,\sen {B}}{\e{P}}$. Since $\sen A \subseteq \intv{k}$, we have $1 \leq i \leq k$, which means  symbol $\pm w_{\max \sen A, \sen A}$ is located in one of the top $k$ rows of  $\eSMat{Q}{m}{x,\sen {B}}{\e{P}}$. Moreover, since the columns of $\e{Q}$ are labeled by sets of size $\size{\sen{A}}-1$ (see equation~\eqref{eq:vw-symbols}), the relation between the mode of the child matrix and the parent matrix for a fixed injection pair $(x,\sen B)$ is given by
		\begin{align}
		\begin{split}
		\md{\e{Q}} &= \size{\sen A} - 1 
		= \md{\e{P}} - \size{\sen B} -1,
		\end{split}
		\label{eq:mode:relation}
		\end{align}
	where the second equality is due to the fact that $(\sen A,\sen B)$ provides a disjoint partitioning for $\sen I$ (see Definition~\ref{def:groups}).
	\begin{ex}
		\label{ex:modes}
		Following up from Example~\ref{ex:inj:pair}, the mode of each child matrix 
		defined in 
		Table~\ref{table:T0:children} can be evaluated according to~\eqref{eq:mode:relation}:
		\begin{align*}
		\md{\e{T}_1} &= \md{\e{T}_2} = \md{\e{T}_3} = \md{\e{T}_{0}} - 1 - 1 = 2\\
		\md{\e{T}_4} &= \md{\e{T}_5} = \md{\e{T}_{0}} - 2 - 1 = 1.
		\end{align*}
	\end{ex}
	\noindent \underline{\textbf{Signature vector of a child matrix:}}\\
	In order to fully characterize a child matrix, it is only left  to define its signature vector. For a parent matrix $\e{P}$ with signature vector $\s{\e{P}}{\cdot}$, the signature vector of $\eSMat{Q}{m}{x,\sen {B}}{\e{P}}$, the child matrix associated with an injection pair $(x,\sen B)$, is given by 
		\begin{align}
		\s{\e{Q}}{i} = 1+ \s{\e{P}}{i} + \ind{\sen B \cup \set{i}}{i}, \qquad \forall i \in \intv{d}.
		\label{eq:inj:sign}
		\end{align}
		
	Such a signature vector is chosen in order to enable the exact repair property. This will be more clear once we discuss the repair process in Section~\ref{sec:noderepair}.
	\begin{ex}
		\label{ex:signature}
		Continuing from Example~\ref{ex:inj:pair},  and the fact that $\sigma_{\e{T}_0}(i)=0$ for $i\in \{1,\dots, 6\}$, 
		the signature vectors for the child matrices $\e{T}_2$ and $\e{T}_5$ of  Table~\ref{table:T0:children} can be obtained from
		\begin{align*}
		\forall i\in \set{1,2,3,4,5,6} : \sigma_{\e{T}_2}(i) = 1+ \s{\e{T}_0}{i} + \ind{\set{6} \cup \set{i}}{i} \rightarrow \sigma_{\e{T}_2} = (2,2,2,2,2,2),
		\end{align*}
		and
		\begin{align*}
		\forall i\in \set{1,2,3,4,5,6} : \sigma_{\e{T}_5}(i) = 1+ \s{\e{T}_0}{i} + \ind{\set{5,6} \cup \set{i}}{i} \rightarrow \sigma_{\e{T}_5} = (2,2,2,2,2,3). 
		\end{align*}
	\end{ex}
	\noindent \underline{\textbf{Modified child matrix after the injection:}}\\
	Let $\eSMat{Q}{m}{x,\sen {B}}{\e{P}}$ be a child matrix of $\e{P}$ associated with the injection pair $(x,\sen B)$, into which some symbols of $\e{P}$ will be injected. We use  $\eSMat{Q}{m}{x,\sen {B}}{\e{P}}$ and  $\fSMat{Q}{m}{x,\sen {B}}{\e{P}}$ to refer to this matrix before and after injection (modification), respectively. 
	This modification is modeled by adding an \emph{injection matrix} $\tOMat{m}{x,\sen {B}}{\e{P}}$ to the original matrix $\eSMat{Q}{m}{x,\sen {B}}{\e{P}}$, i.e., 
	\begin{align}
	\fSMat{Q}{m}{x,\sen {B}}{\e{P}} = \eSMat{Q}{m}{x,\sen {B}}{\e{P}} + \tOMat{m}{x,\sen {B}}{\e{P}}.
	\label{eq:Qdef}
	\end{align}
	Here $\tOMat{m}{x,\sen {B}}{\e{P}}$ is a matrix with size and row/column labeling identical to those of  $\e{Q}$, that consists of a signed version of entries of matrix $\e{P}$ which should be injected into $\e{Q}$. 
	The entry at position $(i, \sen I)$ is given in~\eqref{eq:inj:mat}.\footnote{Note that for a cascade code constructed over Galois field $\mathsf{GF}(2^s)$ with characteristic $2$ (i.e., $-1=+1$), the injection equation in~\eqref{eq:inj:mat} reduces to $\left[ \tOMat{m}{x,\sen {B}}{\e{P}} \right]_{i,\sen I}   = \e{P}_{x, \sen I \cup \set{i} \cup \sen B } \1{i> \max  \sen I , i \notin \sen B  ,\sen I  \cap \sen B = \varnothing}$.}
    \begin{align}
    	\left[ \tOMat{m}{x,\sen {B}}{\e{P}} \right]_{i,\sen I} = \begin{cases}
    	(-1)^{1+\s{\e{P}}{i} + \ind{\sen I \cup\set{i} \cup \sen B}{i}}\e{P}_{x, \sen I \cup \set{i} \cup \sen B } & \mathsf{if\;} i> \max  \sen I , i \notin \sen B  ,\sen I  \cap \sen B = \varnothing,  \\
    	0 & \mathsf{otherwise}.
    	\end{cases}
    	\label{eq:inj:mat}
	\end{align}
	Here, the coordinates of injection satisfy $i\in \intv{d}$ and $\sen I \subseteq \intv{d}$ with $\size{\sen I} =\md{\e{Q}} = \md{\e{P}} - |\sen B|-1$. Note that the entries of $\tOMat{m}{x,\sen {B}}{\e{P}}$ are non-zero only for certain positions, which correspond to the parity entries of $\e{Q}$ into which an injection is performed. 
	The following remark further elaborates on the above injection matrix.
		\begin{rem}
			\label{rem:inj:explain}
			The following facts hold for the injection matrix given in ~\eqref{eq:inj:mat}:
			\begin{enumerate}[label=\arabic*.]
				\item  The conditions $i > \max \sen I$, $i \notin \sen B$, and $\sen I \cap \sen B = \varnothing$ guarantee that sets $\set{i}$, $\sen I$, and $\sen B$  are mutually disjoint, and hence the size of $\sen I \cup \set{i} \cup \sen B$ (the column label of the injected symbol) equals  $\md{\e{P}}$. 
				In other words, we can verify that $\md {\e{P}} = \size{\sen I \cup \set{i} \cup \sen B} = \size{\sen I} + \size{\set{i}} + \size{\sen B} = \md{\e{Q}} + 1 + \size{\sen B}$.
				\item The condition $i> \max \sen I$ will additionally guarantee that the entry at position $(i,\sen I)$ of
				the matrix $\e{Q}$ is a parity symbol of the child matrix (see Remark~\ref{rem:raw:parity}). Therefore, an injection only occurs in a parity symbol. Note that the entries of $\tOMat{m}{x,\sen {B}}{\e{P}}$ are non-zero only for certain positions, which correspond to the parity entries of $\eSMat{Q}{m}{x,\sen {B}}{\e{P}}$ into whom the injection is performed.
				\item The symbol $\e{P}_{x,\sen A \cup \sen B} \in \dgp{\e{P}}$ will be primarily injected into position $(\max \sen A,\sen A \setminus \set{\max \sen{A}})$ of the child matrix $\eSMat{Q}{m}{x,\sen {B}}{\e{P}}$. This can be seen from~\eqref{eq:inj:mat} as
				\begin{align}
				    &\left[ \tOMat{m}{x,\sen {B}}{\e{P}} \right]_{\max {\sen A},\sen A \setminus \set{\max \sen{A}}} = (-1)^{1+\s{\e{P}}{\max \sen A} + \ind{\sen A \cup \sen B}{\max {\sen A}}}\e{P}_{x, \sen A \cup \sen B }. \label{eq:primary:inj}
				\end{align}
			\end{enumerate}
		\end{rem}

	\begin{ex}
		\label{ex:modified:child}
		Let us continue with our running construction from Example~\ref{ex:root}. For an $(n,k=4,d=6;\mu=4)$ code, no injection takes place into $\e{T}_0$, since it is the root node, and hence, we have $\f{T}_0 = \e{T}_0$, i.e., $\f{T}_0$ is  a purely signed determinant message matrix of mode $m_0=4$. 
		
		Consider  the child matrix $\e{T}_2 =
		\eSMatt{T}{2}{m}{6, \{6\}}{\e{T}_0}$, whose mode $m_2=2$ and  signature vector $\sigma_{\e{T}_2}$ were evaluated in Examples~\ref{ex:modes} and~\ref{ex:signature}, respectively. This message matrix has $\binom{d}{m_2}=\binom{6}{2} = 15$ columns, and its columns are labeled by subsets of $\{1,2,3,4,5,6\}$ of size $m_2=2$. The entries of this message matrix before injection are driven from~\eqref{eq:def:S} and given in~\eqref{eq:D2}.
		\begin{figure*}
    		\begin{align}
    		\centering
    		\includegraphics[width=\textwidth,page=2]{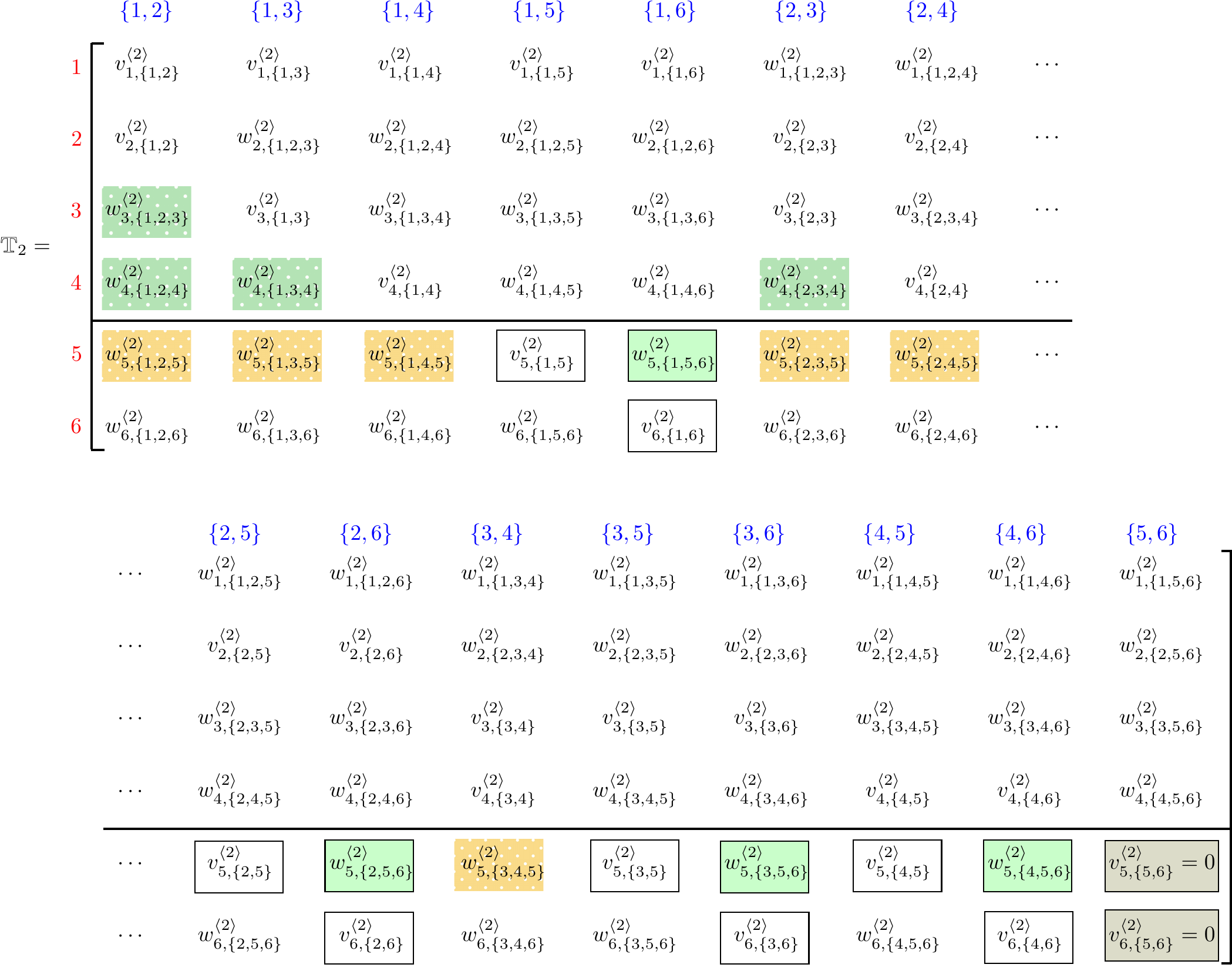},\qquad
    		\label{eq:D2}
    		\end{align} 
    	\end{figure*}
		Again, the horizontal line separates the top $k=4$ rows in $\up{\e{T}_2}$ from the lower $(d-k)=2$ rows in $\down{\e{T}_2}$. Note that symbols in $\down{\e{T}_2}$ are further partitioned into $\dgpi{\e{T}}{2}$, $\ngpi{\e{T}}{2}$, and $\pgpi{\e{T}}{2}$. The two symbols $v_{5,\{5,6\}}^{\langle 2\rangle}$ and $v_{6,\{5,6\}}^{\langle 2\rangle}$ belong to $\ngp{\e{T}_2}$, and hence they are set to zero (see Definition~\ref{def:groups}). Symbols in $\down{\e{T}_2}$ without a surrounding solid frame belong to $\pgpi{\e{T}}{2}$, and therefore do not need any protection. Finally, the remaining symbols in $\dgpi{\e{T}}{2}$ are designated by a box around them, which need to be injected into child matrices of $\e{T}_2$. This will be further discussed in Example~\ref{ex:injection-from-T2}.
		
		Moreover, some of the symbols in $\e{T}_2$ (in both $\up{\e{T}_2}$ and $\down{\e{T}_2}$) are highlighted by a background color without a solid frame around them. This indicates the symbols in $\e{T}_2$ into which injections from $\e{T}_0$ occurs. However, an injection into such designated entries is primary if the host symbol is located in $\up{\e{T}_2}$ (above the horizontal line), and will be secondary if the host symbol lies in $\down{\e{T}_2}$ (below the horizontal line).
		
		We denote this message matrix after injection by $\f{T}_2$ and the injection matrix by $\tOMat{1}{6,\set {6}}{\e{T}_0}$. Therefore, we have  $\f{T}_2 = \e{T}_2 + \tOMat{1}{6,\set {6}}{\e{T}_0}$. The matrix $\tOMat{1}{6,\set {6}}{\e{T}_0}$ can be found from~\eqref{eq:inj:mat} and its complete form is given in~\eqref{eq:O2}.
		\begin{align}
    		\centering
    		\includegraphics[width=\textwidth,page=3]{Figures/example2.pdf}.\qquad
    		\label{eq:O2}
		\end{align} 
		The child matrix $\e{T}_2$ is hosting symbols $\set{v_{6,\set{1,2,3,6}}^{\nd{0}},v_{6,\set{1,2,4,6}}^{\nd{0}},v_{6,\set{1,3,4,6}}^{\nd{0}}, v_{6,\set{2,3,4,6}}^{\nd{0}}}$ as primary injections since they will be added to the parities located in the top $k=4$ rows of $\e{T}_2$. For instance, the injected symbol at position $(4,\set{1,2})$ of $\e{T}_2$ will be $\left[\tOMat{1}{6,\set {6}}{\e{T}_0} \right]_{4,\set{1,2}}= v_{6,\set{1,2,4,6}}^{\nd{0}}$. Moreover, symbols 
		\begin{align*}
		\left\{
		v_{6,\set{1,2,5,6}}^{\nd{0}},v_{6,\set{1,3,5,6}}^{\nd{0}},v_{6,\set{1,4,5,6}}^{\nd{0}},v_{6,\set{2,3,5,6}}^{\nd{0}},v_{6,\set{2,4,5,6}}^{\nd{0}},v_{6,\set{3,4,5,6}}^{\nd{0}}\right\}
		\end{align*} 
		are secondarily injected into $\e{T}_2$ as they will be added to the entries in the bottom rows, i.e., rows indexed by $5$ and $6$. Each missing symbol requires to be primarily injected, regardless of its secondary injection. Therefore, in the following it is shown that these symbols  will be also primarily injected into $\e{T}_5$.
		
		Recall from Table~\ref{table:T0:children} that $\e{T}_5$ is introduced as the child matrix of $\e{T}_0$ associated with the injection pair  $(6,\set{5,6})$. The mode and the signature vector of this child matrix were also evaluated in Examples~\ref{ex:modes} and~\ref{ex:signature} as $m_5=1$ and $\sigma_{\e{T}_5} = (2,2,2,2,2,3)$.  Therefore, the corresponding message matrix has $\binom{d}{m_5}=\binom{6}{1} = 6$ columns. Note that the sign for the entries in the sixth row of $\e{T}_5$ given in \eqref{eq:O5} is negative, which is due to $(-1)^{\sigma_{\e{T}_5}(6)} = (-1)^3=-1$.
 		\begin{align}
    		\centering
    		\includegraphics[width=.8\textwidth,page=4]{Figures/example2.pdf}\qquad
    		\label{eq:D5}
    	\end{align} 
    	Similarly, the parity symbols of $\e{T}_5$ hosting symbols from $\e{T}_0$ are designated by a background color without a solid surrounding frame. We  denote by $\f{T}_5$  the matrix $\e{T}_5$ after injection, and by $\tOMat{1}{6,\set {5,6}}{\e{T}_0}$ the  injection matrix. Thus, we have $\f{T}_5 = \e{T}_5 + \tOMat{1}{6,\set {5,6}}{\e{T}_0}$. The matrix $\tOMat{1}{6,\set {5,6}}{\e{T}_0}$ is also evaluated from~\eqref{eq:inj:mat} and given in~\eqref{eq:O5}.
		\begin{align}
		\centering
		\includegraphics[width=0.7\textwidth,page=5]{Figures/example2.pdf}\qquad
		\label{eq:O5}
		\end{align}
		Note that every entry of $\e{T}_0$ of the form  $\left[\e{T}_{0}\right]_{6,\sen A \cup \set{5,6}}$ is injected into the symbol $w_{\max \sen A,\sen A}$. Here, $\sen A$ can be either of $\set{1,2}$, $\set{1,3}$, $\set{1,4}$, $\set{2,3}$, $\set{2,4}$, or $\set{3,4}$ of $\e{T}_5$. Also, note that unlike $\e{T}_2$ that hosts both primary and secondary injections, all the injections into $\e{T}_5$ are primary, since there is no injection into its $(d-k)=2$ lower rows.
	\end{ex}
	\subsection{Cascading Structure of the Super Message Matrix} 
	\label{subsec:cascade:struct}
	
	Consider a pair of parent and child matrices $(\e{P}, \e{Q})$, where some of the symbols of matrix $\e{P}$ are injected into parity symbols of $\e{Q}$. The symbols in the lower $(d-k)$ rows of $\e{Q}$ should be also protected. Recall that symbols  in $\down{\e{Q}}$ can be partitioned to $\dgp{Q}$, $\ngp{Q}$, and $\pgp{Q}$, where symbols in $\dgp{Q}$ need to be injected into another signed determinant code. To this end, we need to introduce child matrices for $\e{Q}$, into which these symbols will be injected. Note that a child matrix of $\e{Q}$  will be indeed a  grandchild matrix of $\e{P}$.

	\begin{ex}
	\label{ex:injection-from-T2}
		Recall the child matrices $\e{T}_2$ and $\e{T}_5$ discussed in  Example~\ref{ex:modified:child}.  According to Definition~\ref{def:groups}, the symbols in $\down{\e{T}_2}$ can be partitioned to 
			\begin{align}
			\begin{split}
			\dgpi {\e{T}}{2} &= \left\{
			\begin{array}{l}
			v_{5,\set{1,5}}^{\nd{2}}, 
			v_{5,\set{2,5}}^{\nd{2}}, v_{5,\set{3,5}}^{\nd{2}}, 
			v_{5,\set{4,5}}^{\nd{2}}, v_{6,\set{1,6}}^{\nd{2}}, 
			v_{6,\set{2,6}}^{\nd{2}}, v_{6,\set{3,6}}^{\nd{2}}, 
			v_{6,\set{4,6}}^{\nd{2}}\\
			w_{5,\set{1,5,6}}^{\nd{2}}, w_{5,\set{2,5,6}}^{\nd{2}}, w_{5,\set{3,5,6}}^{\nd{2}}, w_{5,\set{4,5,6}}^{\nd{2}}
			\end{array}
			\right\},\\
			\ngpi {\e{T}}{2} &=\left\{v_{5,\set{5,6}}^{\nd{2}}, v_{6,\set{5,6}}^{\nd{2}}\right\},\\
			\pgpi{\e{T}}{2} &=\left\{ \begin{array}{l} 
			w_{5,\set{1,2,5}}^{\nd{2}}, w_{5,\set{1,3,5}}^{\nd{2}}, 
			w_{5,\set{1,4,5}}^{\nd{2}}, w_{5,\set{2,3,5}}^{\nd{2}}, w_{5,\set{2,4,5}}^{\nd{2}} ,w_{5,\set{3,4,5}}^{\nd{2}}, 
			w_{6,\set{1,2,6}}^{\nd{2}}, w_{6,\set{1,3,6}}^{\nd{2}},\\ w_{6,\set{1,4,6}}^{\nd{2}}, w_{6,\set{1,5,6}}^{\nd{2}}, 
			w_{6,\set{2,3,6}}^{\nd{2}}, w_{6,\set{2,4,6}}^{\nd{2}},
			w_{6,\set{2,5,6}}^{\nd{2}}, w_{6,\set{3,4,6}}^{\nd{2}}, 
			w_{6,\set{3,5,6}}^{\nd{2}}, w_{6,\set{4,5,6}}^{\nd{2}}
			\end{array}\right\}.
			\end{split}
			\label{eq:group-T2}
			\end{align}
		The symbols in $\ngpi {\e{T}}{2}$ are set to zero, and those in $\pgpi{\e{T}}{2}$ are parity symbols that do not need protection. However, symbols in $\dgpi {\e{T}}{2}$ need to be injected into the message matrices of some other signed determinant codes. Therefore, $\e{T}_2$ will have its own child nodes. There are three child matrices, namely $\e{T}_9$, $\e{T}_{10}$, and $\e{T}_{11}$, associated with injection pairs $(5,\{5\})$, $(6,\{6\})$, and $(5,\{6\})$, respectively. The modes of these codes are given by  $\md{\e{T}_9}=\md{\e{T}_{10}}=\md{\e{T}_{11}}=2-1-1=0$. In particular,  the message matrix and injection matrix of $\e{T}_{11}$ corresponding to the injection pair $(5,\set{6})$ are given in~\eqref{eq:D11}. Recall that the message matrix of a determinant code of mode zero is an all-zero column vector as it was explained in the definition of message matrix in~\eqref{eq:def:S}.
		\begin{align}
		\includegraphics[height=60mm,page=6]{Figures/example2.pdf}\qquad\qquad
		\includegraphics[height=60mm,page=7]{Figures/example2.pdf}
		\label{eq:D11}
		\end{align} 
		On the other hand, the symbols in the lower part of $\e{T}_5$ can be partitioned to 
			\begin{align*}
			\dgpi {\e{T}}{5}&= \varnothing,\\
			\ngpi {\e{T}}{5} &=\{v_{5,\set{5}}^{\nd{5}}, -v_{6,\set{6}}^{\nd{5}}, w_{5,\set{5,6}}^{\nd{5}}\},\\
			\pgpi{\e{T}}{5} &=\left\{ \begin{array}{l} w_{5,\set{1,5}}^{\nd{5}}, w_{5,\set{2,5}}^{\nd{5}}, w_{5,\set{3,5}}^{\nd{5}}, w_{5,\set{4,5}}^{\nd{5}}, 
			-w_{6,\set{1,6}}^{\nd{5}}, -w_{6,\set{2,6}}^{\nd{5}}, -w_{6,\set{3,6}}^{\nd{5}}, -w_{6,\set{4,6}}^{\nd{5}},-w_{6,\set{5,6}}^{\nd{5}}\end{array}\right\}.
			\end{align*}
			The symbols in $\ngpi {\e{T}}{5}$  are marked in \eqref{eq:D5}, and set to zero. Recall that only symbols in the first group need to be injected, and since $\dgpi {\e{T}}{5}= \varnothing$, no further child matrix is needed for $\e{T}_5$. 
	\end{ex}
	\begin{rem}
	As mentioned before, we may need multiple child matrices for a parent matrix $\e{P}$ to protect all symbols in  $\dgp{P}$. However, it is worth emphasizing that for each child matrix $\e{Q}$, there is only one parent matrix, whose symbols are injected  into it. This leads to an injection hierarchy which can be represented as a tree structure  (see Fig.~\ref{fig:tree}). The ultimate super-message matrix will be obtained by concatenation of all modified (after injection) message matrices. \hfill $\diamond$
	\end{rem}
	 Note that the process of injection will eventually terminate since the mode of a child code is strictly less than that of its parent code, and modes are limited to non-negative integers. 
	\begin{rem}
		Consider a  $w$-symbol in a child matrix $\e{Q}$ into which a symbol from the parent matrix $\e{P}$ is injected. From Remark~\ref{rem:inj:explain}, this symbol either appears in $\up{\e{Q}}$, or it belongs to $\pgp{\e{Q}}$ if it lies in the lower part of $\e{Q}$. Therefore, the host symbol in $\e{Q}$ does not need to be injected into a child matrix of $\e{Q}$ (i.e., a grandchild matrix of $\e{P}$). The implication is that a $w$-symbol in a child matrix hosting a symbol from the parent matrix is never injected into a grandchild matrix, and therefore injected symbols will not be carried for multiple hops of injections. Consequently, the specification of injections form a parent matrix $\e{P}$ into a child matrix $\e{Q}$ does not depend on whether $\e{P}$ is already modified by an earlier injection or not, that is,
		\[
		\tOMat{m}{x, \sen B}{\e{P}} = \tOMat{m}{x, \sen B}{\f{P}}.
		\] 
		\hfill $\diamond$
	\end{rem}
	\begin{ex}
	Continuing with our running example, here we complete the construction of the super message matrix for a cascade code with parameters $(n, k=4, d=6;\mu=4)$. Recall that we started (in Example~\ref{ex:root}) with a root code segment  which is a determinant code $\e{T}_0$ of mode $m_0 = \mu = 4$. This root message matrix has five child matrices that were introduced in Example~\ref{ex:inj:pair}. The child matrices $\e{T}_1,\e{T}_2,$ and $\e{T}_3$ have their own child matrices. For instance, the child matrices of $\e{T}_2$ are discussed in Example~\ref{ex:injection-from-T2}.

	There will be a total of $15$ code segments, namely $\e{T}_0,\e{T}_1,\cdots,\e{T}_{14}$, required to perform all injections. Each code segment $\e{T}_i$ will be modified by injection to obtain $\f{T}_i$, for $i=0,1,\dots, 14$. Fig.~\ref{fig:tree} depicts the tree structure of the hierarchy of injections.  The tree consists of one node for each code segment. The labeled arrows in the figure indicate the injection pair from the parent to the child matrix.
    \begin{figure*}
		\includegraphics[width=\textwidth]{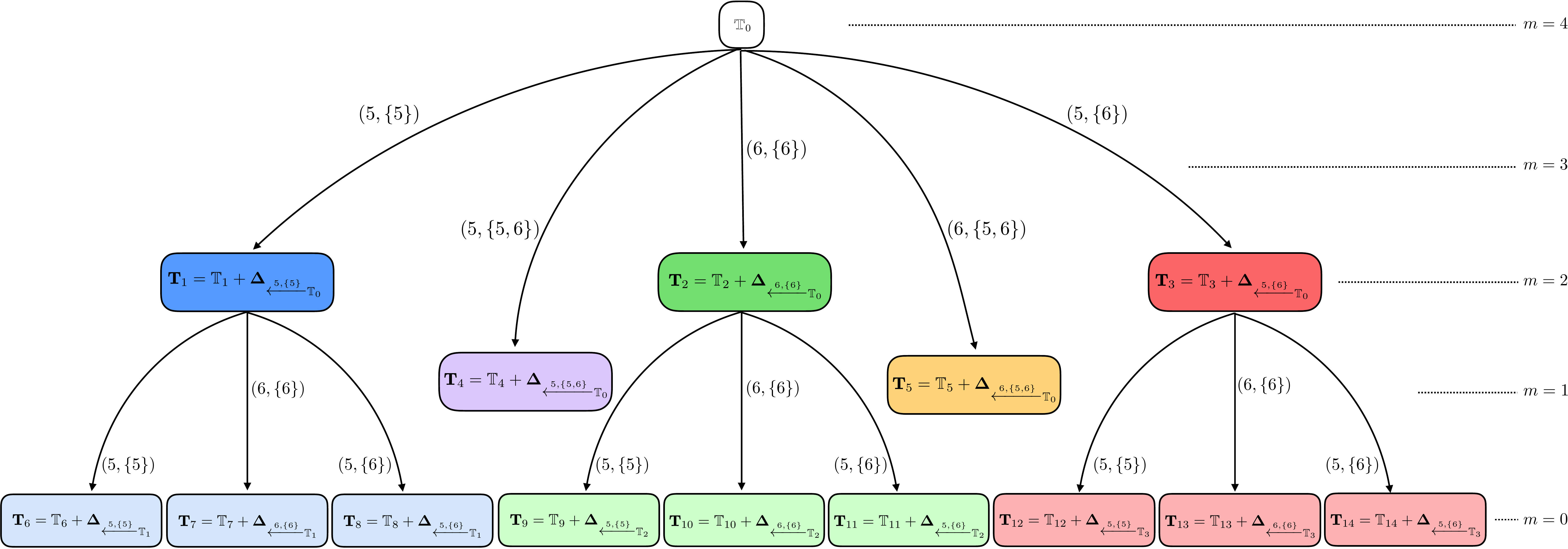}
		\caption{The hierarchical tree for an $(n, k=4, d=6;\mu=4)$ code. Each level on the tree shows code segments with the same mode.  The injection from each parent matrix to its child matrices are shown by arrows, labeled by the corresponding injection pair.} 
		\label{fig:tree}
	\end{figure*}
	The mode of each child matrix is evaluated using ~\eqref{eq:mode:relation},  from the mode of the parent matrix and size of $\sen B$ in the injection pair $(x,\sen B)$. Following this rule, the modes of the code segments are evaluated as
		\begin{align*}
		(m_0,m_1,\dots, m_{14}) = (4,2,2,2,1,1,0,0,0,0,0,0,0,0,0).
		\end{align*}
		where $m_i=\md{\f{T}_i}$.
		
	The signature vector of each child node can be also obtained using the signature vector of the parent matrix and according to~\eqref{eq:inj:sign}. 
	Finally, the super-message matrix of the cascade code, $\MM$, will be formed by concatenation of message matrices $\f{T}_1, \f{T}_2, \dots, \f{T}_{14}$, as shown in Fig.~\ref{fig:cascade-ex}. 
	\begin{figure*}
    		\includegraphics[width=0.9\textwidth,page=2]{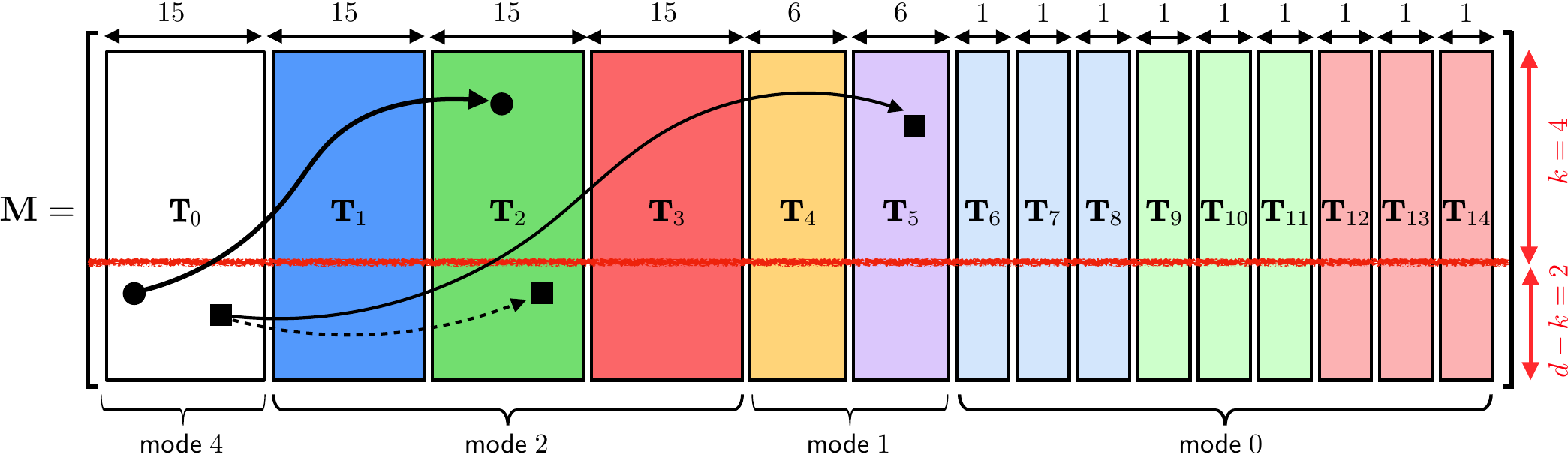}
    		\caption{The super message matrix of the cascade code with system parameters $(n,k=4,d=4; \mu=4)$ and code parameters $(\alpha=81, \beta=27,F=324)$, as a concatenation of code segments $\f{T}_0, \f{T}_1, \dots, \f{T}_{14}$. Some of the injections are depicted by arrows. A solid arrow shows a primary injection and a dashed one indicates a secondary injection. }
    		\label{fig:cascade-ex}
    \end{figure*} 
    	
    Note that this matrix has $d=6$ rows, and $\alpha = \sum_{i=0}^{14} \binom{6}{m_i} = 1\times \binom{6}{4} + 3\times \binom{6}{2} + 2\times \binom{6}{1} + 9 \times \binom{6}{0} = 81$ columns. 	
	\end{ex}
	The code construction presented in this section may look sophisticated at first glance. However, it has a simple and recursive nature, which facilitates its implementation. All stages of the super-message matrix construction are efficiently summarized in Algorithm~\ref{alg:supermsg}, in the next page.
	\begin{algorithm*}
	    \setstretch{1.4}
		\SetKwFunction{detcode}{\FuncName{SgnDetCodeMsgMat}}
		\SetKwFunction{injection}{\FuncName{InjectionMat}}
		\KwIn{Parameters $k,d$ and $\mu$.}
		\KwOut{Super-Message Matrix $(\MM)$ of a cascade code operating at mode $\mu$.}    
	    $\mathsf{UnvisitedNodeCollection} \gets \varnothing$\;
		$\e{T}_{0} \gets \textsc {SgnDetCodeMsgMat}(d,\mu,\mathbf{0}_{1\times d})$\; 
		add $\e{T}_0$ to $\mathsf{UnvisitedNodeCollection}$\;
		$\Delta_{0} \gets \mathbf{0}_{d \times \binom{d}{\mu}}$; $\f{T}_0 \gets \e{T}_0 + \Delta_{0}$; 
		$\MM \gets \f{T}_0$\;
		\While {there exists a node $\e{P}$ in  $\mathsf{UnvisitedNodeCollection}$}{
			\ForEach { $\sen B \subseteq \intv{k+1:d} \mathsf{\;with\;} \size{\sen B} <  \md{\e{P}} $}{
			    $m \gets \md{\e{P}}-\size{\sen B}-1$\;
				\ForEach {$x \in \intv{k+1:d} \mathsf{\;with\;} x \leq \max \sen B$ }{
					\lFor{$i \gets 1$ \KwTo $d$}{$\s{\e{Q}}{i} \gets 1+ \s{\e{P}}{i} + \ind{\sen B \cup \set{i}}{i}$}
					$\eSMat{Q}{m}{x,\sen {B}}{\e{P}} \gets \textsc {SgnDetCodeMsgMat}(k,d,m,\sigma_{\e{Q}} )$\;
					$ \tOMat{m}{x,\sen {B}}{\e{P}} \gets \textsc{InjectionMat}(d, m, x, \sen B,  \e{P})$\;
					$\fSMat{Q}{m}{x,\sen {B}}{\e{P}} \gets \eSMat{Q}{m}{x,\sen {B}}{\e{P}} +  \tOMat{m}{x,\sen {B}}{\e{P}}$\;
					$\MM \gets \left[\MM \;|\; \fSMat{Q}{m}{x,\sen {B}}{\e{P}} \right]$\; 
				add $\eSMat{Q}{m}{x,\sen {B}}{\e{P}}$ to $\mathsf{UnvisitedNodeCollection}$\;
				}
			}
			remove $\e{P}$ from $\mathsf{UnvisitedNodeCollection}$; 
		}
 		\Return $\MM$\;
		\hrule
 	    \SetKwProg{myproc}{Procedure}{:}{\KwRet $\e{D};$}
		\myproc(\Comment*[f]{Create one instance of signed det. code using~\eqref{eq:def:S}}){\detcode{$k,d,m,\sigma_{\e{D}}$}}{
		$\e{D} \gets \mathbf{0}_{d \times \binom{d}{m}}$\;
			\For{$x \gets 1$ \KwTo $d$}{
				\ForEach{$\sen I \subseteq \intv {d} \mathsf{\;with\;} \size{\sen I} = m$ }{				
					\lIf{$x \in \sen I$}{
						$\e{D}_{x,\sen I} \gets (-1)^{\sigma_{\e{D}}(x)} v_{x,\sen I}^{\langle \e{D} \rangle}$
					}
					\lElse
					{
						$\e{D}_{x,\sen I}  \gets (-1)^{\sigma_{\e{D}}(x)} w_{x,\set{x} \cup \sen I}^{\langle \e{D} \rangle }$
					}
					\lIf(\Comment*[f]{Null elements belong to $\ngp{D}$}){$x \leq \max \sen I \mathsf{\;and\;} \sen I \cap \intv{k}= \varnothing$ }{	$\e{D}_{x,\sen I} \gets 0$}
				}
			}
 		}
		\hrule
		\SetKwProg{myproc}{Procedure}{:}{\KwRet $\Delta;$}
		\myproc(\Comment*[f]{Create injection matrix using~\eqref{eq:inj:mat}}){\injection{$d, m, x, \sen B, \e{P}$}}{
		$\Delta \gets \mathbf{0}_{d \times \binom{d}{m}}$\;
			\For{$i \gets 1$ \KwTo $d$}{
				\ForEach{ $\sen I \subseteq \intv {d} \mathsf{\;with\;} \size{\sen I} = m$ }{				
					\lIf{$i > \max \sen I \mathsf{\;and\;} i \notin \sen B \mathsf{\;and\;} \sen I \cap \sen B = \varnothing$}{
						$\Delta_{i,\sen I} \gets (-1)^{1+\s{\e{P}}{i}+\ind{\sen I \cup \set{i} \cup \sen B}{i} } \e{P}_{x, \sen I \cup \set{i} \cup \sen B}$
					}
					\lElse
					{
						$\Delta_{i,\sen I} \gets 0$
					}
				}
			}
		}
		\caption{Construction of Cascade Codes Super-Message Matrix}\label{alg:supermsg}
	\end{algorithm*}
	
	\clearpage
	\section{The Exact Repair Property}
	\label{sec:noderepair}
	In the following, we discuss the exact repair property by introducing the repair data  sent in order to repair a failed node $f\in \intv{n}$ from a set of helper nodes $\sen H \subseteq \intv{n}\setminus \set{f}$ with $ \size{\sen H}=d$.
	
	The repair process is performed in a recursive manner from top-to-bottom, i.e., from segments of the codeword of node $f$  with the highest mode to those with the lowest mode. 
	
	The repair data sent from a helper node $h\in \sen H$ to the failed node $f$ is simply formed by the concatenation of the repair data for each code segment. The repair data for each code segment can be obtained by treating each segment as an ordinary signed determinant code. More precisely,  helper node $h$ sends
	\begin{align*}
	\bigcup_{\f{Q} \textrm{ is a code segment in } \MM}  \left\{\enc_{h,:} \cdot \f{Q} \cdot \repMat{f}{\md{\f{Q}}}\right\},
	\end{align*}
	where the union is taken over all message matrices $\f{Q}$ appearing in the super-message matrix $\MM$, the product $\enc_{h,:} \cdot \f{Q}$ is the codeword segment of node $h$ corresponding to code segment $\f{Q}$, and $\repMat{f}{\md{\f{Q}}}$ is the repair-encoder matrix for node $f$ for a code of $\md{\f{Q}}$, as defined in \eqref{eq:rep:enc}. In other words, for each codeword segment, the helper node needs to multiply this codeword segment by the repair encoder matrix of the proper mode, and send the collection of all such multiplications to the failed node. Recall from Proposition~\ref{prop:beta} that the rank of  $ \repMat{f}{\md{\f{Q}}}$ for a code segment of mode $m=\md{\f{Q}}$ is only $\beta^{(m)} = \binom{d-1}{m-1}$. The total repair bandwidth of the code can be obtained by summing up the repair bandwidth of all the code segments, and is evaluated in \eqref{eq:overall:bndwidth} in Section~\ref{sec:parameters}.
	
	Upon receiving all the repair data $\bigcup_{\f{Q}}  \left\{\enc_{h,:} \cdot \f{Q} \cdot \repMat{f}{\md{\f{Q}}} : h\in \sen H\right\}$, the failed node stacks the segments corresponding to each code segment $\f{Q}$, to obtain 
	\[
	\enc[\sen H,:] \cdot \f{Q} \cdot \repMat{f}{\md{\f{Q}}},
	\]
	from which, the repair spaces of code segment $\f{Q}$ can be retrieved as
	\begin{align}
	\rep{f}{\f{Q}} = \enc[\sen H,:]^{-1} \cdot \enc[\sen H,:] \cdot \f{Q} \cdot \repMat{f}{\md{\f{Q}}} =  \f{Q} \cdot \repMat{f}{\md{\f{Q}}}.
	\label{eq:inj:repspc}
	\end{align}
	Recall that Condition~\ref{cond:P} guarantees that matrix  $\enc[\sen H,:]$ is  invertible. Also note that the above repair space is defined for the code segment \emph{after injection}, while the  one defined in~\eqref{eq:eRepSpace} is for the raw determinant code (before injection). Therefore, the relationship between the two repair spaces is given by
	\begin{align*}
	\rep{f}{\f{Q}} = \f{Q} \cdot \repMat{f}{\md{\f{Q}}} = (\e{Q} + \mathbf{\Delta}) \cdot \repMat{f}{\md{\f{Q}}} =\erep{f}{\e{Q}} + \mathbf{\Delta} \cdot \repMat{f}{\md{\f{Q}}},
	\end{align*}
    Where $\mathbf{\Delta}$ is the injection matrix to code matrix $\e{Q}$ from its parent matrix. Thus, the repair process cannot be  performed as indicated in~\eqref{eq:ndd:repair}. However, having the repair spaces for all the code segments, the content of node $f$ can be reconstructed according to the following proposition. Again, it is worth emphasizing that the codeword segment corresponding to $\f{Q}$ is represented by a row vector whose entries have the same labeling as columns of $\f{Q}$ that are subsets of $\intv{d}$ of size $m=\md{\f{Q}}$.
	
	\begin{prop}
		\label{prop:nkd:repair}	
		In an $(n,k,d)$ cascade code introduced in sections~\ref{sec:code} and~\ref{sec:supermessage}, for every failed node $f\in \intv{n}$ and a set of helpers $\sen H \subseteq \intv {n} \setminus\set{f}$ with $\size{\sen H}=d$,  the content of node $f$  can be exactly regenerated from the received repair spaces.
		More precisely, the symbols at position $\sen I$ in a codeword segment corresponding to a code segment $\f{Q}$ will be repaired  through\footnote{Note that for a code in Galois field $\mathsf{GF}(2^s)$ with characteristic $2$, the repair equation reduces to$\left[\enc_{f,:}\cdot \fSMat{Q}{m}{x,\sen {B}}{\e{P}}\right]_{\sen I} =\sum_{i \in \sen I} [\rep{f}{\f{Q}}]_{i,\sen I\setminus \set{i}}+ 		[\rep{f}{\f{P}}]_{x, \sen I \cup \sen B} \1{ \sen I \cap \sen B =\varnothing} $}
		\begin{align}
		\begin{split}
		\left[\enc_{f,:}\cdot \fSMat{Q}{m}{x,\sen {B}}{\e{P}}\right]_{\sen I} &=\sum_{i \in \sen I} (-1)^{\s{\e{Q}}{i} +\ind{\sen I}{i}} [\rep{f}{\f{Q}}]_{i,\sen I\setminus \set{i}}-\begin{cases}
		[\rep{f}{\f{P}}]_{x, \sen I \cup \sen B} & \mathsf{if\;}   \sen I \cap \sen B =\varnothing, \\
		0& \mathsf{otherwise},
		\end{cases}
		\end{split}
		\label{eq:repair}
		\end{align}
		where $\f{P}$ is the parent matrix of $\f{Q}$, and $(x, \sen B)$ is the corresponding injection pair.
	\end{prop}	
	Note that for positions $\sen I$ that have an overlap with $\sen B$ (the second injection parameter), the repair identity above reduces to \eqref{eq:ndd:repair}, and the repair will be performed as in an ordinary signed determinant code. This is due to the fact that no symbol is injected to any position in column $\sen I$ of message matrix $\f{Q}$, as indicated in~\eqref{eq:inj:mat}, and hence the symbol $\left[\enc_{f,:}\cdot \fSMat{Q}{m}{x,\sen {B}}{\e{P}}\right]_{\sen I}$ is not affected by the injection process. However, for a position $\sen I$ which is disjoint from $\sen B$, the coded symbol at position $\sen I$ of the codeword stored at node $f$ has an interference caused by the injected symbols. However, this interference can be canceled using the repair space of the parent matrix, as indicated in~\eqref{eq:repair}.
	
	In the following, first the repair property is explained for the $(n,k=4,d=6)$ code from our running example, and then the formal proof of Proposition~\ref{prop:nkd:repair} is provided.
	\begin{ex}
		\label{ex:repair}
		In the $(n,k=4,d=6)$ example of the previous section, assume a node $f$ is failed, and its codeword needs to be repaired using the repair data received from a set of helper nodes, say $\sen H$ with $|\sen H|=d=6$. The cascade code, as it is indicated in~Fig.~\ref{fig:tree}, has $15$ segments, which will be repaired in a sequential manner, from $\f{T}_0$ to $\f{T}_{14}$. Each helper node $h\in \sen H$ computes and sends $\bigcup_{i=0}^{14} \{ \enc_{h,:} \cdot \f{T}_i \cdot \repMat{f}{m_i}\}$ to node $f$, where $m_i$ is the mode of code $\f{T}_i$, e.g., $m_0=4$.
		
		The process starts from the first codeword segment (corresponding to the code segment with mode $\mu=m_0=4$, located at the root of the hierarchical tree) of the failed node $f$.
		Note that no symbol is injected into  $\f{T}_0$ and hence we have $\f{T}_0 = \e{T}_0$, and  the repair of the first $\alpha^{(4)} = 15$ symbols of node $f$ is identical to that of  a signed determinant code, as described in Proposition~\ref{prop:ndd:repair}. For the sake of demonstration, we focus on the repair of the symbol at position $\sen I = \seq{1,2,3,6}$ of the codeword segment corresponding to $\f{T}_0$ of the failed node, i.e.,
		\begin{align*}
		\begin{split}
		\Big[\enc_{f,:} \cdot \f{T}_{0}\Big]_{\seq{1,2,3,6}} &= \psi_{f,1}v_{1,\set{1,2,3,6}}^{\nd{0}}+\psi_{f,2}v_{2,\set{1,2,3,6}}^{\nd{0}}+\psi_{f,3}v_{3,\set{1,2,3,6}}^{\nd{0}}\\&\quad+\psi_{f,4}w_{4,\set{1,2,3,4,6}}^{\nd{0}}+\psi_{f,5} w_{5,\set{1,2,3,5,6}}^{\nd{0}}+\psi_{f,6}v_{6,\set{1,2,3,6}}^{\nd{0}}
		\end{split}
		\end{align*}
		Following from~\eqref{eq:ndd:repair}, and since  $\rep{f}{\f{T}_0}=\erep{f}{\e{T}_0}$, for the repair of the  normal signed determinant code at $\sen I=\set{1,2,3,6}$ we have
		\begin{align}
		\begin{split}
		\sum_{i\in \sen I} (-1)^{\s{\e{T}_0}{i}+\ind{\sen I}{i}} \left[\erep{f}{\e{T}_0}\right]_{i,\sen I \setminus \seq{i}}
		&=-\left[\erep{f}{\e{T}_0}\right]_{1,\set{2,3,6}}+\left[\erep{f}{\e{T}_0}\right]_{2,\set{1,3,6}}\\
		&\quad \phantom{=}-\left[\erep{f}{\e{T}_0}\right]_{3,\set{1,2,6}}+\left[\erep{f}{\e{T}_0}\right]_{6,\set{1,2,3}}.
		\end{split}
		\label{eq:ex:rep:t0:3}
		\end{align}
		
		Note that matrix $\erep{f}{\e{T}_0}$ has $\binom{6}{3}=20$ columns, and cannot be fully written here. However. the submatrix of $\erep{f}{\e{T}_0}$ corresponding to the columns with labels in  $\sen Y = \set{\set{1,2,3},\set{1,2,6},\set{1,3,6},\set{2,3,6}}$ is given in~\eqref{eq:RT0}.
		\begin{align}
    		\centering
    		\includegraphics[width=\textwidth]{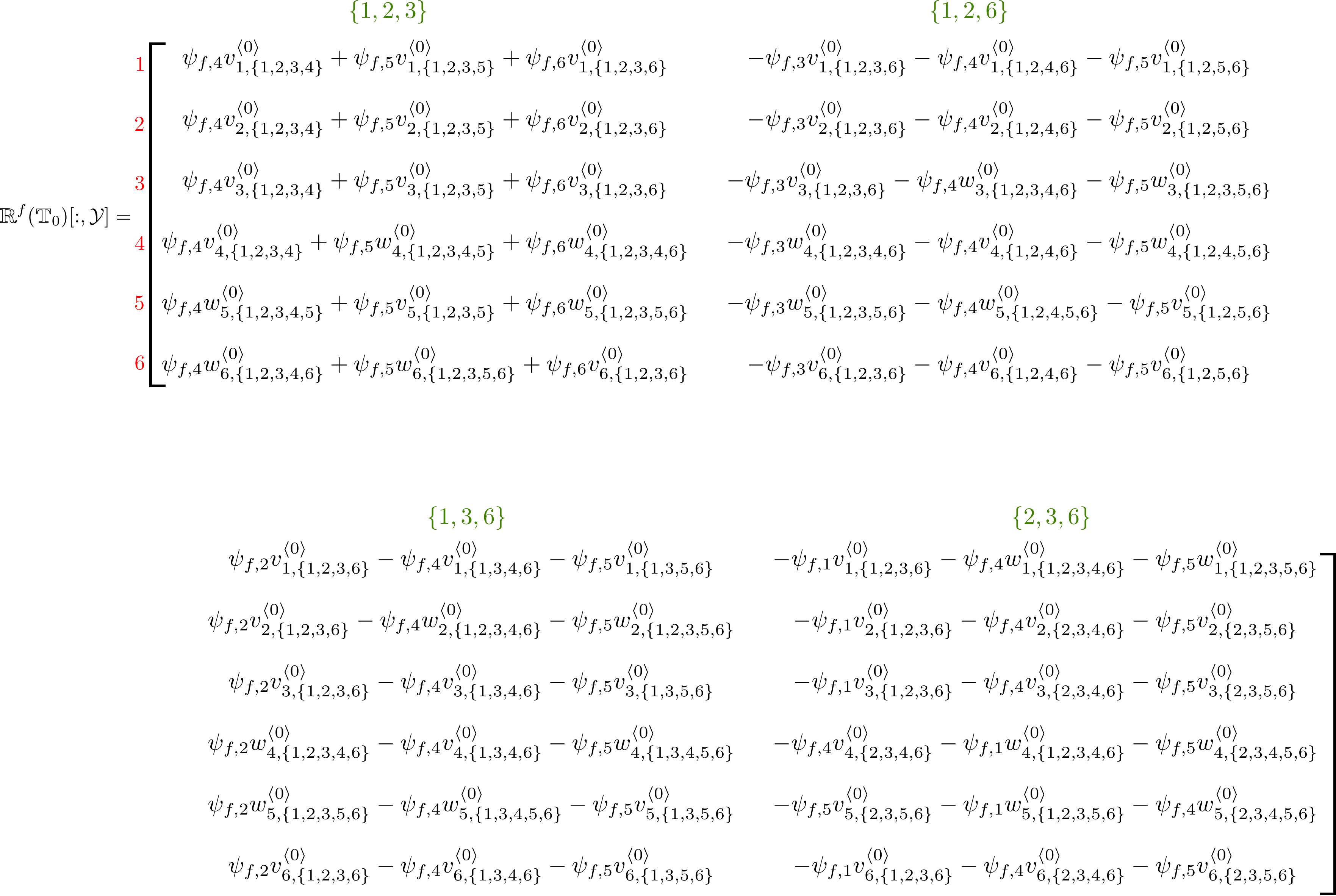}\qquad
    		\label{eq:RT0}
		\end{align}

		Continuing from \eqref{eq:ex:rep:t0:3}, we have
	    \begin{align*}
        	\sum_{i\in \sen I}(-1)^{\s{\f{T}_0}{i}+\ind{\sen I}{i}} \left[\rep{f}{\f{T}_0}\right]_{i,\sen I \setminus \seq{i}}
        	&= \psi _{f,1}v_{1,\{1,2,3,6\}}^{\left<0\right>}+\psi _{f,4} w_{1,\{1,2,3,4,6\}}^{\left<0\right>}+\psi _{f,5} w_{1,\{1,2,3,5,6\}}^{\left<0\right>}\\
        	&\phantom{=}+\psi _{f,2} v_{2,\{1,2,3,6\}}^{\left<0\right>}-\psi _{f,4} w_{2,\{1,2,3,4,6\}}^{\left<0\right>}-\psi _{f,5} w_{2,\{1,2,3,5,6\}}^{\left<0\right>} \\
        	&\phantom{=}+\psi _{f,3}v_{3,\{1,2,3,6\}}^{\left<0\right>}+\psi _{f,4} w_{3,\{1,2,3,4,6\}}^{\left<0\right>}+\psi _{f,5} w_{3,\{1,2,3,5,6\}}^{\left<0\right>}\\
        	&\phantom{=}+\psi _{f,4} w_{6,\{1,2,3,4,6\}}^{\left<0\right>}+\psi _{f,5} w_{6,\{1,2,3,5,6\}}^{\left<0\right>}+ \psi _{f,6} v_{6,\{1,2,3,6\}}^{\left<0\right>}\\
        	&= \psi_{f,1}v_{1,\{1,2,3,6\}}^{\left<0\right>}+\psi _{f,2} v_{2,\{1,2,3,6\}}^{\left<0\right>}+\psi _{f,3}v_{3,\{1,2,3,6\}}^{\left<0\right>}+ \psi _{f,6} v_{6,\{1,2,3,6\}}^{\left<0\right>}\\
        	&\phantom{=}+\psi _{f,4} \left(w_{1,\{1,2,3,4,6\}}^{\left<0\right>}-w_{2,\{1,2,3,4,6\}}^{\left<0\right>}+w_{3,\{1,2,3,4,6\}}^{\left<0\right>}+w_{6,\{1,2,3,4,6\}}^{\left<0\right>}\right)\\
        	&\phantom{=}+\psi _{f,5} \left(w_{1,\{1,2,3,5,6\}}^{\left<0\right>}-w_{2,\{1,2,3,5,6\}}^{\left<0\right>}+w_{3,\{1,2,3,5,6\}}^{\left<0\right>}+w_{6,\{1,2,3,5,6\}}^{\left<0\right>}\right)\\
        	&= \psi_{f,1}v_{1,\{1,2,3,6\}}^{\left<0\right>}+\psi _{f,2} v_{2,\{1,2,3,6\}}^{\left<0\right>}+\psi _{f,3}v_{3,\{1,2,3,6\}}^{\left<0\right>}\\
        	&\phantom{=}+ \psi _{f,6} v_{6,\{1,2,3,6\}}^{\left<0\right>}+\psi _{f,4} w_{4,\{1,2,3,4,6\}}^{\left<0\right>}+\psi _{f,5} w_{5,\{1,2,3,5,6\}}^{\left<0\right>}.
	    \end{align*}

		As it can be seen, the $v$-symbols are repaired directly, while the $w$-symbols are repaired using the parity equation in~\eqref{eq:parityeq}. Other symbols of the code segment corresponding to $\f{T}_0$ can be also repaired in a similar manner. Once the segment corresponding to $\f{T}_0$ is repaired, we can proceed with the codeword segments corresponding to the child matrices in the injection hierarchy.
		
		Let us focus on the repair of the segment corresponding to the child matrix $\f{T}_2$. In particular, we focus on the symbol at position $\sen I=\seq{2,3}$. The missing symbol at this position (which should be regenerated by the repair process) is given by
		\begin{align}
		\Big[\enc_{f,:} \cdot \f{T}_{2}\Big]_{\seq{2,3}} &= \Big[\enc_{f,:} \cdot (\e{T}_{2}+\mathbf{\Delta}_2)\Big]_{\seq{2,3}}  \nonumber
		\\&= \psi _{f,1} w_{1,\{1,2,3\}}^{\nd {2}}+\psi _{f,2} v_{2,\{2,3\}}^{\nd {2}}+\psi _{f,3} v_{3,\{2,3\}}^{\nd {2}}
		+\psi _{f,4}\left( w_{4,\{2,3,4\}}^{\nd {2}}+ v_{6,\{2,3,4,6\}}^{\nd {0}} \right)\nonumber\\
		&\phantom{=}+ \psi _{f,5} \left(w_{5,\{2,3,5\}}^{\nd {2}} + v_{6,\{2,3,5,6\}}^{\nd {0}}\right)
		+\psi _{f,6} w_{6,\{2,3,6\}}^{\nd {2}} .
		\label{eq:ex-ic-1}
		\end{align}
		Note that symbols $v_{6,\{2,3,4,6\}}^{\nd {0}}$ and  $v_{6,\{2,3,5,6\}}^{\nd {0}}$ are injected from $\e{T}_0$ into $\e{T}_2$. 
		Upon receiving the repair symbols, node $f$  recovers $\rep{f}{\f{T}_2}= \f{T}_2  \cdot \repMat{f}{2}$, where  $\repMat{f}{2}$ is defined in~\eqref{eq:rep:enc}. 
		Evaluating~\eqref{eq:repair} for $\sen I=\{2,3\}$, we first obtain
		\begin{align}
		\sum_{x \in \{2,3\}} &(-1)^{\s{\e{T}_2}{x} + \ind{\{2,3\}}{x}}  [\mathbf{R}^{f} (\f{T}_2)]_{x,\{2,3\}\setminus \set{x}} =-\left[\mathbf{R}^{ f} (\f{T}_2 )  \right]_{2,\seq{3}} + \left[\mathbf{R}^{ f}(\f{T}_2)  \right]_{3,\seq{2}}\nonumber\\
		&= -\sum_{\substack{\sen L \subseteq \intv{6} \nonumber\\ \size{\sen L}=2}} \left[\f{T}_2\right]_{2,\sen L} \repMat{f}{2}_{\sen L,\set{3}}+\sum_{\substack{\sen L \subseteq \intv{6} \nonumber\\ \size{\sen L}=2}} \left[\f{T}_2\right]_{3,\sen L} \repMat{f}{2}_{\sen L,\set{2}} \nonumber\\
		&= -\sum_{y\in \{1,2,4,5,6\}} \left[\f{T}_2\right]_{2,\{y,3\}} \repMat{f}{2}_{\{y,3\},\set{3}} \nonumber \\
		&\phantom{+} +\sum_{y\in \{1,3,4,5,6\}} \left[\f{T}_2\right]_{3,\{y,2\}} \repMat{f}{2}_{\{y,2\},\set{2}} \nonumber\\
		&= -\left(-\psi_{f,1} w_{2,\{1,2,3\}}^{\nd {2}}-\psi_{f,2} v_{2,\{2,3\}}^{\nd {2}}+\psi_{f,4} w_{2,\{2,3,4\}}^{\nd {2}}+\psi _{f,5} w_{2,\{2,3,5\}}^{\nd {2}}  +\psi_{f,6} w_{2,\{2,3,6\}}^{\nd {2}}\right)\nonumber\\
		&\phantom{=} 
		+\left(-\psi_{f,1} \left(w_{3,\{1,2,3\}}^{\nd {2}} + v_{6, \{1,2,3,6\}}^{\nd{0}}\right)+\psi_{f,3} v_{3,\{2,3\}}^{\nd {2}} +\psi_{f,4} w_{3,\{2,3,4\}}^{\nd {2}} +\psi _{f,5} w_{3,\{2,3,5\}}^{\nd {2}}  +\psi _{f,6} w_{3,\{2,3,6\}}^{\nd {2}}  \right)\nonumber\\
		&= 	\psi _{f,1} w_{1,\{1,2,3\}}^{\nd {2}}
		+\psi_{f,2} v_{2,\{2,3\}}^{\nd {2}}
		+\psi_{f,3} v_{3,\{2,3\}}^{\nd {2}}+\psi_{f,4} w_{4,\{2,3,4\}}^{\nd {2}}\nonumber\\
		&\phantom{=}
		+\psi _{f,5} w_{5,\{2,3,5\}}^{\nd {2}} 
		+\psi _{f,6} w_{6,\{2,3,6\}}^{\nd {2}} -\psi_{f,1}  v_{6, \{1,2,3,6\}}^{\nd{0}}\label{eq:ex-ic-2}
		\end{align}
		Comparing \eqref{eq:ex-ic-1} and \eqref{eq:ex-ic-2}, we  obtain
		\begin{align}
		\Big[\enc_{f,:} \cdot \f{T}_{2}\Big]_{\seq{2,3}} &- \sum_{x \in \{2,3\}} (-1)^{+\s{\f{T}_2}{x} + \ind{\{2,3\}}{x}}  [\mathbf{R}^{f} (\f{T}_2)]_{x,\{2,3\}\setminus \set{x}} \nonumber\\
		&= \psi_{f,1}  v_{6, \{1,2,3,6\}}^{\nd{0}}
		+\psi _{f,4} v_{6,\{2,3,4,6\}}^{\nd {0}} 
		+ \psi _{f,5}  v_{6,\{2,3,5,6\}}^{\nd {0}}.
		\label{eq:repair-cascade-diff}
		\end{align}
		That means the simple repair strategy of (signed) determinant codes given in \eqref{eq:ndd:repair} cannot be directly applied to repair a this symbol.  However, note that all the three terms in the difference given in~\eqref{eq:repair-cascade-diff} are symbols of the code segment $\e{T}_0$. Also, recall that $\e{T}_0$ is the parent matrix of $\e{T}_2$. The plan is to compute the RHS of \eqref{eq:repair-cascade-diff}  from the repair space of the parent matrix $\f{T}_0$, and subtract it from \eqref{eq:ex-ic-2}, to recover the missing symbol in~\eqref{eq:ex-ic-1}. Recall that the injection pair for $\e{T}_2$ is $(x,\sen B)=(6,\{6\})$.  Interestingly, the entry at position $(x, \sen I \cup \sen B) = (6,\{2,3,6\})$ of  $\rep{f}{\TM_0}$ is given by (see \eqref{eq:RT0})
		\begin{align}
		-[\rep{f}{\TM_0}]_{6,\{2,3,6\}}
		&=-\sum_{\substack{\sen L \subseteq \intv{6} \\ \size{\sen L}=4}} \left[\f{T}_0\right]_{6,\sen L} \repMat{f}{4}_{\sen L,\set{2,3,6}}\nonumber\\
		&=-\sum_{y\in \{1,4,5\}} \left[\f{T}_0\right]_{6,\{y,2,3,6\}} \repMat{f}{4}_{\{y,2,3,6\},\set{2,3,6}}\nonumber\\
		&=\psi _{f,1}v_{6,\{1,2,3,6\}}^{\nd 4}+\psi _{f,4} v_{6,\{2,3,4,6\}}^{\nd 4}+\psi _{f,5} v_{6,\{2,3,5,6\}}^{\nd 4},\label{eq:ex-ic-3}
		\end{align}
		which is exactly identical to the difference  in~\eqref{eq:repair-cascade-diff}. This term $-[\rep{f}{\TM_0}]_{6,\{2,3,6\}}$ is actually the second term in~\eqref{eq:repair} for the case when the intersection of $\sen I$ and $\sen B$ is non-empty. Therefore, we have 
		\begin{align*}
		   \Big[\enc_{f,:} \cdot \f{T}_{2}\Big]_{\seq{2,3}} =  \sum_{x \in \{2,3\}} (-1)^{+\s{\f{T}_2}{x} + \ind{\{2,3\}}{x}}  [\mathbf{R}^{f} (\f{T}_2)]_{x,\{2,3\}\setminus \set{x}} - [\rep{f}{\TM_0}]_{6,\{2,3,6\}}.
		\end{align*}
		 This provides an exact recovery for the failed symbol in \eqref{eq:ex-ic-1}. A similar procedure can be used to repair all the other symbols and codeword segments of the failed node $f$. 
		
	\end{ex}	
	
	\begin{proof}[Proof of Proposition~\ref{prop:nkd:repair}]
		Let $f$ be a failed node, and its content needs to be repaired using the repair data received from the helper nodes in $\sen H$ with $|\sen H|=d$. The content of node $f$ will be reconstructed segment by segment. 
		
		Consider a code segment $\f{Q}=\fSMat{Q}{m}{x,\sen {B}}{\e{P}}$, that is a determinant code with $\md{\f{Q}} = m$ and injection pair $(x,\sen B)$, into which symbols from its parent code segment $\e{P}$, with $\md{\e{P}}=j>m$, are injected. Recall that the corresponding code segment matrix can be written as 
		\begin{align*}
		\fSMat{Q}{m}{x,\sen {B}}{\e{P}} =  \eSMat{Q}{m}{x,\sen {B}}{\e{P}} +  \tOMat{m}{x,\sen {B}}{\e{P}},
		\end{align*}
		where the first term is a signed determinant code and the second term indicates the contribution of injection. For a given position $\sen I$ within this codeword segment,  i.e, $\sen I\subseteq \intv{d}$ and $|\sen I| = m$, the corresponding symbol of the failed node is given by
		\begin{align}
		\left[\enc_{f,:} \cdot \fSMat{Q}{m}{x,\sen {B}}{\e{P}} \right]_{\sen I}= \left[ \enc_{f,:} \cdot \eSMat{Q}{m}{x,\sen {B}}{\e{P}} \right]_{\sen I}+ \left[ \enc_{f,:}\cdot  \tOMat{m}{x,\sen {B}}{\e{P} } 	\right]_{\sen I}.
		\label{eq:prf:repair:1}
		\end{align}
		As it is clear from~\eqref{eq:repair} that the repair of segment $\f{Q}$ of the codeword of node $f$ will be performed similarly to that of the determinant codes using the repair space $\rep{f}{\f{Q}} = \f{Q} \cdot \repMat{f}{m}$, together with additional \emph{correction} from the repair space of the repair of the parent matrix, that is  $\rep{f}{\f{P}}= \f{P} \cdot \repMat{f}{j}$. Note this latter correction will take care of the deviation of the code segment from the standard determinant code, which is caused by the injection of the symbols from the parent matrix $\e{P}$ into $\e{Q}$. 
		
		We start with the first term in the right-hand-side of \eqref{eq:repair}, which is
		\begin{align}
		\sum_{i \in \sen I} &(-1)^{\s{\e{Q}}{i} +\ind{\sen I}{i}} \left[\rep{f}{\f{Q}}\right]_{i,\sen I\setminus \set{i}}=
		\sum_{i \in \sen I} (-1)^{\s{\e{Q}}{i} +\ind{\sen I}{i}} \left[\fSMat{Q}{m}{x,\sen {B}}{\e{P}} \cdot \repMat{f}{m}\right]_{i,\sen I \setminus \set{i}} \nonumber\\
		&=\sum_{i \in \sen I} (-1)^{\s{\e{Q}}{i} +\ind{\sen I}{i}} \left[\eSMat{Q}{m}{x,\sen {B}}{\e{P}} \cdot \repMat{f}{m}\right]_{i,\sen I \setminus \set{i}}+\sum_{i \in \sen I} (-1)^{\s{\e{Q}}{i} +\ind{\sen I}{i}} \left[\tOMat{m}{x,\sen {B}}{\e{P}} \cdot \repMat{f}{m}\right]_{i,\sen I \setminus \set{i}}\label{eq:rep:first:1}\\
		&=\sum_{i \in \sen I} (-1)^{\s{\e{Q}}{i} +\ind{\sen I}{i}} \left[\erep{f}{\e{Q}}\right]_{i,\sen I \setminus \set{i}}+\sum_{i \in \sen I} (-1)^{\s{\e{Q}}{i} +\ind{\sen I}{i}} \left[\tOMat{m}{x,\sen {B}}{\e{P}} \cdot \repMat{f}{m}\right]_{i,\sen I \setminus \set{i}}\nonumber\\
		&=\left[\enc_{f,:} \cdot \eSMat{Q}{m}{x,\sen {B}}{\e{P}}\right]_{\sen I} +\sum_{i \in \sen I}  (-1)^{\s{\e{Q}}{i} +\ind{\sen I}{i}} \left[\tOMat{m}{x,\sen {B}}{\e{P}}\cdot \repMat{f}{m}\right]_{i,\sen I\setminus \set{i}},\label{eq:rep:first:3}
		\end{align}
		where~\eqref{eq:rep:first:1} holds due to the linearity of the operations, and 	
		in~\eqref{eq:rep:first:3} we used Proposition~\ref{prop:ndd:repair} for the repair process of $\eSMat{Q}{m}{x,\sen {B}}{\e{P}}$, which is a  signed determinant code with parameters $(d,m)$ and the  signature vector $\sigma_{\e{Q}}$.  Therefore, from~\eqref{eq:prf:repair:1} and~\eqref{eq:rep:first:3}, we can conclude that proving the claimed identity in~\eqref{eq:repair} is equivalent to show
		\begin{align}
		\term_1 - \term_2 =  \term_3,
		\label{eq:repair:leftover}
		\end{align}
		where
		\begin{align}
		\term_1 & = \left[\enc_{f,:} \cdot \tOMat{m}{x,\sen {B}}{\e{P}}\right]_{\sen I}\nonumber\\
		\term_2 &= \sum_{i \in \sen I} (-1)^{\s{\e{Q}}{i} +\ind{\sen I}{i}} \left[\tOMat{m}{x,\sen {B}}{\e{P}}\cdot \repMat{f}{m}\right]_{i,\sen I \setminus \set{i}} \nonumber\\
		\term_3 &=
		- \left[\rep{f}{\PM}\right]_{x, \sen I \cup \sen B}
		\mathbbm{1}\left\{ \sen I \cap \sen B = \varnothing\right\}
		\label{eq:rep:second:1}
		\end{align}
		Note that all the data symbols appearing in~\eqref{eq:repair:leftover} belong the parent matrix $\e{P}$. We can distinguish the following two cases in order to prove~\eqref{eq:repair:leftover}.

		\noindent \textbf{Case I: $\sen I \cap \sen B = \varnothing$}: 
		Starting from $\term_1$ we have	
		\begin{align}
		\term_1 &= \left[\enc_{f,:}\cdot \tOMat{m}{x,\sen {B}}{\e{P}}\right]_{\sen I}\nonumber\\ &=\sum_{y\in \intv {d}} \psi_{f,y}\left[\tOMat{m}{x,\sen {B}}{\e{P}}\right]_{y,\sen I}\nonumber\\
		&= \sum_{y \in [\max \sen I+1:d] \setminus \sen B }\psi_{f,y} \left[\tOMat{m}		{x,\sen {B}}{\e{P} }\right]_{y,\sen I} \label{eq:rep:third:1} \\
		&= \sum_{y \in [\max \sen I+1:d] \setminus (\sen I \cup \sen B) }\psi_{f,y} 		\left[\tOMat{m}{x,\sen {B}}{\e{P} }\right]_{y,\sen I} \label{eq:rep:third:1-2} \\
		&= \hspace{-10mm} \sum_{y \in [\max \sen I+1:d] \setminus (\sen I \cup \sen B) } \hspace{-10mm}
		(-1)^{1+\s{\e{P}}{y} + \ind{\sen I \cup \set{y} \cup \sen B}{y}  } 
		\psi_{f,y} \e{P}_{x, \sen I \cup \set{y} \cup \sen B},
		\label{eq:rep:third:2} 
		\end{align}
		where \eqref{eq:rep:third:1} follows the definition of the injection symbols in \eqref{eq:inj:mat} which implies a non-zero injection occurs at position $(y,\sen I)$ only if $y>\max \sen I$ and $y\notin \sen B$; \eqref{eq:rep:third:1-2} holds since $\intv{\max \sen I +1 :d} \cap \sen I = \varnothing$; and in  \eqref{eq:rep:third:2} we plugged in the entries of $\tOMat{m}{x,\sen {B}}{\e{P} }$  from \eqref{eq:inj:mat}.
	    Next, $\term_2$ can be expanded as
  	 	\begin{align}
        	\term_2 &= \sum_{i \in \sen I} (-1)^{\s{\e{Q}}{i} +\ind{\sen I}{i}} \left[\tOMat{m}{x,\sen {B}}{\e{P}}\repMat{f}{m} \right]_{i,\sen I \setminus \set{i}} \nonumber\\
        	&=\sum_{i \in \sen I} (-1)^{\s{\e{Q}}{i} +\ind{\sen I}{i}} \sum_{\substack{\sen L\subseteq \intv{d} \\ \size{\sen L}=m}}\left[\tOMat{m}{x,\sen {B}}{\e{P}}\right]_{i, \sen L} \left[\repMat{f}{m} \right]_{\sen L,\sen I \setminus \set{i}}\nonumber\\
        	&=\sum_{i \in \sen I} (-1)^{\s{\e{Q}}{i} +\ind{\sen I}{i}} \sum_{\substack{y\in \intv{d} \\ y\notin \sen I\setminus\{i\}}}  \left[\tOMat{m}{x,\sen {B}}{\e{P}}\right]_{i, (\sen I\setminus\{i\}) \cup \{y\}} \left[\repMat{f}{m} \right]_{(\sen I\setminus\{i\}) \cup \{y\},\sen I \setminus \set{i}} \label{eq:rep:fourth:3}\\
        	&= \sum_{i =\max \sen I} (-1)^{\s{\e{Q}}{i} +\ind{\sen I}{i}} \sum_{\substack{y<i \\ y\notin \sen I \cup \sen B}}  \left[\tOMat{m}{x,\sen {B}}{\e{P}}\right]_{i, (\sen I\setminus\{i\}) \cup \{y\}} \left[\repMat{f}{m} \right]_{(\sen I\setminus\{i\}) \cup \{y\},\sen I \setminus \set{i}} \label{eq:rep:fourth:4} \\
        	&=(-1)^{\s{\e{Q}}{\max \sen I} + \ind{\sen I}{\max \sen I}} \hspace{-8mm}\sum_{y\in \intv{\max{\sen I}}\setminus (\sen I \cup \sen B)}   \left[\tOMat{m}{x,\sen {B}}{\e{P}}\right]_{\max \sen I ,\sen I\cup \set{y} \setminus \set{\max \sen I}} \left[\repMat{f}{m}\right]_{ \sen I\cup \set{y} \setminus \set{\max \sen I} ,\sen I \setminus \set{\max \sen I}} \label{eq:rep:fourth:4-2} \\
        	&=(-1)^{\s{\e{Q}}{\max \sen I} + \ind{\sen I}{\max \sen I}} \!\!\!\!\sum_{y\in \intv{\max{\sen I}}\setminus (\sen I \cup \sen B)} \Bigg\{ \left[(-1)^{1+\s{\e{P}}{\max \sen I} + \ind{\sen I \cup \{y\} \cup \sen B}{\max \sen I}} \e{P}_{x, \sen I \cup \set{y} \cup \sen B}\right] \nonumber\\
        	&\hspace{70mm} 
        	\cdot 	 \left[(-1)^{\s{\e{Q}}{y} + \ind{(\sen I\setminus\{\max \sen I\}) \cup \set{y}}{y}}\psi_{f,y}\right]\Bigg\}. \label{eq:rep:fourth:5} 
		\end{align}
		Note that in~\eqref{eq:rep:fourth:3} we have used the definition of matrix $\repMat{f}{m}$ given in~\eqref{eq:rep:enc}, that implies the entry in position $(\sen L, \sen I\setminus\{i\})$ is non-zero only if $\sen L= (\sen I\setminus\{i\}) \cup \{y\}$ for some $y\notin \sen I\setminus\{i\}$. Moreover, \eqref{eq:rep:fourth:4} follows from the definition of injected entries in \eqref{eq:inj:mat}, which implies the entry of $\tOMat{m}{x,\sen {B}}{\e{P}}$ at position $(i, (\sen I\setminus\{i\}) \cup \{y\})$ is non-zero only if all the following conditions hold:
		\begin{align*}
		\left\{ \hspace{-2mm}
		\begin{array}{ll}
		i > \max  \Big[(\sen I \setminus \set{i} )\cup \set{y} \Big]  &\hspace{-2mm}\Rightarrow  i>y \text{ and } i > \max  (\sen I \setminus \set{i} ), \\
		\left((\sen I \setminus \set{i} )\cup \set{y} \right) \cap \sen B= \varnothing &\hspace{-2mm} \Rightarrow y\notin \sen B.
		\end{array}
		\right.
		\end{align*} 
		These together with the fact that the outer summation is taken over $i\in \sen I$ imply $i=\max \sen I$. Moreover, the inner summation over  $y \in \intv{d},y\notin \sen I \setminus\{i\}$ reduces to a summation over $y$'s satisfying $y<i=\max \sen I$ and $y\notin \sen I \cup \sen B$, or simply
		$y\in [\max \sen I]\setminus (\sen I \cup \sen B)$ as indicated in~\eqref{eq:rep:fourth:4}. Finally, in~\eqref{eq:rep:fourth:5} the matrix entries are replaced from their definitions from \eqref{eq:rep:enc} and \eqref{eq:inj:mat}. 
		
		Next, we simplify the overall sign in \eqref{eq:rep:fourth:5}.  First, for every $y<\max \sen I$ we have
        \begin{align}
		&\s{\e{Q}}{\max \sen I} \hspace{-0.8mm}+\hspace{-0.8mm} \ind{\sen I}{\max \sen I} \hspace{-0.8mm}+\hspace{-0.8mm}  \s{\e{P}}{\max \sen I} \hspace{-0.8mm}+\hspace{-0.8mm} \ind{\sen I \cup \{y\} \cup \sen B}{\max \sen I} \nonumber\\
		&=  \left[1\hspace{-0.8mm}+\hspace{-0.8mm} \s{\e{P}}{\max \sen I} \hspace{-0.8mm}+\hspace{-0.8mm} \ind{\sen B \cup \set{\max \sen I}}{\max \sen I}\right] \hspace{-0.8mm}+\hspace{-0.8mm} \ind{\sen I}{\max \sen I} \hspace{-0.8mm}+\hspace{-0.8mm}  \s{\e{P}}{\max \sen I} \hspace{-0.8mm}+\hspace{-0.8mm} \ind{\sen I \cup \{y\} \cup \sen B}{\max \sen I} \label{eq:rep:sign:first:A}\\
		&\equiv 1\hspace{-0.8mm}+\hspace{-0.8mm} \ind{\sen B \cup \set{\max \sen I}}{\max \sen I} \hspace{-0.8mm}+\hspace{-0.8mm} \ind{\sen I}{\max \sen I}  \hspace{-0.8mm}+\hspace{-0.8mm} \ind{\sen I \cup \{y\} \cup \sen B}{\max \sen I}\qquad \textrm{(mod $2$)} \nonumber\\
		&= 1\hspace{-0.8mm}+\hspace{-0.8mm} \ind{\sen B}{\max \sen I} \hspace{-0.8mm}+\hspace{-0.8mm} \ind{\set{\max \sen I}}{\max \sen I} \hspace{-0.8mm}+\hspace{-0.8mm} \ind{\sen I}{\max \sen I}  \hspace{-0.8mm}+\hspace{-0.8mm} \ind{\sen I}{\max \sen I}\hspace{-0.8mm}+\hspace{-0.8mm} \ind{\{y\}}{\max \sen I}\hspace{-0.8mm}+\hspace{-0.8mm} \ind{\sen B}{\max \sen I} \label{eq:rep:sign:first:B} \\
		&\equiv 1 \qquad \textrm{(mod $2$)}, \label{eq:rep:sign:first:C}
		\end{align}
		where~\eqref{eq:rep:sign:first:A} is due to the definition of the child matrix's signature in \eqref{eq:inj:sign}; equality in \eqref{eq:rep:sign:first:B} holds since $\sen I,\sen B$, and $\set{y}$ are disjoint sets; in  \eqref{eq:rep:sign:first:C} we used the fact that $y< \max \sen I$. Similarly, we can write
        \begin{align}
		\s{\e{Q}}{y}+  \ind{(\sen I \setminus \set{\max \sen I}) \cup \set{y}}{y}  &= 
		\Big[ 1+ 
		\s{\e{P}}{y} + \ind{\set{y} \cup \sen B}{y}\Big]  + \ind{(\sen I \setminus \set{\max \sen I}) \cup \set{y}}{y}\label{eq:rep:sign:second:A}\\
		&  = 
		1+ 
		\s{\e{P}}{y} + \ind{\set{y} \cup \sen B}{y} + \ind{\sen I \cup \set{y}}{y}\label{eq:rep:sign:second:B}\\
		&  = 
		1+ 
		\s{\e{P}}{y} + \ind{\set{y}}{y} + \ind{\sen I \cup \set{y} \cup \sen B}{y} \label{eq:rep:sign:second:C}\\
		&\equiv \s{\e{P}}{y} + \ind{\sen I \cup \set{y} \cup \sen B}{y} \qquad \textrm{(mod $2$)},\label{eq:rep:sign:second:D}
		\end{align}
		where in \eqref{eq:rep:sign:second:A} we used the definition of the child matrix's signature vector in \eqref{eq:inj:sign}; equality in \eqref{eq:rep:sign:second:B} follows the fact that $y<\max \sen I$; and \eqref{eq:rep:sign:second:C} holds since $\sen I$, $\set{y}$, and $\sen B $ are disjoint sets. 
		
		Plugging \eqref{eq:rep:sign:first:C} and \eqref{eq:rep:sign:second:D} into \eqref{eq:rep:fourth:5} we get 
		\begin{align}
		\term_2 
		&=
		- \sum_{\mathclap{y\in \intv{\max{\sen I}}\setminus (\sen I \cup \sen B)}}
		(-1)^{1+ \s{\PM}{y} + \ind{ \sen I \cup \set{y} \cup \sen B}{y} } \psi_{f,y} \e{P}_{x, \sen I \cup \set{y} \cup \sen B}. \label{eq:rep:fifth:6}
		\end{align} 	
		
		Lastly, since $\md{\f{P}}=j$, we have
		\begin{align}
		\term_3 = - \left[\rep{f}{\f{P}}\right]_{x, \sen I \cup \sen B} 
		& = -\left[\f{P}\cdot \repMat{f}{j}\right]_{x, \sen I \cup \sen B}\nonumber\\
		&=\hspace{-1pt}  -\sum_{\substack{\sen L \subseteq \intv {d} \\ \size{\sen{L}}=j}} \f{P}_{x,\sen L}  \cdot [\repMat{f}{j}]_{\sen L,\sen I \cup \sen B} \nonumber\\
		&= - \hspace{-10pt}\sum_{y \in \intv{d} \setminus (\sen I \cup \sen B)} \f{P}_{x, \sen I  \cup \set{y} \cup \sen B}  \cdot \left[\repMat{f}{j}\right]_{\sen I \cup \set{y} \cup \sen B ,\sen I \cup \sen B} \label{eq:prf:repair:2}\\
		&= \hspace{-10pt} \sum_{y \in \intv{d} \setminus (\sen I \cup \sen B)} \hspace{-5pt} (-1)^{1+\s{\e{P}}{y} +\ind{\sen I \cup \set{y} \cup \sen B}{y}} \psi_{f,y} \f{P}_{x, \sen I \cup \set{y} \cup \sen B} \label{eq:prf:repair:3}\\
		&= \hspace{-10pt} \sum_{y \in \intv{d} \setminus (\sen I \cup \sen B)}\hspace{-2mm} \hspace{-5pt} (-1)^{1+\s{\e{P}}{y} +\ind{\sen I \cup \set{y} \cup \sen B}{y}} \psi_{f,y} \e{P}_{x, \sen I \cup \set{y} \cup \sen B},
		\label{eq:fP=eP}
		\end{align}
		where in~\eqref{eq:prf:repair:2} we used the definition of matrix $\repMat{f}{j}$ in \eqref{eq:rep:enc} that implies the entry in position $(\sen L, \sen I\cup \sen B)$ is non-zero only if $\sen L= \sen I \cup \set{y} \cup \sen B$ for some $y\notin \sen I\cup \sen B$, in~\eqref{eq:prf:repair:3} we plugged in the entry at position $({\sen I \cup \set{y} \cup \sen B ,\sen I \cup \sen B})$ of  $\repMat{f}{j}$ using~\eqref{eq:rep:enc}. Moreover, \eqref{eq:fP=eP} holds since no injection occurs at position $(x, \sen I \cup \set{y} \cup \sen B)$ of matrix $\e{P}$. To see this, note that $(x, \sen B)$ is a valid injection pair from parent $\f{P}$ to the child $\f{Q}$ This means that by Condition~\ref{cond:IP-iv} of Remark~\ref{rem:injpair}, $(x,\sen B)$ should satisfy $x\leq \max \sen B$. On the other hand, if $\f{P}$ is hosting any injected symbol at position $(x, \sen I \cup \set{y} \cup \sen B)$ from its own parent, based on~\eqref{eq:inj:mat}, the symbol $\e{P}_{x, \sen I \cup \set{y} \cup \sen B}$ must be a parity symbol and the relation $x > \max \Big[\sen I \cup \set{y} \cup \sen B \Big]$ should hold. This yields the fact that $x > \max \Big[\sen I \cup \set{y} \cup \sen B \Big] > \max \sen B$ which is in contradiction with $x \leq \max \sen B$. Therefore, $\e{P}_{x, \sen I \cup \set{y} \cup \sen B}$ is not hosting any injected symbol.
		Putting  \eqref{eq:rep:third:2}, \eqref{eq:rep:fifth:6}, and \eqref{eq:fP=eP} together, we get
		\begin{align*}
		\term_1 - \term_2 &= 
		\sum_{y \in [\max \sen I+1:d] \setminus (\sen I \cup \sen B) }\hspace{-5mm}
		(-1)^{1+\s{\e{P}}{y} + \ind{\sen I \cup \set{y} \cup \sen B}{y}  }
		\psi_{f,y} \e{P}_{x, \sen I \cup \set{y} \cup \sen B}\nonumber\\
		&\phantom{=}
		\ \ \ +\!\!\!
		\sum_{y\in \intv{\max{\sen I}}\setminus (\sen I \cup \sen B)} \hspace{-5mm}
		(-1)^{1+\s{\e{P}}{y} + \ind{\sen I \cup \set{y} \cup \sen B}{y}  } 
		\psi_{f,y} \e{P}_{x, \sen I \cup \set{y} \cup \sen B} \nonumber\\
		&= \sum_{y \in [d] \setminus (\sen I \cup \sen B) } 
		(-1)^{1+\s{\e{P}}{y} + \ind{\sen I \cup \set{y} \cup \sen B}{y}  } 
		\psi_{f,y} \e{P}_{x, \sen I \cup \set{y} \cup \sen B} = \term_3,
		\end{align*}
		which is the desired identity in \eqref{eq:repair:leftover}. 
		
		\noindent \textbf{Case II: $\sen I \cap \sen B \neq \varnothing$}: 
		Similar to case I, we can expand $\term_1$ to get the summation in~\eqref{eq:rep:third:1}. Then each term in \eqref{eq:rep:third:1} consists of an entry from $\tOMat{m}{x,\sen {B}}{\e{P}}$ at position $(y, \sen I)$, which is zero for  $\sen I \cap \sen B \neq \varnothing$. Therefore we have  $\term_1=0$. 
		
		Similarly, $\term_2$ can be expanded to the summation given in~\eqref{eq:rep:fourth:3}. However, based on ~\eqref{eq:inj:mat}, the entry of $\tOMat{m}{x,\sen {B}}{\e{P}}$ at position $(i, (\sen I\setminus\{i\}) \cup \{y\})$ is non-zero only if $\left((\sen I\setminus\{i\}) \cup \{y\}\right) \cap \sen B= \varnothing$ and $i\notin \sen B$. This implies
			\begin{align*}
			0 = \left|\left((\sen I \setminus \{i\}) \cup \set{y}\right) \cap \sen B\right| &\geq \left|(\sen I \setminus \{i\})  \cap \sen B\right|= \left|\sen I \cap (\sen B \setminus \{i\})\right|  = \left| \sen I \cap \sen B\right| \geq 1. 
			\end{align*}
		This contradiction implies that all the terms in \eqref{eq:rep:fourth:3} are zero, and thus $\term_2=0$. Finally, $\term_3$ is zero whenever $\sen I \cap \sen B \neq \varnothing$ as it is defined in~\eqref{eq:rep:second:1}. Therefore, the 
		identity \eqref{eq:repair:leftover} clearly holds. 
		This completes the proof of Proposition~\ref{prop:nkd:repair}. 
	\end{proof}

	\section{The Data Recovery Property}
	\label{sec:datarec}
	In this section, the data recovery property of the proposed code is discussed. This property guarantees the availability of the storage system in spite of up to $n-k$ failures, meaning that the original data file can be recovered from any subset of $k$ nodes among $n$ nodes. The following proposition is a formal statement for this property.
	\begin{prop}
		Consider an arbitrary subset $\sen K \subseteq \intv{n}$ with $\size{\sen K} = k$. In a distributed storage system with an $(n,k,d;\mu)$ cascade code and parameters defined in~\eqref{eq:params}, all of the $F(k,d;\mu)$ information symbols can be recovered from the coded data stored in the nodes indexed by $i \in \sen K$.
		\label{prop:nkd:recovery}
	\end{prop}
	
	We provide an algorithmic proof for the above proposition. The general description of the data recovery algorithm is discussed in Subsection~\ref{subsec:data:rec:desc}. Then, the full data recovery process is presented in Algorithm~\ref{alg:data:rec}. Next, in Subsection~\ref{subsec:data:rec:proof} the formal proof of Proposition~\ref{prop:nkd:recovery} is presented which also explains how each part of the data recovery algorithm functions. Finally, in Subsection~\ref{subsec:data:rec:exmp} the details of data recovery are explained for the running example of this paper.
	
	\subsection{The Data Recovery Algorithm}
	\label{subsec:data:rec:desc}
	Consider an arbitrary subset of nodes $\sen K\subseteq  \intv {n}$ with $|\sen K| = k$. Let $\encdc$ denote the $k \times d$ sub-matrix formed by collecting the corresponding $k$ rows of $\enc$. The collection of coded symbols  stored in these $k$ nodes can be written as $\encdc \cdot \MM$. The goal is to recover the original super message matrix $\MM$ from the observed data $\encdc \cdot \MM$.  Unlike the $(n,k=d,d)$ signed determinant codes, here the encoder matrix $\pdc$ is not square, and hence is not invertible. So, one cannot simply multiply the stacked data by the inverse of the encoder matrix to retrieve $\MM$ from $\pdc \cdot \MM $. Nevertheless, using the properties of the encoder matrix in Definition~\ref{def:encoder}, the matrix $\pdc$ can be decomposed into  
	\begin{align*}
	    \pdc = \left[ \gdc \Big| \ddc \right],
	\end{align*} 
	where $\gdc$ is a $k\times k$ invertible matrix. 
	Recall that the super message matrix $\MM$ is formed by the concatenation of several (after injection) message matrices. Therefore, the matrix $\encdc \cdot \MM$ consists of several codeword segments, each of the form $\encdc \cdot \fSMat{\seg}{m}{y,\sen {Y}}{\e{P}}$, where $\fSeg=\fSMat{\seg}{m}{y,\sen {Y}}{\e{P}}$ is a segment of $\MM$ operating at mode $m=\md{\fSeg}$, $\e{P}$ is the message matrix of the parent code of $\e{S}$, and $(y,\sen Y)$ is the injection pair used to generate $\f{S}$. The matrix $\fSeg$ is a $d\times \alpha_m = d\times \binom{d}{m}$ matrix and can be partitioned into $\up{\fSeg}$ and $\down{\fSeg}$, corresponding to the top $k$ and bottom $(d-k)$ rows of $\fSeg$, respectively. Therefore, after multiplying $\encdc \cdot \MM$ by $\gdc^{-1}$, for codeword segment corresponding to $\fSeg$ we get
    \begin{align}
        \gdc^{-1} \cdot \left( \encdc \cdot\fSeg \right) &= \gdc^{-1} \cdot \left[\gdc \Big|\ddc \right]\cdot\fSeg \nonumber\\ &=
	    \left[\mathbf{I}_{k \times k}\Big|\gdc^{-1} \cdot  \ddc \right]\cdot\fSeg \nonumber  \\
	    &=  \left[\mathbf{I}_{k \times k}\Big|\gdc^{-1} \cdot  \ddc \right]\cdot \left[\begin{array}{c}\up{\fSeg} \nonumber \\ 
	    \hline \down{\fSeg}\end{array}\right] \nonumber \\
	    &=  \up{\fSeg}+\gdc^{-1} \cdot  \ddc \cdot \down{\fSeg}  \label{eq:proc:rec:B}
	\end{align}
	In this section, we explain how we first recover the bottom matrix $\down{\fSeg}$ separately. Then we compute $\gdc^{-1} \cdot  \ddc \cdot \down{\fSeg}$ and use it to get a copy of $\up{\fSeg}$ from \eqref{eq:proc:rec:B}. Recall that the goal of the symbol injection introduced in this paper is to provide data recovery. Therefore, some symbols of the bottom matrix $\down{\fSeg}$ may be retrieved from child matrices of $\f{S}$. This will impose an order for the data recovery, where we start by recovering the data symbols of the segments at the lowest level of the hierarchical tree (with the smallest mode) and proceed to segments with higher modes, until we reach to the code segment at the root of the hierarchical tree. Hence, data recovery is considered as a \emph{bottom-to-top} process. 
	\begin{ex}
    	In the $(n,k=4,d=6;\mu=4)$ example of Section~\ref{sec:supermessage}, consider the data recovery from an arbitrary subset of $k=4$ nodes say $\sen K=\set{1,3,6,7}$. In this example, the data collector observes an encoded matrix $\encdc  \hspace{-1pt}\cdot  \hspace{-1pt} \MM  \hspace{-1pt}= \hspace{-1pt} \begin{bmatrix} \encdc \cdot \f{T}_0 \hspace{-1pt}&\hspace{-1pt} \encdc \cdot \f{T}_1 \hspace{-1pt}&\hspace{-1pt} \cdots  \hspace{-1pt}&\hspace{-1pt}  \encdc \hspace{-1pt}\cdot\hspace{-1pt} \f{T}_{14} \end{bmatrix}$. The goal of data recovery is to extract a complete copy of \begin{align*}
    	    \MM=\begin{bmatrix} \f{T}_1 & \f{T}_2 & \cdots  &  \f{T}_{14} \end{bmatrix}.
    	\end{align*} We recover segments starting from $\f{T}_{14}$ and finish the decoding at $\f{T}_{0}$. Now, for instance in the data recovery of coded segment $\encdc \cdot \f{T}_2$ and for column $\sen I = \set{2,3}$ the observed data by the data collector is given by (see \eqref{eq:D2} and \eqref{eq:O2})
        \vspace{15pt}\begin{align*}
             \encdc \cdot \left[\f{T}_2\right]_{:,\set{2,3}}= 
             \vcenter{\vspace{-15pt}\includegraphics[width=0.5\linewidth]{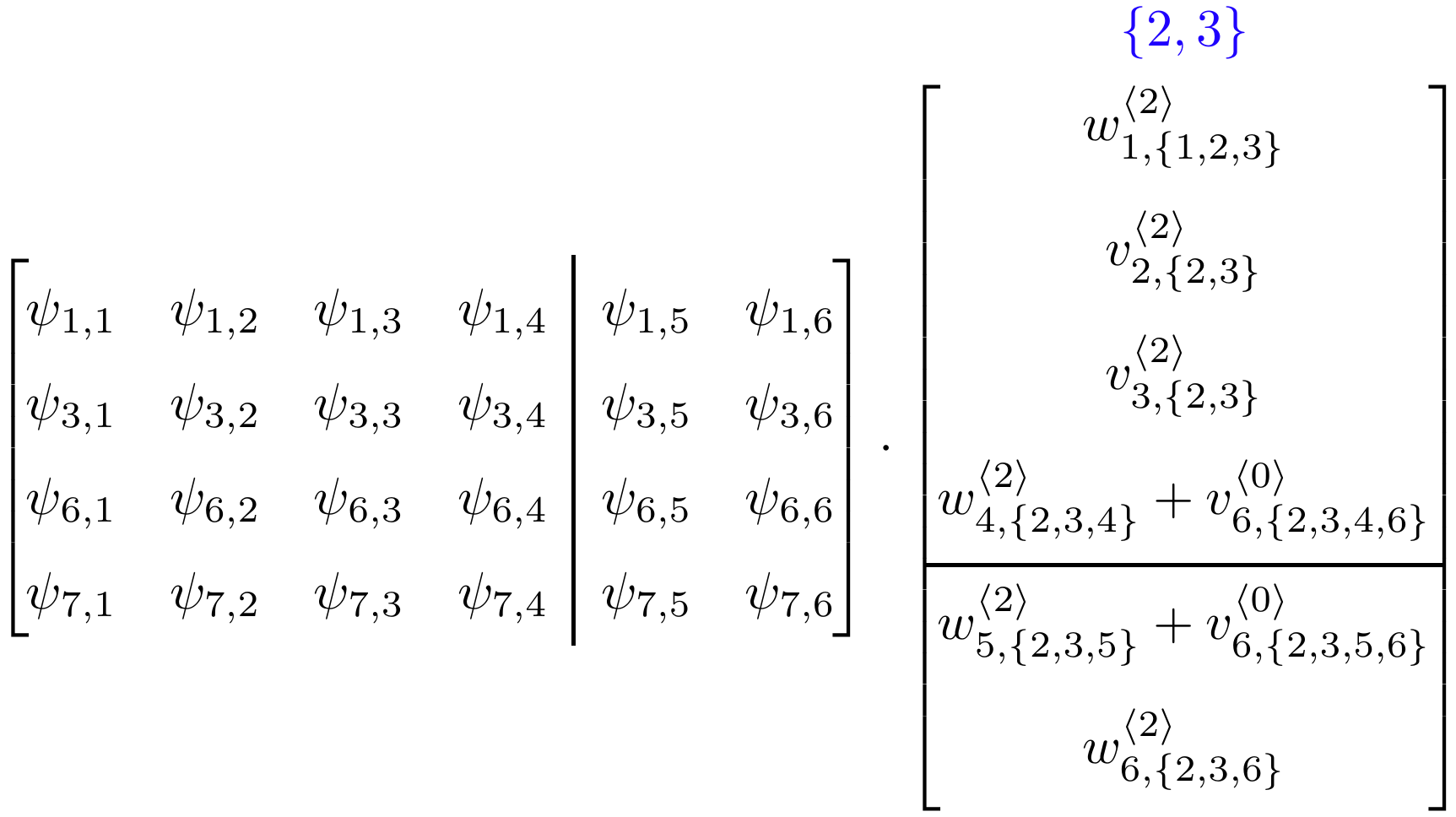}}
        \end{align*}
        In the above equation, $\left[\f{T}_2\right]_{:,\set{2,3}}$ cannot be extracted from $\encdc \cdot \left[\f{T}_2\right]_{:,\set{2,3}}$, since $\encdc$ is not square/invertible. However, we will explain how the symbols $\set{w_{5,\set{2,3,5}}^{\nd{2}},v_{6,\set{2,3,5,6}}^{\nd{0}},w_{6,\set{2,3,6}}^{\nd{2}}}$  can be recovered from another piece of collected data, and then the top rows of $\left[\f{T}_2\right]_{:,\set{2,3}}$ can be computed using \eqref{eq:proc:rec:B}, that is,
        \begin{align*}               \left[\up{\f{T}}_2\right]_{:,\set{2,3}}&=\gdc^{-1} \cdot \left( \encdc \cdot \left[\f{T}_2\right]_{:,\set{2,3}} \right)  -\gdc^{-1} \cdot  \ddc \cdot \left[\down{\f{T}}_2\right]_{:,\set{2,3}}.    
        \end{align*}
 	\end{ex}
	The major steps of data recovery are described below and the full process is presented in Algorithm~\ref{alg:data:rec}.
	\begin{enumerate}[label=\bf{(St.\arabic*)}, ref=\bf{(St.\arabic*)}, leftmargin=9.5mm ]
		\item The segments are decoded from segments with mode $0$ to those with higher modes, until the algorithm reaches the root of the tree (Loops $1$ and $2$ in Algorithm~\ref{alg:data:rec}).\label{instr:I1}
		\item Within each mode and each code segment $\fSMat{\seg}{m}{y,\sen {Y}}{\e{P}}$, the columns are decoded according to the reverse-lexicographical order (\textsc{\decode} procedure in Algorithm~\ref{alg:data:rec}). \label{instr:I2}
		\item For each column of $\fSMat{\seg}{m}{y,\sen {Y}}{\e{P}}$ labeled by $\sen I$,  entries located in the lower $(d-k)$ rows are decoded first \\(\textsc{\column} procedure in Algorithm~\ref{alg:data:rec}). To do so, recall that a given symbol $\left[\fSMat{\seg}{m}{y,\sen {Y}}{\e{P}}\right]_{x,\sen I}$ is made from the original and injected part, i.e.,
		\begin{align*}
		    \left[\fSMat{\seg}{m}{y,\sen {Y}}{\e{P}}\right]_{x,\sen I}= \left[\eSMat{\seg}{m}{y,\sen {Y}}{\e{P}}\right]_{x,\sen I}+ \left[\tOMat{m}{y,\sen {Y}}{\e{P}}\right]_{x,\sen I}.
		\end{align*}
		 The  decoding strategy adopted for each symbol $\eSeg_{x,\sen I}=\left[\eSMat{\seg}{m}{y,\sen {Y}}{\e{P}}\right]_{x,\sen I}$ depends on the group that the symbol belongs to (see Definition~\ref{def:groups}): \label{instr:I3}
		\begin{enumerate}[label= \bf{(St.\arabic{enumi}.\arabic*)},ref = \bf{(St.\arabic{enumi}.\arabic*)}, leftmargin=13mm]
		    \item If $\eSeg_{x,\sen I} \in \dgp{\eSeg} $ then $\fSeg_{x,\sen I}=\eSeg_{x,\sen I}$, and $\left[\tOMat{m}{y,\sen {Y}}{\e{P}}\right]_{x,\sen I}=0$. This symbol will be decoded using the child matrices of $\f{S}$ who have lower modes, and hence are already decoded before $\fSeg$. \label{instr:I3:1}
			\item If $\eSeg_{x,\sen I} \in \ngp{\eSeg} $ then $\eSeg_{x,\sen I}=0$. Moreover, no symbol from $\e{P}$ is injected into this entry, i.e.,   $\left[\tOMat{m}{y,\sen {Y}}{\e{P}}\right]_{x,\sen I}=0$. Therefore, we have $  
			\fSeg_{x,\sen I}= \eSeg_{x,\sen I} =0$. \label{instr:I3:2}
			\item If $\eSeg_{x,\sen I} \in \pgp{\eSeg}$ then the original part $\eSeg_{x,\sen I}$ and the (potentially) injected part $\left[\tOMat{m}{y,\sen {Y}}{\e{P}}\right]_{x,\sen I}$ will be decoded separately. To this end, the symbol $\e{S}_{x,\sen I}$ will be recovered using the parity equation in~\eqref{eq:parityeq}. We will show that the other symbols participating in the parity equation appear in columns $\fSeg_{:,\sen J}$, with $\sen J \succ \sen I$ (according to the lexicographical order defined in Subsection~\ref{subsec:notation}). Therefore, based on Step~\ref{instr:I2}, all such symbols are already decoded, before column $\fSeg_{:,\sen I}$. Next, note that the injection of symbol $\left[\tOMat{m}{y,\sen {Y}}{\e{P}}\right]_{x,\sen I}$ ( a possible injection from $\e{P}$)  is a secondary one. Hence, such a symbol is also primarily injected to a sibling code segment of $\e{S}$, say $\e{Q}$ (another child code of $\e{P}$), with $\md{\f{Q}} < \md{\f{S}}$. Thus, the injected symbol is already decoded when the message matrix $\f{Q}$ is decoded.\label{instr:I3:3}
    	\end{enumerate}
		\item Once the lower part of column $\sen I$ of code segment $\f{S}$ (i.e., $\down{\fSeg}_{:,\sen I}$) is decoded, its upper part $\up{\fSeg}_{:,\sen I}$ will be recovered using the identity  in~\eqref{eq:proc:rec:B} (Line \ref{alg:upper} of Algorithm~\ref{alg:data:rec}). \label{instr:I4}
		\item Once codes segment $\fSMat{\seg}{m}{y,\sen {Y}}{\e{P}} = \eSMat{\seg}{m}{y,\sen {Y}}{\e{P}} + \tOMat{m}{y,\sen {Y}}{\e{P}}$ is decoded, both $\eSMat{\seg}{m}{y,\sen {Y}}{\e{P}}$ and $\tOMat{m}{y,\sen {Y}}{\e{P}}$ matrices will be extracted (\textsc{\extract} procedure of Algorithm~\ref{alg:data:rec}). The matrix $\tOMat{m}{y,\sen {Y}}{\e{P}}$ will be later used in the recovery of the parent matrix, $\e{P}$.  \label{instr:I5}
	\end{enumerate}	
	In Algorithm~\ref{alg:data:rec} of the next page, it is assumed that the structure of the code, including the hierarchical tree which specifies parent-child relationship between nodes and all the injection pairs are known at the data collector. Note that the procedure may look sophisticated at first glance. However, it follows a recursive and identical routine for all code segments.
	
	In the following subsection, we first present the formal proof for the data recovery, which also provides a detailed description of the above instruction. This is followed by applying the data recovery algorithm to our running example in Section~\ref{subsec:data:rec:exmp}.  
    \begin{algorithm*}
        	\setstretch{1.4}
			\KwIn{Stacked contents of $k$ nodes $\sen K\subseteq \intv{n}$ in the form of $\enc[\sen K,:]\cdot \MM$.}
			\KwOut{Recovered data file (entries of the super-message matrix $\MM$).}  
			$\MM \gets [\;]$\;
			\For (\Comment*[f]{Loop $1$})  {$m\gets 0$  to $\mu$} {
				\ForEach(\Comment*[f]{Loop $2$}) {$
				\fSMat{\seg}{m}{y,\sen {Y}}{\e{P}}$ in $
				\MM$ with $\md{\fSMat{\seg}{m}{y,\sen {Y}}{\e{P}}}=m$}{
				   	$\fSMat{\seg}{m}{y,\sen {Y}}{\e{P}} \gets \textsc {DecodeSegment}(\enc[\sen K,:]\cdot \fSMat{\seg}{m}{y,\sen {Y}}{\e{P}})$\; 
					\vspace{1mm}$\MM \gets \left[\begin{array}{c|c}\fSMat{\seg}{m}{y,\sen {Y}}{\e{P}} & \MM \end{array}\right]$\;
					$\eSMat{\seg}{m}{y,\sen {Y}}{\e{P}},\tOMat{T}{y,\sen {Y}}{\e{P} } \gets \textsc {\extract}(\fSMat{\seg}{m}{y,\sen {Y}}{\e{P}})$ \label{alg:extract} \Comment*{Globally store $\eSMat{\seg}{m}{y,\sen {Y}}{\e{P}}$ and $\tOMat{T}{y,\sen {Y}}{\e{P} }$}
				}	
			}
			\Return $\MM$;
		\hrule
		\SetKwProg{myproc}{Procedure}{:}{\KwRet $\fSMat{\seg}{m}{y,\sen {Y}}{\e{P}}$}
 		\myproc{\decode{$\enc[\sen K,:]\cdot\fSMat{\seg}{m}{y,\sen {Y}}{\e{P}}$}}{
		    \ForEach  { $\sen I \subseteq \intv{d}$ with $\size{\sen I} = m$ in the reverse lexicographical order}{
		        $\left[\,\down{\fSeg}\,\right]_{:,\sen I} \gets \textsc {DecodeColumnBottom}(\enc[\sen K,:]\cdot \fSMat{\seg}{m}{y,\sen {Y}}{\e{P}},\sen I)$\;
			    $ \up{\fSeg}_{:,\sen I} \gets \gdc^{-1} \cdot \left[\enc[\sen K,:]\cdot \fSeg \right]_{:,\sen I} -\gdc^{-1} \cdot \enc[\sen K,:]\cdot \left[\,\down{\fSeg}\,\right]_{:,\sen I}$ \label{alg:upper} \Comment*{Decode upper part using~\eqref{eq:proc:rec:B}} 
			}
			$\fSMat{\seg}{m}{y,\sen {Y}}{\e{P}} \gets \begin{bmatrix} \up{\fSeg} \\[-1mm] \down{\fSeg} \end{bmatrix}$\;
 		}
 		\hrule
 		\SetKwProg{myproc}{Procedure}{:}{\KwRet $\left[\,\down{\fSeg}\,\right]_{:,\sen I}$}
		\myproc{\column{$\enc[\sen K,:]\cdot \fSMat{\seg}{m}{y,\sen {Y}}{\e{P}},\sen I$}}{
		 $\sen A \gets \sen I \cap \intv{k}$; $\sen B \gets \sen I \cap \intv{k+1:d}$ \label{alg:AB}\;
		 \For{$x \gets k+1 \;\KwTo\; d$}{
		            \uIf (\Comment*[f]{Symbol belonging to $\dgp{S}$}){$x \leq \max \sen B \mathsf{\;and\;} \sen A \neq \varnothing$ \label{alg:dgp}}{$\left[\,\down{\fSeg}\,\right]_{x,\sen I} \gets (-1)^{1+\s{\eSeg}{\max \sen A} + \ind{\sen A \cup \sen B}{\max {\sen A}}} \left[\tOMat{m}{x,\sen {B}}{\eSeg} \right]_{\max {\sen A},\sen A \setminus \set{\max \sen{A}}}\hspace{-20mm}$ \label{alg:dgp:A} \Comment*{Get injected symbol using~\eqref{eq:primary:inj}}
    			   }
    			   \uElseIf (\Comment*[f]{Symbol belonging to $\ngp{S}$}){$x \leq \max \sen B \mathsf{\;and\;} \sen A = \varnothing$}{$\left[\,\down{\fSeg}\,\right]_{x,\sen I} \gets 0$ \label{alg:ngp}\;}
    			   \ElseIf (\Comment*[f]{Symbol belonging to $\pgp{S}$}){$x > \max \sen B$}{
    			        $\left[\,\down{\fSeg}\,\right]_{x,\sen I} \gets (-1)^{\sigma_{\e{S}}(x)+ m} \sum_{t\in \sen I} (-1)^{\sigma_{\e{S}}(t)+\ind{\sen I \cup \{x\}}{t}} \left[\fSMat{\seg}{m}{y,\sen {Y}}{\e{P}} \right]_{t,(\sen I \cup \{x\})\setminus \{t\}}$ \label{alg:pgp}\;
    			       \If (\Comment*[f]{ A symbol is secondarily injected into $\left[\,\down{\fSeg}\,\right]_{x,\sen I}$})
    			        {$x\notin \sen Y \mathsf{\;and\;} \sen I \cap \sen Y=\varnothing \mathsf{\;and\;} m>0 $ \label{alg:parent:inj}}
    			        {$\sen {B}' \gets \sen B \cup \set{x}\cup\sen Y$ \label{alg:pgp:A}\; 
    			         $\e{P}_{y,\sen I \cup \set{x} \cup \sen Y}\gets (-1)^{1+\s{\e{P}}{\max \sen A} + \ind{\sen A \cup \sen {B}'}{\max {\sen A}}} \left[\tOMat{m}{y,\sen {B}'}{\e{P}} \right]_{\max \sen A,\sen {A} \setminus \set{\max \sen A}}$ \label{alg:pgp:B} \Comment*{Get injected symbol}
    			        $\left[\tOMat{m}{y,\sen Y}{\e{P}}\right]_{x,\sen I}\gets (-1)^{1+\s{\e{P}}{x} + \ind{\sen I \cup \set{x} \cup \sen Y}{x}} \e{P}_{y,\sen I \cup \set{x} \cup \sen Y}$ \;\label{alg:pgp:C} 
    			        $\left[\,\down{\fSeg}\,\right]_{x,\sen I} \gets \left[\,\down{\fSeg}\,\right]_{x,\sen I} +\left[\tOMat{m}{y,\sen Y}{\e{P}}\right]_{x,\sen I}$ \label{alg:pgp:D}\;}
    			   }
			    }
			  }
			\caption{Data Recovery Algorithm}\label{alg:data:rec}
		\end{algorithm*}	
    \begin{algocontinue}{\ref{alg:data:rec}}
    	    \setstretch{1.4}
    	    \setcounter{AlgoLine}{35}
        	\SetKwProg{myproc}{Procedure}{:}{\KwRet $\eSMat{\seg}{m}{y,\sen {Y}}{\e{P}},\tOMat{T}{y,\sen {Y}}{\e{P}}$}
    		\myproc{\extract{$\fSMat{\seg}{m}{y,\sen {Y}}{\e{P}}$}}{
    		  $m \gets \md{\fSMat{\seg}{m}{y,\sen {Y}}{\e{P}}}$\;     
        		\For{$i\gets 1$ \KwTo $d$}{\label{alg:deltas:loop1}
        		    \For{$\sen I \subseteq \intv {d+1}$ with $\size{\sen I}=m$}{\label{alg:deltas:loop2}
        		      \uIf{$ i > \max \sen I \mathsf{\;and\;} i \notin \sen Y  \mathsf{\;and\;} \sen I \cap \sen Y = \varnothing$ \label{alg:deltas:cond}}{
        		            $\left[\eSMat{\seg}{m}{y,\sen {Y}}{\e{P}}\right]_{i,\sen I} \gets (-1)^{\s{\eSeg}{i}+m}\sum_{t \in \sen I} (-1)^{\s{\e{\eSeg}}{t}+\ind{\sen I}{t}} \left[\fSMat{\seg}{m}{y,\sen {Y}}{\e{P}}\right]_{t,(\sen I \cup\{i\}) \setminus\{t\}}$ \label{alg:extract:cond:A}
        		        }\Else {	$\left[\eSMat{\seg}{m}{y,\sen {Y}}{\e{P}}\right]_{i,\sen I} \gets \left[\fSMat{\seg}{m}{y,\sen {Y}}{\e{P}}\right]_{i,\sen I} $ \label{alg:extract:cond:B}}
        	    $\left[\tOMat{T}{y,\sen {Y}}{\e{P}}\right]_{i,\sen I} \gets \left[ \fSMat{\seg}{m}{y,\sen {Y}}{\e{P}}\right]_{i,\sen I} - \left[\eSMat{\seg}{m}{y,\sen {Y}}{\e{P}}\right]_{i,\sen I}$ \label{alg:extract:delta}\;
        		    }
        		}
    		}
 		\caption{Data Recovery Algorithm (continuation)}
	\end{algocontinue}

    \subsection{The Proof of Data Recovery}
    \label{subsec:data:rec:proof}
        This section is dedicated to the proof of Proposition~\ref{prop:nkd:recovery}. 
        Before diving into the formal proof, we present two lemmas and their proofs, which play an important role in the proof of the proposition. 
        \begin{lm}[Initial Recovery Step]
    		In any $(n,k,d;\mu)$ cascade code we have $\down{\f{S}}=\mathbf{0}_{(d-k)\times 1}$ for every code segment $\f{S}$ with $\md{\fSeg}=0$.
    		\label{lm:bottom-mode0}
    \end{lm}
    
    \begin{proof}[Proof of Lemma~\ref{lm:bottom-mode0}]
    Consider a segment $\fSeg$ with $\md{\fSeg}=0$, which is introduced by a parent segment $\e{P}$ via an injection pair $(y,\sen Y)$, i.e, $\fSMat{\seg}{m}{y,\sen {Y}}{\e{P}} = \eSMat{\seg}{m}{y,\sen {Y}}{\e{P}} + \tOMat{m}{y,\sen {Y}}{\e{P}}$. Recall from Remark~\ref{rem:injpair} that   $y\in[k+1:d]$ and $\sen Y  \subseteq [k+1:d]$. Also note that the only column of $\eSeg$ is indexed by $\varnothing$, which is an all-zero vector, i.e., $\eSMat{\seg}{m}{y,\sen {Y}}{\e{P}} = \mathbf{0}_{d\times 1}$. According to \eqref{eq:inj:mat}, the  symbol of parent matrix $\e{P}$ to be injected into position $(i, \varnothing)$ of $\e{S}$ is $\e{P}_{y, \varnothing \cup \{i\} \cup \sen Y} = \e{P}_{y, \{i\} \cup \sen Y} $ (up to a sign). Now, if $\left[\eSMat{\seg}{m}{y,\sen {Y}}{\e{P}}\right]_{i,\varnothing}$ belongs to the lower part of $\fSeg$, we have $i\in \intv{k+1:d}$. This together with $\sen Y \subseteq \intv{k+1:d}$ implies that $(\{i\}\cup \sen Y) \cap [k]=\varnothing$, and thus $\e{P}_{y, \{i\} \cup \sen Y} \in \ngp{P}$ (see Definition~\ref{def:groups}). Hence, this symbol is set to zero, and will not be injected. Thus, no injections occur at position $(i,\varnothing)$ with $i\in \intv{k+1:d}$, and we have $\f{S}_{i,\varnothing} = \e{S}_{i,\varnothing} = 0$. This implies the claim of the lemma. 
 
	\end{proof}
	
    	\begin{lm} [Separability of the Injected and Original Symbols]
		If a code segment $\fSeg = \fSMat{\seg}{m}{y,\sen {Y}}{\e{P}} = \eSMat{\seg}{m}{y,\sen {Y}}{\e{P}} + \tOMat{m}{y,\sen {Y}}{\e{P}}$ 
		is decoded (all entries of $\fSeg$ are recovered), then the symbols of the original message matrix  $\eSMat{\seg}{m}{y,\sen {Y}}{\e{P}}$ and the injected symbols in $\tOMat{m}{y,\sen {Y}}{\e{P}}$ can be uniquely extracted. 
		\label{lm:decode:separate}
	\end{lm}	
	\begin{proof}[Proof of Lemma~\ref{lm:decode:separate}]
			First note that if no injection occurs into the position $(i, \sen I)$, then we have $\left[\eSMat{\seg}{m}{y,\sen {Y}}{\e{P}}\right]_{i,\sen I}=\left[\fSMat{\seg}{m}{y,\sen {Y}}{\e{P}}\right]_{i,\sen I}$, and $\left[\tOMat{m}{y,\sen {Y}}{\e{P}}\right]_{i,\sen I}=0$, and the claim clearly holds. 
			
			Next, consider a position $(i,\sen I)$ with an injection. 
			Recall from~\eqref{eq:inj:mat} and Remark~\ref{rem:inj:explain}, that injection takes place into a position $(i,\sen I)$ of $\eSeg$ only if  conditions $i>\max \sen I$, $i \notin \sen Y$ and $\sen I \cap \sen Y = \varnothing$ are all satisfied. Let $\md{\f{S}}=m$. We start by obtaining a copy of  $\left[\eSMat{\seg}{m}{y,\sen {Y}}{\e{P}}\right]_{i,\sen I}$ as follows. 
		    \begin{align}
            \left[\eSMat{\seg}{m}{y,\sen Y}{\e{P}}\right]_{i,\sen I} 
            &= (-1)^{\s{\eSeg}{i}} w_{i, \sen I \cup \set{i}} \nonumber \\
                             &= (-1)^{\s{\eSeg}{i}} \left[(-1)^{1+\ind{\sen I \cup \set{i}}{i}} \sum_{t \in \sen I} (-1)^{\ind{\sen I \cup \set{i}}{t}} w_{t,\sen I \cup \set{i}}\right] \label{eq:parity:rec:B}\\
                             &= (-1)^{\s{\eSeg}{i}+m}\sum_{t \in \sen I} (-1)^{\ind{\sen I}{t}} w_{t,\sen I \cup \set{i}}\label{eq:parity:rec:C}\\
                             &= (-1)^{\s{\eSeg}{i}+m}\sum_{t \in \sen I} (-1)^{\s{\e{\eSeg}}{t}+\ind{\sen I}{t}} \left[\eSMat{\seg}{m}{y,\sen {Y}}{\e{P}}\right]_{t,(\sen I \cup\{i\}) \setminus\{t\}} \nonumber \\
                             &= (-1)^{\s{\eSeg}{i}+m}\sum_{t \in \sen I} (-1)^{\s{\e{\eSeg}}{t}+\ind{\sen I}{t}} \left[\fSMat{\seg}{\size{\sen{I}}}{y,\sen {Y}}{\e{P}}\right]_{t,(\sen I \cup\{i\}) \setminus\{t\}},\label{eq:datarec:pgp}
            \end{align}
			where in~\eqref{eq:parity:rec:B} we used parity equation in~\eqref{eq:parityeq} to rewrite $w_{i, \sen I \cup \set{i}}$ based on other symbols of the $w$-group $\sen I \cup \set{i}$; in~\eqref{eq:parity:rec:C} we used the fact that $i$ is the maximum of the $m+1$ element set $\sen I \cup \set{i}$ to conclude that $\ind{\sen I \cup \set{i}}{i}=m+1$ and $\ind{\sen I \cup \set{i}}{t}=\ind{\sen I }{t}$; the last equality in~\eqref{eq:datarec:pgp} follows from the fact that $t \leq \max \sen I < i = \max ((\sen I \cup \{i\}) \setminus \{t\})$ and hence, according to~\eqref{eq:inj:mat} no injection takes place into $\left[\eSMat{\seg}{m}{y,\sen {Y}}{\e{P}}\right]_{t,(\sen I \cup\{i\})}$. This implies that $\left[\eSMat{\seg}{m}{y,\sen Y}{\e{P}}\right]_{i,\sen I}$ can be retrieved as a linear combination of some symbols in $\fSMat{\seg}{m}{y,\sen Y}{\e{P}}$. 
			
			Finally, once $\left[\eSMat{\seg}{m}{y,\sen {Y}}{\e{P}}\right]_{i,\sen I}$ is recovered, we can find the injected symbol from 
			\begin{align*}
			\left[ \tOMat{m}{y,\sen Y}{\e{P}}\right]_{i,\sen I} = \left[\fSMat{\seg}{m}{y,\sen {Y}}{\e{P}}\right]_{i,\sen I} - \left[\eSMat{\seg}{m}{y,\sen {Y}}{\e{P}}\right]_{i,\sen I}.
			\end{align*}
			This completes the proof.
	\end{proof}
    	\begin{proof}[Proof of Proposition~\ref{prop:nkd:recovery}]  
    	We aim to show that all the code segments in $\MM$ can be recovered using the data recovery process in  Algorithm~\ref{alg:data:rec}. The data recovery is a recursive process,  
        which starts from the code segments with the lowest modes (at the bottom of the hierarchical tree) and continues to code segments with the highest modes (at the top of the hierarchical tree). Furthermore, data recovery within each code segment is performed column by column, according to the reverse-lexicographical order, i.e., column $\left[\f{S}\right]_{:,\sen I}$ will be decoded after column  $\left[\f{S}\right]_{:,\sen J}$ if an only if $\sen I \prec \sen J$. We also prove the proposition by induction over the code segments and column labels. 
        
        As the base step of induction, consider a code segment $\f{S}$  with $\md{\f{S}}=0$. Then, from~\eqref{eq:proc:rec:B}  and Lemma~\ref{lm:bottom-mode0}
         we have 
        \begin{align*}
            \gdc^{-1} \cdot \left( \encdc \cdot\fSeg \right) &=  \up{\fSeg}+\gdc^{-1} \cdot  \ddc \cdot \down{\fSeg} = \up{\fSeg},
        \end{align*}
        which provides us with $\up{\fSeg}$. This together with $\down{\fSeg}$ (from Lemma~\ref{lm:bottom-mode0} fully recover $\f{S}$. Then, using Lemma~\ref{lm:decode:separate} we can recover the original symbols in $\eSMat{\seg}{m}{y,\sen {Y}}{\e{P}}$ (which are all zero by definition) and the injected symbols $\tOMat{m}{y,\sen Y}{\e{P}}$.
        
        It remains to prove the  induction step, which is the following statement: 
        \begin{quote}
        \emph{For any code segment $\f{S}$ with $\md{\f{S}}=m$ and any column index $\sen I$, 
        if the message matrices for all code segments $\f{Q}$ with $\md{\f{Q}}<m$ and all columns $\left[\f{S}\right]_{:,\sen J}$ with $\sen J > \sen I$ are decoded, then 
        the content of column $\left[\f{S}\right]_{:,\sen I}$ can be extracted from the observed data $\encdc \cdot\fSeg$ and the already decoded symbols.  }
        \end{quote}

    	To prove this statement, consider an arbitrary code segment  $\f{S}=\fSMat{\seg}{m}{y,\sen {Y}}{\e{P}}$  in $\MM$, 

    	and focus on an arbitrary column  $\left[\f{S}\right]_{:,\sen I}$. 
    Let $x\in \intv{k+1:d}$ and consider an entry  $\left[\fSMat{\seg}{m}{y,\sen {Y}}{\e{P}}\right]_{x,\sen I}$ in the lower part of this column. Recall that
    \begin{align*}
    	\left[\fSMat{\seg}{m}{y,\sen {Y}}{\e{P}}\right]_{x,\sen I} = \left[\eSMat{\seg}{m}{y,\sen {Y}}{\e{P}}\right]_{x,\sen I} + \left[\tOMat{m}{y,\sen {Y}}{\e{P}}\right]_{x,\sen I}.
	\end{align*}
	Now, we can distinguish the following three cases:     
	
    \noindent$\bullet\; \eSeg_{x,\sen I} \in \dgp{\eSeg}$: From the second part of Remark~\ref{rem:inj:explain} we know  that no injection is performed into the symbols in $\dgp{\eSeg}$, and hence $\left[\tOMat{m}{y,\sen {Y}}{\e{P}}\right]_{x,\sen I} = 0$ and $\left[\fSMat{\seg}{m}{y,\sen {Y}}{\e{P}}\right]_{x,\sen I}  = \left[\eSMat{\seg}{m}{y,\sen {Y}}{\e{P}}\right]_{x,\sen I}$  
    
    Then,  let $\sen {A} = \sen I \cap \intv{k}$, and $\sen {B} = \sen I \cap \intv{k+1:d}$. Again, the third part of Remark~\ref{rem:inj:explain} implies that symbol $\eSeg_{x,\sen I}$ is injected into position $(\max \sen{A}, \sen A \setminus \set{\max \sen{A}})$ of the child matrix $\eSMat{Q}{m}{x,\sen {B}}{\e{S}}$ of $\eSeg$. Therefore, we have 
    \begin{align}
        \fSeg_{x,\sen I} =  \eSeg_{x,\sen I} = (-1)^{1+\s{\eSeg}{\max \sen A} + \ind{\sen A \cup \sen B}{\max {\sen A}}} \left[\tOMat{m}{x,\sen {B}}{\eSeg} \right]_{\max {\sen A},\sen A \setminus \set{\max \sen{A}}}\hspace{-3pt}.\label{eq:rec:lower:case1}
    \end{align}
	Since $\f{Q}$ is a child matrix of $\f{S}$, we have $\md{\f{Q}} < \md{\f{S}}$, and hence by the induction assumption $\f{Q}$ is already decoded. Moreover, Lemma~\ref{lm:decode:separate} ensures that all the entries of $\eSMat{Q}{m}{x,\sen {B}}{\e{S}}$ and $\tOMat{m}{x,\sen {B}}{\e{S}}$ can be extracted from $\f{Q}$. Therefore, $\f{S}_{x,\sen I}$ can be recovered using~\eqref{eq:rec:lower:case1} and $\tOMat{m}{x,\sen {B}}{\eSeg} $. 

	\noindent$\bullet\;\eSeg_{x,\sen I} \in \ngp{\eSeg}$: Recall from Definition~\ref{def:groups} that symbols in this group are set to zero.  Moreover, belonging to $\ngp{\eSeg}$ implies that $x\leq \max \sen I$, and hence~\eqref{eq:inj:mat} ensures that no injection will be performed into the entry $(x,\sen I)$ of $\e{S}$. Therefore, we have 
	$\left[\fSMat{\seg}{m}{y,\sen {Y}}{\e{P}}\right]_{x,\sen I} = \left[\eSMat{\seg}{m}{y,\sen {Y}}{\e{P}}\right]_{x,\sen I}=
	\left[\tOMat{m}{y,\sen Y}{\eSeg}\right]_{x,\sen I}=
	0$.
		
	\noindent$\bullet\;\eSeg_{x,\sen I} \in \pgp{\eSeg}$: Since this symbol belongs to $\pgp{\eSeg}$, we have  $x\in \intv{k+1:d}$ and $x>\max \sen I$. Our goal is to decode $\fSeg_{x, \sen I}$. First, note that position $(x,\sen I)$ may be hosting a (secondary) injection from the parent matrix $\e{P}$. According to \eqref{eq:inj:mat} we have
	\begin{align}
	\fSeg_{x, \sen I}  &= \left[\fSMat{\seg}{m}{y,\sen {Y}}{\e{P} }\right]_{x,\sen I} \nonumber\\ & = \left[\eSMat{\seg}{m}{y,\sen {Y}}{\e{P} }\right]_{x,\sen I} + \left[\tOMat{T}{y,\sen {Y}}{\e{P} }\right]_{x,\sen I} \nonumber\\
	&= \eSeg_{x,\sen I} + 
	(-1)^{1+\s{\e{P}}{x}  + \ind{\sen I \cup\set{x} \cup \sen Y}{x}}\e{P}_{y,\sen I \cup \set{x} \cup \sen Y } \1{x\notin \sen Y, \sen I \cap \sen Y = \varnothing }.
	\label{eq:decode:G2}
	\end{align}
	Hence, in order to decode $\fSeg_{x, \sen I}$ we need to find both $\eSeg_{x,\sen I}=(-1)^{\s{\eSeg}{x}} w_{x,\sen I \cup \set{x}} $ and the injected symbol $\e{P}_{y,\sen I \cup \set{x} \cup \sen Y }$. 
	The first term in \eqref{eq:decode:G2} can be decoded using the parity equation \eqref{eq:parityeq}. Note that $w_{x,\sen I \cup \set{x}}$ satisfies a parity equation along with $\set{w_{t,\sen I \cup \set{x}}: t\in \sen I}$. A  $w$-symbol $w_{t,\sen I \cup \set{x}}$ is located in column $\sen J = (\sen I \cup \set{x}) \setminus \set{t}$ of $\e{S}$. 
	Since $t\leq \max \sen I <x$, the lexicographical order of $\sen I$ and $\sen J$ satisfies $\sen I \prec \sen J$. Therefore, due to the induction assumption, every symbol in column $\sen J$ including $\fSeg_{t, \sen J}$ is already decoded. Moreover, \eqref{eq:inj:mat} implies that no injection is performed into position $(t,\sen J)$ of $\e{S}$, since $t\leq \max \sen I <x = \max \sen J$. This leads to $\fSeg_{t, \sen J} = \eSeg_{t, \sen J} =  w_{t,\sen J \cup \set{t}} = w_{t,\sen I \cup \set{x}}$.  Thus we can  retrieve $\eSeg_{x,\sen I}$ from

	\begin{align}
	\eSeg_{x,\sen I} = (-1)^{\sigma_{\e{S}}(x)+ m} \sum_{t\in \sen I} (-1)^{\sigma_{\e{S}}(t)+\ind{\sen I \cup \{x\}}{t}} \left[\fSMat{\seg}{m}{y,\sen {Y}}{\e{P}} \right]_{t,(\sen I \cup \{x\})\setminus \{t\}}. \label{eq:parity:rec}
	\end{align}

    Next, we need to  retrieve the second term in \eqref{eq:decode:G2}, i.e., $\e{P}_{y,\sen I \cup \set{x} \cup \sen Y }$. We note that the injection of this symbol into code matrix $\f{S}$ is a \emph{secondary} injection, since the injection position is in the lower part of $\f{S}$, which is the child matrix (see Section~\ref{subsec:injection}) of $\e{P}$. Therefore, the symbol $\e{P}_{y,\sen I \cup \set{x} \cup \sen Y}$ is also \emph{primarily} injected into another child matrix. By partitioning the column index $\sen I \cup\{x\}\cup \sen Y$ into  $\sen {A}'=[\sen I \cup \set{x} \cup \sen Y] \cap \intv{k}$ and $\sen {B}'=[\sen I \cup \set{x} \cup \sen Y] \cap \intv{k+1:d}$, we can find the position and injection pair of the primary injection of $\e{P}_{y,\sen I \cup \set{x} \cup \sen Y}$. 

    This leads to position $(\max \sen {A}', \sen {A}' \setminus \set{\max \sen {A}'})$ of the child $\fSMat{Q}{m}{y,\sen {B'}}{\e{P}}$ (see part $3$ of Remark~\ref{rem:inj:explain}), where $\f{Q}$ is a sibling of $\f{S}$. 

	The mode of $\f{Q}$ satisfies
	\begin{align}
		\md{\f{Q}} &= \size{\sen{A}'}-1 \label{eq:2childmode:1}\\
		&= \size{\left[\sen I \cup \set{x} \cup \sen Y\right] \cap \intv{k}}-1\\
		&= \size{\sen I \cap \intv{k}}-1\label{eq:2childmode:2}\\
		&< \size{\sen I} = \md{\fSeg} \label{eq:2childmode:3},
	\end{align}
	where \eqref{eq:2childmode:1} follows from the relation between the mode of child matrix and the injected symbol of the parent matrix for $\e{P}_{y,\sen I \cup \set{x} \cup \sen Y}$ given in~\eqref{eq:mode:relation}, equality in \eqref{eq:2childmode:2}  holds since 
	$x\in \intv{k+1:d}$ and $\sen Y \subseteq \intv{k+1:d}$, 
	and the last equality in~\eqref{eq:2childmode:3} holds since $\sen I$ is a columns index of $\fSeg$. 
	Now, since $\md{\f{Q}} < \md{\f{S}}$, the induction assumption implies that the child matrix $\f{Q}$ is already decoded. Moreover, Lemma~\ref{lm:decode:separate} ensures that $\e{Q}$ and $\tOMat{T}{y,\sen {B}'}{\e{P}}$  can be extracted from $\f{Q}$. Therefore, the entry at position $(\max \sen A', \sen A' \setminus \set{\max \sen A'})$ of $\tOMat{T}{y,\sen {B}'}{\e{P}}$ will provide us with 
	\begin{align}
	    \left[\tOMat{T}{y,\sen {B}'}{\e{P}}\right]_{\max \sen A', \sen A' \setminus \set{\max \sen A'}}  &= (-1)^{1+\s{\e{P}}{\max \sen A'}+\ind{\sen A' \cup \sen{B}'}{\max \sen{A}'}} \e{P}_{y,\sen A' \cup \sen {B}'} \nonumber\\ &  = (-1)^{1+\s{\e{P}}{\max \sen A'}+\ind{\sen A' \cup \sen{B}'}{\max \sen{A}'}} \e{P}_{y,\sen I \cup \set{x} \cup \sen Y}, \label{eq:rec:secondary}
	\end{align}
	from which the injected symbol $\e{P}_{y,\sen I \cup \set{x} \cup \sen Y}$ can be  recovered. Finally, we can plugin $\e{S}_{x,\sen I}$ and $\e{P}_{y,\sen I \cup \set{x} \cup \sen Y}$ (obtained in~\eqref{eq:parity:rec} and \eqref{eq:rec:secondary}, respectively) in \eqref{eq:decode:G2}, to recover $\f{S}_{x,\sen I}$ in group $\pgp{\eSeg}$.

	With this, the lower part of column $ \f{S}_{:,\sen I}$ is fully recovered. Then, the upper part can be decoded using~\eqref{eq:proc:rec:B} and the observed data $\encdc \cdot \fSeg_{:,\sen I}$.  Stacking $\up{\f{S}}_{:,\sen I}$ on top of $\down{\f{S}}_{:,\sen I}$, we find the entire column $\left[\f{S}\right]_{:,\sen I}$. 
	This completes the proof of the induction step. 
	
	Repeating the induction steps  for all code segments and all columns leads to the recovery of the entire matrix $\MM$. This completes the proof.
	\end{proof}
	
    \subsection{An Illustrative Example for Data Recovery}
    \label{subsec:data:rec:exmp}
    \begin{ex} 
	    \label{ex:datarec:A}
    We continue the data recovery for the running example of this paper. First, note that as it was explained earlier in this section, once the bottom part (the bottom $n-k$ rows) of each column is decoded, its upper part (the top $k$ entries) can be easily recovered using~\eqref{eq:proc:rec:B}. 
    
    We start with  code segments at the lowest level of the hierarchical tree in Fig.~\ref{fig:tree}, i.e., segments with mode zero, namely, $\f{T}_6,\dots, \f{T}_{14}$. Here,  $\down{\f{T}_6}=\down{\f{T}_7}=\cdots =\down{\f{T}_{14}}=\mathbf {0}$. This can be verified   for $\f{T}_{11}$, using  $\e{T}_{11}$ and  $\tOMat{0}{5,\set{6}}{\e{T}_2}$ (that includes the injections into $\f{T}_{11}$) given in  \eqref{eq:D11}. Hence, we can recover the entire column $\f{T}_i$ for $i\in\intv{6:14}$. 
    
   We continue the  data recovery process with code segments of  mode $1$. Let us focus on $\f{T}_5= \e{T}_5+\tOMat{0}{6,\set{5,6}}{\e{T}_0}$, where $\e{T}_5$ and $\tOMat{0}{6,\set{5,6}}{\e{T}_0}$  are given in~\eqref{eq:D5} and~\eqref{eq:O5}, respectively. We recover the columns of $\f{T}_5$ according to reverse-lexicographical order, i.e., we start from  $\set{6}$, and continue with $\set{5}$, $\set{4}$, $\set{3}$, and $\set{2}$, until we get to $\set{1}$. For  column $\set{6}$, both elements in the bottom part belong to $\ngpi{T}{5}$, and hence we have elements $\left[\e{T}_5\right]_{5,\{6\}}= \left[\e{T}_5\right]_{6,\{6\}}=0$ and 
   $\left[\tOMat{0}{6,\set{5,6}}{\e{T}_0}\right]_{5,\{6\}}
   =\left[\tOMat{0}{6,\set{5,6}}{\e{T}_0}\right]_{6,\{6\}} =0$. Therefore, we have $\left[\down{\f{T}}_5\right]_{:,\set{6}} = \mathbf{0}$, and 
   $\left[\up{\f{T}}_5\right]_{:,\set{6}}$ can be recovered using~\eqref{eq:proc:rec:B}. 
   
    Next, to decode column $\set{5}$ we first need to find its bottom entries,  first note that the entry, $\left[\f{T}_5\right]_{5,\{5\}}$ and $\left[\f{T}_5\right]_{6,\{5\}}$. It is easy to check that the entry at position $(5,\set{5})$ belongs to  $\ngpi{T}{5}$ and hence, 
    $\left[\e{T}_5\right]_{5,\{5\}} = \left[\tOMat{0}{6,\set{5,6}}{\e{T}_0}\right]_{5,\{5\}} =0$. The entry at position $(6,\set{5})$ belongs to
     $\pgpi{T}{5}$ and there is no secondary injection into this symbol, i.e, $\left[\tOMat{1}{6,\set {5,6}}{\e{T}_0}\right]_{6,\set{5}}= 0$. However, we find it is 
    equals to zero, i.e. $\left[\e{T}_5\right]_{6,\{5\}} = - w_{6,\set{5,6}}^{\nd{5}}$.
    This is because the parity equation in~\ref{eq:parityeq} for the parity group $\set{5,6}$ implies that 
    \begin{align*}
        \left[\e{T}_5\right]_{6,\{5\}} = - w_{6,\set{5,6}}^{\nd{5}} = - w_{5,\set{5,6}}^{\nd{5}}
        = -\left[\e{T}_5\right]_{5,\{6\}},
    \end{align*}
    and the latter symbol appears in column $\set{6}$, which is  already decoded. Thus, we have recovered $\left[\down{\f{T}}_5\right]_{:,\set{5}}$ and 
   $\left[\up{\f{T}}_5\right]_{:,\set{5}}$ can be extracted from~\eqref{eq:proc:rec:B}.

    For the remaining columns, $\set{4}$,  $\set{3}$, $\set{2}$, and $\set{1}$ we notice that no injection is performed into their lower entries, and the symbols in $\down{\e{T}}_5$ are all $w$-symbols, that can be recovered from the parity equation and other $w$-symbols which are previously decoded. Hence, the data recovery process continues in a similar fashion. 
	The next step consists of decoding message matrices of code segments with mode $m=2$, i.e., $\f{T}_1$, $\f{T}_2$, and $\f{T}_3$. We focus on decoding $\f{T}_2$ for the sake of illustration.
	Note that $\f{T}_2 = \e{T}_2 + \tOMat{0}{6,\set{6}}{\e{T}_0}$, and $\e{T}_2$ and $\tOMat{0}{6,\set{6}}{\e{T}_0}$ are given in~\eqref{eq:D2} and~\eqref{eq:O2}, respectively. The symbols in each group of  $\down{\e{T}_2}$ are given in \eqref{eq:group-T2}. In particular, the symbols in $\ngpi{T}{2}=\set{v_{5,\set{5,6}}^{\nd{2}},v_{6,\set{5,6}}^{\nd{2}}}$ are set zero. Hence, the lower part of column $\set{5,6}$ is zero, and $\left[\f{T}_2\right]_{:,\set{5,6}}$ can be easily decoded.
	
	The symbols in $\dgpi {T}{2}$ are injected into code segments $\e{T}_9$, $\e{T}_{10}$, and $\e{T}_{11}$, and hence, can be retrieved once $\f{T}_{9}$, $\f{T}_{10}$, and $\f{T}_{11}$ are decoded. For instance, consider the entry of $\e{T}_2$ at position $(5,\set{5,6})$, which is $w_{5,\set{4,5,6}}^{\nd{2}}$. For this position we have  $\sen A = \set{4,6} \cap \intv{4} = \set{4}$, and $\sen B = \set{4,6} \cap \intv{5:6} = \set{6}$. Therefore, this symbol is injected into the child code of $\e{T}_2$ with injection pair $(5,\set B) = (5,\set{6})$ at row $\max \sen A =4$ and column $\sen A \setminus \set{\max \sen A} = \varnothing$. That is, $\left[\f{T}_{11}\right]_{4,\varnothing}$, as can verified in~\eqref{eq:D11}. Therefore, having $\f{T}_{11}$ decoded, we also have $w_{5,\set{4,5,6}}^{\nd{2}}$.

	Next, consider column $\{3,4\}$ of $\f{T}_2$, that includes some elements from $\pgpi{T}{2}$ in its bottom part. In particular, we have  $\left[\f{T}_2\right]_{5,\set{3,4}}=w_{5,\set{3,4,5}}^{\nd{2}}+v_{6,\set{3,4,5,6}}^{\nd{0}}$, where $w_{5,\set{3,4,5}}^{\nd{2}} = \left[\f{T}_2\right]_{5,\set{3,4}} \in \pgpi{T}{2}$ and $v_{6,\set{3,4,5,6}}^{\nd{0}}$ is the symbol injected from $\e{T}_0$. The first ingredient $w_{5,\set{3,4,5}}^{\nd{2}}$ can be decoded from the parity equation for $w$-group $\set{3,4,5}$: 

	\[
	w_{5,\set{3,4,5}}^{\nd{2}} = - w_{3,\set{3,4,5}}^{\nd{2}} + w_{4,\set{3,4,5}}^{\nd{2}}.
	\]
	Note that $w_{3,\set{3,4,5}}^{\nd{2}}$ and   $w_{4,\set{3,4,5}}^{\nd{2}}$ appear in columns $\set{4,5}$ and $\set{3,5}$, which are already decoded before we arrive to decoding column $\set{4,5}$. In order to retrieve the second ingredient $v_{6,\set{3,4,5,6}}^{\nd{0}}$, we notice that the current injection is secondary, and the same symbol is also primarily injected to some other child code of $\e{T}_0$, which is decoded before we arrive to decoding $\e{T}_2$. To find the injection pair and position of such primary injection, we note that  $v_{6,\set{3,4,5,6}}^{\nd{0}}$ appears  in position $(y,\sen J) = (6,\set{3,4,5,6})$ of $\e{T}_0$  (up to a sign). Therefore, we have $\sen A' = \sen J \cap \intv{4} = \set{3,4}$ and $\sen B'=\sen J \cap \intv{5:6} = \set{5,6}$. Therefore, the injection pair for the primary injection of $v_{6,\set{3,4,5,6}}^{\nd{0}}$ is $(y,\sen B') = (6,\set{5,6})$ which corresponds to $\e{T}_5$ (see Fig.~\ref{fig:tree}). The position of injection is $(\max \sen A', \sen A' \setminus \set{\max \sen A'}) = (4,\set{3})$. Therefore, $v_{6,\set{3,4,5,6}}^{\nd{0}}$ can be found from $\left[\f{T}_5\right]_{4,\set{3}}$. This can be verified from~\eqref{eq:O5}. Note that $\md{\e{T}}_5 = 1 < 2 = \md{\e{T}_2}$, and hence $\f{T}_5$ is decoded before we arrive to decoding $\f{T}_5$.

	Repeating a similar procedure we can decode all columns of $\f{T}_2$. Once all three code segments of mode $2$, i.e., $\f{T}_1$, $\f{T}_2$, and $\f{T}_3$ are decoded, we can proceed to the recovery of the  root code segment $\f{T}_0$. 
	\end{ex} 
	
		\section{The Code Parameters}
	\label{sec:parameters}	
	The construction of the cascade code is described in Section~\ref{sec:supermessage}. However, the finding of the parameters of the resulting code (such as the size of super-message matrix $\MM$) requires an explicit evaluation of the number of code segments introduced throughout the injection process.

	\subsection{An Implicit Evaluating of the Code Parameters}
	Consider the construction of an $(n,k,d;\mu)$ cascade code. Recall that we start from a signed determinant code of mode $\mu$ and may introduce many code segments to complete the injection process. Let $\mdcnt_m$ be the total number of $(d;m)$ code segments of mode $m$ needed to complete all the injections. The super-message matrix $\MM$ is obtained by concatenating all code segments, which results in a matrix with $d$ rows and a total of $\sum_{m=0}^\mu \mdcnt_m \alpha_m$ columns. Therefore, the required per-node storage capacity of the resulting code is 
	\begin{align}
	\alpha(k,d;\mu) = \sum_{m=0}^{\mu} \mdcnt_m \alpha_{m}  = \sum_{m=0}^{\mu} \mdcnt_m \binom{d}{m}.
	\label{eq:alpha-in-L}
	\end{align}
	Similarly, the repair bandwidth of the resulting code is given by 
	\begin{align}
	\beta(k,d;\mu) = \sum_{m=0}^{\mu} \mdcnt_m \beta_{m}  = \sum_{m=0}^{\mu} \mdcnt_m \binom{d-1}{m-1}.
	\label{eq:beta-in-L}
	\end{align}
	The total number of data symbols stored in matrix $\MM$ is the sum of the number of data symbols in each code segment. A code segment $\e{P}$ of mode $m$ can store up to $F_m = m\binom{d+1}{m+1}$ symbols. However, recall that data symbols in group $\ngp{P}$ (see Definition~\ref{def:groups}) will be set to zero, which yields to a reduction in the number of stored symbols. For an $(n,d,d;m)$ signed determinant code used as a code segment, the reduction due to nulling the data symbols in $\ngp{P}$ can be found from 
		\begin{align}
	    N_m  &= \size{\ngp{P}}\nonumber\\
	    &=\big|\big\{(x,\sen A, \sen B)\! :  \sen A \subseteq [k], \sen B \subseteq [k+1:d], |\sen A \cup \sen B|= m, x\in \intv{k+1:d}, x \leq \sen \max {\sen B}, \sen A = \varnothing\big\}\big| \nonumber\\
		&= \size{\set{(x,\sen B)\!: \sen B \subseteq [k+1\!:d], |\sen B|=m, x \in \sen B}}\nonumber \\ 
&\phantom{=}\quad +|\{(x, \sen B) \!: \sen B \subseteq [k+1\!:\hspace{-1pt}d], |\sen B|=m, x \in [k+1\!:\hspace{-1pt}d]\hspace{-1pt}\setminus \hspace{-2pt}\sen B, x\leq \max \sen B\}|\nonumber\\
&= \size{\set{(x,\sen B)\!: \sen B \subseteq [k+1\!:d], |\sen B|=m, x \in \sen B}}\nonumber \\ 
&\phantom{=}\quad+|\{(x, \sen B) \!: \set{x}\cup\sen B \subseteq [k+1\!:\hspace{-1pt}d], |\set{x}\cup\sen B|=m+1, x \neq \max \set{x}\cup\sen B\}|  \nonumber\\
		&= m \binom{d-k}{m} + m \binom{d-k}{m+1}  \label{eq:G1_count:3:1}\\
		&= m \binom{d-k+1}{m+1}. \label{eq:G1_count:3}
		\end{align}
	Note that in~\eqref{eq:G1_count:3:1}, for the first set, we  choose a subset of size $m$ from $[k+1:d]$ for $\sen B$, and any element of $\sen B$ is a valid choice for $x$. Similarly, for the second set, we  choose a subset of size $m+1$ from $[k+1:d]$ for $\set{x}\cup\sen B$, and then $x$ can be any element of $\set{x}\cup\sen B$ except the largest one.  Finally, we used Pascal's identity in \eqref{eq:G1_count:3}.
	
	Subtracting the number of nulled symbol, the number of data symbols in a signed determinant code of mode $m$ will be $F_m - N_m$.  Therefore,  the total number of data symbols in the super-message matrix can be evaluated as
	\begin{align}
	F(k,d;\mu) &= \sum_{m=0}^{\mu} \mdcnt_m \left(F_{m}-N_{m}\right) = \sum_{m=0}^{\mu} \mdcnt_m m \left[\binom{d+1}{m+1} -\binom{d-k+1}{m+1}\right].
	\label{eq:F-in-L}
	\end{align}
	The rest of this section is dedicated to the evaluation of  parameters $\mdcnt_m$'s, the number of code segments of mode $m$. Then, we can  explicitly characterize the code parameters.  
	
	\subsection{The Number of  Child Matrices}
	Consider a code segment  $\e{P}$ of  $\md{\e{P}}=j$. Recall that each  child matrix of $\e{P}$ is a (signed) determinant codes of the form $\eSMat{Q}{m}{x, \sen B}{\e{P}}$. Hence, there is a one-to-one map between the child segments of $\e{P}$ and the injection pairs $(x,\sen B)$ satisfying the conditions of Remark~\ref{rem:injpair}.   From~\eqref{eq:mode:relation}, the mode of child matrix with injection pair $(x,\sen B)$ is given by 
	\[
	m = \md{\e{Q}} = \md{\e{P}} - |\sen B| - 1 = j - |\sen B| - 1.
	\]
	Therefore, the injection pairs leading to a child matrix of mode $m$ can be found from  
	\begin{align*}
        \left| \left\{(x, \sen B): \sen {B} \subseteq \intv {k+1:d},\size{\sen B} = j-m-1, x\in \intv {k+1:d}, x \leq \max \sen B \right\}\right|.
	\end{align*}
	We can distinguish two cases for the pairs $(x,\sen B)$ in this set. First, if $x \in \sen B$, then there are $\binom{d-k}{\size{\sen B}}$ choices of $\sen B$, and there are $\size{\sen B}$ choices for $x \in \sen B$. Second, if $x \notin \sen B$, then $\sen B \cup \set{x}$ is a subset of $\intv{k+1:d}$ of size $\size{\sen B}+1$. Therefore, there are $\binom{d-k}{\size{\sen B}+1}$ choices for $\sen B \cup \set{x}$. Moreover, for a given  $\sen B \cup \set{x}$, each entry except the largest one can be chosen to be $x$ (and the remaining ones will form $\sen B$).  
	Therefore, the total number of injection pairs leading to a child matrix of mode $m$ for a parent matrix of mode $j$ is given by 
	\begin{align}
	(j-m-1)\binom{d-k}{j-m-1} + (j-m-1)\binom{d-k}{j-m} = (j-m-1)\binom{d-k+1}{j-m}. \label{eq:number-of-children}
	\end{align}

	\subsection{Recursive Equations for $\mdcnt_m$ Parameters}

	The code construction starts from a determinant code of mode $\mu $ as  the root. Hence, we have  $\mdcnt_\mu=1$. 
	Let $m$ be an integer in $\{0,1,\dots, \mu-1\}$.  
	In general, a code segment of mode $m$ can be introduced by for injecting the  symbols of any parent matrix of mode $j$ with $j>m$. For a fixed $j$, the required number of child matrices of mode $m$ is given by~\eqref{eq:number-of-children}. 

	Note that child matrices are dedicated to their parent matrix, and cannot be shared by multiple parent matrices.  Therefore, if there are $t_j$ code segments of mode $j$ in the super-message matrix $\MM$, then the total number of  child matrices of mode $m$ is given by 
	\begin{align}
	\mdcnt_m = \sum_{j=m+1}^{\mu} \mdcnt_j (j-m-1) \binom{d-k+1}{j-m}.
	\label{eq:ell_rec}
	\end{align}
	This is a (reverse) recursive equation with starting point $\mdcnt_\mu=1$. Next, we solve this recursive equation to  obtain explicit expressions for $\mdcnt_m$'s.   
	\subsection{The Explicit Evaluation of $\mdcnt_m$ Parameters}
		\label{sec:t:eval}
	Note that we are only interested in the value of $\mdcnt_m$ defined by~\eqref{eq:ell_rec} for  $m\in\set{0,1,\cdots ,\mu-1}$, and we can assume  arbitrary values for $t_m$ when $m<0$ or $m>\mu$. 

	In particular, we expand the range of $m$ to include all integers, by defining dummy variables $\{\mdcnt_m: m<0 \textrm{ or } m>\mu\}$ such that 
	\begin{align}
	\mdcnt_m = \left\{
	\begin{array}{ll}
	1 & m=\mu,\\
	\sum_{j=m+1}^{\mu} \mdcnt_j (j-m-1) \binom{d-k+1}{j-m} & m \neq \mu.
	\end{array}
	\right.
	\label{eq:ell_rec_2}
	\end{align}
	Note that this immediately implies $\mdcnt_m = 0$ for $m>\mu$. However, it may lead to some non-trivial (and meaningless) values for $t_m$ with $m<0$. 
	We also define a sequence $\{p_m\}_{m=-\infty}^{\infty}$ as $p_m = \mdcnt_{\mu-m}$ for all $m\in \mathbb{Z}$. 
	The next lemma provides a non-recursive expression for sequence $\{p_m\}$.

	\begin{lm}
		The parameters in  sequence $\{p_m\}$ can be found from 
		\begin{align*}
		p_{m} = \sum_{\ell=0}^{m} (-1)^\ell (d-k)^{m-\ell} \binom{d-k+\ell-1}{\ell},
		\end{align*}
		for $0\leq m \leq \mu$.
		\label{lm:p_rec}	
	\end{lm}
	We refer to Appendix~\ref{sec:lm:p_rec} for the proof of Lemma~\ref{lm:p_rec}.	
	
	It is clear that we can immediately find a non-recursive expression for $\mdcnt_m$ using the fact that  $\mdcnt_m = p_{\mu-m}$. However, it turns out that it is more convenient to work with sequence $p_m$, without directly evaluating $t_m$. 
	\subsection{The Explicit Evaluation of the Code Parameters}
	We need to show that the implicit code parameters obtained in \eqref{eq:alpha-in-L}, \eqref{eq:beta-in-L}, and \eqref{eq:F-in-L} are equal to  those claimed in Theorem~\ref{thm:main}. The following lemma will be helpful to simplify the derivation. The proof of the lemma can be found in Appendix~\ref{sec:lm:helping:iden}.

	\begin{lm}
		\label{lm:helping:iden}
		For integer numbers $ a,b \in \mathbb{Z}$, we have
		\begin{align}
		\sum_{m=-\infty}^{\infty} p_{m-\mu} \binom{d+a}{m+b}  = \sum_{m=-b}^{\mu}(d-k)^{\mu-m} \binom{k+a}{m+b}.
		\label{eq:lm:a-b}
		\end{align}
	\end{lm}

	Now, we are ready to evaluate the code parameters using Lemma~\ref{lm:helping:iden}. 
	
	\begin{itemize}
		\item Node storage size: Starting from~\eqref{eq:alpha-in-L}, we have
		\begin{align}
		\alpha(k,d;\mu) 
		&= \sum_{m=0}^{\mu} \mdcnt_m \alpha_{m} \nonumber\\ 
		&= \sum_{m=0}^{\mu} \mdcnt_m \binom{d}{m} \nonumber\\
		&= \sum_{m=0}^{\mu} p_{m-\mu} \binom{d}{m}\nonumber\\
		&= \sum_{m=-\infty}^{\infty} p_{m-\mu} \binom{d}{m} \label{eq:alpha:2}\\
		&=\sum_{m=0}^{\mu}(d-k)^{\mu-m} \binom{k}{m}.\label{eq:alpha:3}
		\end{align}
		In \eqref{eq:alpha:2} we used the fact that $\binom{d}{m}=0$ for  $m<0$ and $p_{m-\mu}=0$ for $m>\mu$, and \eqref{eq:alpha:3} follows from Lemma~\ref{lm:helping:iden} for $a=b=0$. Finally,~\eqref{eq:alpha:3} is exactly the node storage size claimed in Theorem~\ref{thm:main}.
		
		\item Repair bandwidth: Similarly, we can start from the repair bandwidth in~\eqref{eq:beta-in-L} and follow similar steps to  simplify it as
		\begin{align}
		\beta(k,d;\mu) &= \sum_{m=0}^{\mu} \mdcnt_m \beta_{m}  \nonumber\\ &= \sum_{m=0}^{\mu} \mdcnt_m \binom{d-1}{m-1} \nonumber\\ &= \sum_{m=0}^{\mu} p_{m-\mu} \binom{d-1}{m-1}\nonumber\\
		&= \sum_{m=-\infty}^{\infty} p_{m-\mu} \binom{d-1}{m-1}\nonumber\\
		&=\sum_{m=1}^{\mu}(d-k)^{\mu-m} \binom{k-1}{m-1}\label{eq:alpha:4}\\
		&=\sum_{m=0}^{\mu}(d-k)^{\mu-m} \binom{k-1}{m-1}.\label{eq:overall:bndwidth}
		\end{align}
		Note that we have  Lemma~\ref{lm:helping:iden} for $a=b=-1$ in \eqref{eq:alpha:4}, and~\eqref{eq:overall:bndwidth} holds since $\binom{k-1}{0-1}=0$. This proves that the repair bandwidth of the proposed code is identical to that claimed in Theorem~\ref{thm:main}.  
		
		\item File size: Starting from the implicit expression in~\eqref{eq:F-in-L}, we can write
    	\begin{align}
    		F(k,d;\mu)&=\sum_{m=0}^{\mu}  \mdcnt_m \left[ m\binom{d+1}{m+1} - m\binom{d-k+1}{m+1}\right]\nonumber\\
    		&=\sum_{m=-\infty}^{\infty} \mdcnt_m \left[m \binom{d+1}{m+1} - m\binom{d-k+1}{m+1}\right]
    		\label{eq:F:3}\\
    		&=\sum_{m=-\infty}^{\infty} p_{\mu-m} \left[ (m+1) \binom{d+1}{m+1} -  \binom{d+1}{m+1} -(m+1) \binom{d-k+1}{m+1} +  \binom{d-k+1}{m+1}\right]  \label{eq:F:4}\\ 
    		&=\sum_{m=-\infty}^{\infty} p_{\mu-m} \left[ (d+1) \binom{d}{m} - \binom{d+1}{m+1} -(d-k+1) \binom{d-k}{m} +  \binom{d-k+1}{m+1}\right] \label{eq:F:4:1}\\
    		&=(d+1)\sum_{m=0}^{\mu}(d-k)^{\mu-m} \binom{k}{m} -\sum_{m=-1}^{\mu}(d-k)^{\mu-m} \binom{k+1}{m+1}\nonumber \\ &\quad- (d-k+1)\sum_{m=0}^{\mu}(d-k)^{\mu-m} \binom{k-k}{m}+\sum_{m=-1}^{\mu}(d-k)^{\mu-m} \binom{k-k+1}{m+1} \label{eq:F:5}\\
    		&= (d+1)\sum_{m=0}^{\mu}(d-k)^{\mu-m} \binom{k}{m} 	-\sum_{m=-1}^{\mu}(d-k)^{\mu-m} \binom{k+1}{m+1}  \nonumber\\
    		&\hspace{12mm}- \left[(d-k+1) (d-k)^{\mu}\right] +  \left[  (d-k)^{\mu+1} +  (d-k)^{\mu}\right]  
    		\label{eq:F:6}\\
    		&= (d+1)\sum_{m=0}^{\mu}(d-k)^{\mu-m} \binom{k}{m} -\sum_{m=-1}^{\mu}(d-k)^{\mu-m} \left[\binom{k}{m} + \binom{k}{m+1}\right]\label{eq:F:6:1}\\
    		&= d \sum_{m=0}^{\mu}(d-k)^{\mu-m} \binom{k}{m} -\sum_{m=0}^{\mu+1}(d-k)^{\mu+1-m} \binom{k}{m} \label{eq:F:7}\\
    		&= d \sum_{m=0}^{\mu}(d-k)^{\mu-m} \binom{k}{m} - (d-k) \sum_{m=0}^{\mu}(d-k)^{\mu-m} \binom{k}{m} - (d-k)^{\mu+1 - (\mu+1)}\binom{k}{\mu+1} \label{eq:F:8}\\
    		&= \sum_{m=0}^{\mu} k(d-k)^{\mu-m} \binom{k}{m}-\binom{k}{\mu+1}. \label{eq:F:9}
    	\end{align} 
		Note that~\eqref{eq:F:3} holds since $\mdcnt_m=0$ for $m>\mu$ and $\left[ \binom{d+1}{m+1} - \binom{d-k+1}{m+1}\right]=0$ for $m< 0$, in \eqref{eq:F:4} we used $\mdcnt_m = p_{\mu-m}$ and some manipulation  $m =(m+1)-1$, \eqref{eq:F:4:1} holds since $a\binom{b}{a} = b \binom{b-1}{a-1}$, and  \eqref{eq:F:5} follows from four times evaluation of Lemma~\ref{lm:helping:iden}  for $(a,b)=(0,0)$, $(a,b)=(1,1)$, $(a,b)=(-k,0)$, and $(a,b)=(-k+1,1)$. The equality in \eqref{eq:F:6} holds since the terms in the third summation in \eqref{eq:F:5} are zero except for $m=0$, and similarly the terms in the fourth summation are zero except for $m=-1,0$. 
		Then, the third and forth terms in~\eqref{eq:F:6} get canceled, and we used the identity $\binom{k+1}{m+1} = \binom{k}{m}+ \binom{k}{m+1}$  in~\eqref{eq:F:6:1}. 
		In~\eqref{eq:F:7}, we used the fact that $\binom{k}{-1}=0$ and applied a unit-shift on the variable of the last summation. 
		The second summation over $m\in\{0,\dots, \mu+1\}$ in~\eqref{eq:F:7} is decomposed into $m\in\{0,\dots,\mu\}$ and $m=\mu+1$ in~\eqref{eq:F:8}. This leads to~\eqref{eq:F:9}, which is the storage capacity of the code, as claimed in Theorem~\ref{thm:main}. 
	\end{itemize}

\section{Discussion}
	\label{sec:conclusion}
	
	\subsection{Cascade Codes vs. Product-Matrix Codes}
	In this section, we compare the code construction proposed in this paper to the product-matrix (PM) code  introduced in  \cite{rashmi2011optimal}. Since both code constructions are linear, they both can be written as a product of the encoder matrix and  a message matrix, i.e.,  $\cC = \enc \cdot \MM$ as the \cite{rashmi2011optimal}. A natural question to ask is \emph{whether the two codes are \emph{equivalent} at MBR and MSR points} (recall that the construction of PM codes is limited to the extreme points, namely, MBR and MSR). In other words, one may wonder if \emph{the two codes structurally equivalent, in spite of their different constructions}. We give a \emph{negative} answer to this question, by proving some fundamental differences between the two codes. Here, we rely on the standard notion of \emph{equivalence} (e.g., see   \cite{shah2012interference})  to check if two codes are convertible to each other. 
	Two linear codes $\mathcal{C}$ and $\mathcal{C}'$ are called equivalent~\cite{shah2012interference}  if
	$\mathcal{C}'$ can be represented in
	terms of $\mathcal{C}$ by 
	\begin{enumerate}
		\item a change of basis of the vector
		space generated by the message symbols~(i.e., a remapping of the message symbols), and
		\item a change of basis
		of the column-spaces of the nodal generator matrices~(i.e., a remapping of the symbols stored within a node). 
		\item scale parameters of $(\alpha, \beta, F)$ of codes by an integer factor so that both codes have the same parameters.
	\end{enumerate}

	\begin{table*}
		\begin{tabular}{|>{\centering\arraybackslash}p{75mm}|>{\centering\arraybackslash}p{75mm}|}
			\hline
			Product-Matrix Code~\cite{rashmi2011optimal} & Cascade Code [this work] \\ \hline \hline
			Only for extreme points: MBR and MSR & For the entire storage-bandwidth trade-off \\  \hline
			MSR code construction: only for $d \geq 2k-2$ & MSR code construction:  all parameter sets \\ \hline
			MSR code parameters: & MSR code parameters: \\
			$(\alpha, \beta, F)  = \left(d-k+1,1,k(d-k+1)\right)$ & $(\alpha, \beta, F)  = \left((d-k+1)^k,(d-k+1)^{k-1},k(d-k+1)^k\right)$ \\ \hline
			Different structure and requirements for the & \multirow{2}{*}{Universal code construction for the entire trade-off}\\
			encoder matrix of MBR and MSR codes &  \\ \hline  
			Different repair mechanism for MBR and MSR codes
			& A generic repair mechanism with an encoder matrix $\mathbf{\Xi}$ for all corner points on the trade-off curve \\ \hline
		\end{tabular}\\[2mm]
		\caption{A comparison between the product-matrix codes and cascade codes.}
		\label{tab:Cascade-vs-PM}
	\end{table*}

	It turns out that the cascade code generated for the MBR point ($\mu=1$) is equivalent (and even identical) to  the one introduced in~\cite{rashmi2011optimal}. However, the codes generated for the MSR point using cascade construction and PM construction are  \emph{fundamentally} different, and not equivalent.  
	To show this, we focus on MSR codes for a distributed storage system with specific parameters of  $(n,k,d=2k-2)$, for some $k>2$. The parameters of an MSR cascade code can be found from Theorem~\ref{thm:main} by setting $\mu=k-1,$ as $(\alpha,\beta, F) = \left((k-1)^k, (k-1)^{k-1}, k(k-1)^k\right)$. On the other hand, the parameters of the PM code are given by $(\alpha', \beta', F') = (k-1, 1, k(k-1))$. So, in order for a fair comparison, one needs to concatenate $N=(k-1)^{k-1}$ copies of \emph{MSR PM codes} with independent message matrices ${\bf M}'_1, {\bf M}'_2, \dots, {\bf M}'_{N}$ to obtain a code with the same parameters as the cascade code. Let $\mathcal{C}$ and $\mathcal{C}'$ be the resulting cascade and PM codes for these parameters, respectively. We have
	\begin{align*}
	    \mathcal{C} = \mathbf{\Psi } \cdot \mathbf{M}, \qquad \mathcal{C}' = \mathbf{\Psi }' \cdot \begin{bmatrix} \mathbf{M}_1 & \cdots & \mathbf{M}_N\end{bmatrix},
	\end{align*}
	where $\mathbf{\Psi }$ and $\mathbf{\Psi }'$ are the encoder matrices for the two constructions. While the conditions for the encoder matrices $\mathbf{\Psi }$ and $\mathbf{\Psi }'$ are in general different, both sets of requirements are satisfied by the choice of  
	 Vandermonde matrices. 
		
	Let us focus on the repair data sent by the helper nodes in each code. 
	We denote by $\mathsf{Rep}_{h\rightarrow f}$ and $\mathsf{Rep}'_{h\rightarrow f}$ the \emph{vector space spanned by  the repair symbols} sent by a helper node $h$ to repair a failed node $f$ for the cascade and PM codes, respectively. It is clear that 
	$\dim (\mathsf{Rep}_{h\rightarrow f} )= \beta = (k-1)^{k-1}$ and $\dim (\mathsf{Rep}'_{h\rightarrow f}) = N \beta' = (k-1)^{k-1}$, i.e., both spaces have identical dimensions. However, the intersection of two of such vector spaces has a dimension, which is different for the two code constructions of interest.  
	This is formally highlighted in the following proposition. 
	\begin{prop}
		For an MSR cascade code for an $(n,k,d=2k-2)$ distributed storage system with parameters $(\alpha,\beta, F) = \left((k-1)^k, (k-1)^{k-1}, k(k-1)^k\right)$, and three distinct nodes $h$, $f$, and $g$, we have 
		\begin{align*}
		\dim (\mathsf{Rep}_{h\rightarrow f} \cap \mathsf{Rep}_{h\rightarrow g})=\sum_{m=0}^{\mu}(d-k)^{\mu-m} \left[2\binom{k-1}{m-1} -\binom{k}{m} - \binom{k-2}{m}\right], 
		\end{align*}	
		while for a concatenation of $N=(k-1)^{k-1}$ independent copies of MSR PM codes we have 
		\begin{align}
		\dim
		(\mathsf{Rep}'_{h\rightarrow f} \cap \mathsf{Rep}'_{h\rightarrow g})= 0. \nonumber
		\end{align}	
		\label{prop:diff}
	\end{prop}
	\begin{proof}[Proof of Proposition~\ref{prop:diff}]
		Recall that the repair space $\mathsf{Rep}_{h\rightarrow f}$ in a cascade code is simply a concatenation of repair data for each code segment, which is an $(n,d,d;m)$ (modified) signed determinant code.  
		For each code segment, we can use the result of~\cite{elyasi2018newndd} to evaluate  the overlap between  subspaces spanned by the repair symbols sent to two failed nodes. For a code segment with mode $m$, the dimension of the overlap between the subspaces spanned by the repair symbols sent from $h$ to $f$ and $g$ is given by  $2\binom{d-1}{m-1} - \left[\binom{d}{m} + \binom{d-2}{m}\right]$, as reported in \cite[Theorem~2]{elyasi2018newndd}. Recall that there are $\mdcnt_m$ code segments of mode $m$, where $\mdcnt_m$ is evaluated in Section~\ref{sec:t:eval}. Hence, summing up over all code segments, we get 
		\begin{align}
		\dim (\mathsf{Rep}_{h\rightarrow f} \cap \mathsf{Rep}_{h\rightarrow g})  &= \sum_{m=0}^{\mu} \mdcnt_m \left[2\binom{d-1}{m-1} - \binom{d}{m} + \binom{d-2}{m} \right] \nonumber \\
		&= \sum_{m=-\infty}^{\infty} p_{\mu-m} \left[2\binom{d-1}{m-1} - \binom{d}{m} + \binom{d-2}{m} \right] \nonumber\\
		&=\sum_{m=0}^{\mu}(d-k)^{\mu-m} \left[2\binom{k-1}{m-1} -\binom{k}{m} - \binom{k-2}{m}\right]. \label{eq:prop:diff}
		\end{align}
		Note that we have used Lemma~\ref{lm:helping:iden} 
		with $(a,b)=(-1,-1)$, $(a,b)=(0,0)$, and $(a,b)=(-2,0)$ in the last equation in 
		with where in~\eqref{eq:prop:diff}. 
		
		For the MSR PM code, we note that it is obtained by concatenating $N$ \emph{independent} copies of \emph{little} PM codes with parameters $(\alpha', \beta', F') = (k-1, 1, k(k-1))$.
		 For each  little PM code, the dimension of the overlap between the subspaces of interest is an integer number, which can be either $0$ or $1$. If the latter holds, then the subspace spanned by the repair 
		 symbols sent from $h$ to $f$ and $g$ are identical,  i.e., $\mathsf{Rep}'_{h\rightarrow f} = \mathsf{Rep}'_{h\rightarrow g}$. By symmetry, this should hold for any other failed node. Now, consider a set of helper nodes $\sen H$ with $|\sen H|=d$, and a set of failed nodes with $\sen F$ with $|\sen F|=k$. We can repair the content of all the nodes in $\sen F$ by sending only $\beta=1$ symbol from each of the helper nodes in $\sen H$, since  $\mathsf{Rep}'_{h\rightarrow f} = \mathsf{Rep}'_{h\rightarrow g}$ for any $f,g\in \sen F$ and any $h\in \sen H$. Moreover,  the entire information of the little PM code should be recoverable for the content of nodes in $\sen F$, since $|\sen F| = k$. This implies 
		\begin{align*}
		k(k-1)  = F \leq \sum_{h\in \sen H} \dim ( \mathsf{Rep}'_{h\rightarrow f}) = d\beta = 2(k-1),    
		\end{align*}
		which is in contradiction with $k>2$. Therefore, for each little MSR PM code we have $\dim(\mathsf{Rep}'_{h\rightarrow f} \cap \mathsf{Rep}'_{h\rightarrow g}) = 0$.  Summing up over $N$ independent copies,  we obtain  the claim of the proposition. 
	\end{proof}
	An immediate consequence of this proposition is that cascade codes and PM codes are not equivalent, and cannot be converted to each other by any scaling and change of bases.
	
	In general, cascade codes and PM codes are fundamentally different.  Some of the main distinctions between the two constructions are highlighted in Table~\ref{tab:Cascade-vs-PM}.
	
	\subsection{The Role of Redundancy in Cascade Codes}
	The parity symbols in the message matrix of a cascade code play a critical role to guarantee the properties of the code.  Such parity symbols were initially introduced for determinant codes \cite{elyasi2016determinant, elyasi2018newndd} in order to facilitate the repair process. However, they play no role in data recovery when $d=k$. While the redundancy introduced by such parity symbols can affect the overall storage capacity of the system,   the lower bounds in \cite{elyasi2015linear, prakash2015storage, duursma2015shortened} show that determinant codes are optimum for for DSS with  $d=k$. 
	
	In cascade codes, however, these parity  (redundant) symbols play two crucial roles: (1) they help with the repair mechanism, similar to their role in determinant codes, and (2) they make the data recovery possible, in spite of the fact that the data collector only has access to the coded content of $k<d$ nodes. More intuitively, this redundancy is used to provide a backup copy for  symbols of a determinant code that could not be retrieved, if the data collector could only observe the content of $k<d$ nodes. 
	
	It is easy to verify from the definition of injection process in~\eqref{eq:inj:mat} that  all the parity symbols of a child matrix  are filled with an injection from the parent code. This suggests that this redundancy is fully exploited, and thus, the proposed code has no further room for improvement. This is the foundation of our conjecture, that is, cascades codes are \emph{optimum} exact regenerating codes for any set of parameters $(n,k,d)$, and achieve the optimum storage-bandwidth trade-off.

	On the other hand, the proposed construction universally achieves the optimum trade-off of any system parameters with a known lower bound: Those are MBR codes (see Corollary~\ref{cor:MSR} and\cite{dimakis2010network}), MSR codes (see Corollary~\ref{cor:MSR} and \cite{dimakis2010network}),  an interior operating point on the cut-set bound (see Corollary~\ref{cor:mu=k-1} and \cite{dimakis2010network}), linear  codes with $k=d$  \cite{elyasi2015linear, prakash2015storage, duursma2015shortened}, and an optimum code for an $(n,k,d)=(5,3,4)$ system, for which a matching lower bound is provided in \cite{TianSCITL}. These facts altogether support the conjecture regarding the optimality of the codes in general. 
	
	\subsection{Future Work}
	The main remaining open problem to be addressed is to provide a lower bound for the trade-off between the storage and repair-bandwidth of exact-repair regenerating codes. 
	As mentioned above, we conjecture  that a tight lower bound will match with the trade-off achieved by cascade codes, indicating that the proposed codes are optimal. 
	
	The proposed cascade codes, however, can be improved from several different aspects. One major concern  is in regard with  the sub-packetization. Even though the sub-packetization of the proposed codes is independent of the number of nodes $n$, it is  exponential in parameter $k$. An interesting question is a whether an identical trade-off can be achieved using a code with a sub-packetization that is sub-exponential and independent of $n$. Multiple failures repair is another interesting problem to be studied. More importantly, the problem of dynamic repair \cite{ye2017explicit,mahdaviani2018bandwidth,mahdaviani2018product}, referring to the flexibility of varying the number of helper nodes with $d\in [k:n-1]$ (without changing the underlying code) is of both practical and theoretical interest.  
	
	Recently, there has been considerable attention to  \emph{clustered} distributed storage systems \cite{sohn2018capacity,sohn2018class}. A modified version of the proposed construction might be applicable to such clustered systems. Finally, the exact-repair regenerating codes can be viewed in the context of interference alignment problem (see for e.g. \cite{suh2011exact}), where the repair scenario is equivalent to aligning and canceling the interference (mismatch) between a failed symbol and a coded symbol of a helper node. Therefore, the techniques and results developed in  this paper might be also applicable to the design of interference alignment codes for wireless communication.


\begin{thebibliography}{10}
\providecommand{\url}[1]{#1}
\csname url@samestyle\endcsname
\providecommand{\newblock}{\relax}
\providecommand{\bibinfo}[2]{#2}
\providecommand{\BIBentrySTDinterwordspacing}{\spaceskip=0pt\relax}
\providecommand{\BIBentryALTinterwordstretchfactor}{4}
\providecommand{\BIBentryALTinterwordspacing}{\spaceskip=\fontdimen2\font plus
\BIBentryALTinterwordstretchfactor\fontdimen3\font minus
  \fontdimen4\font\relax}
\providecommand{\BIBforeignlanguage}[2]{{%
\expandafter\ifx\csname l@#1\endcsname\relax
\typeout{** WARNING: IEEEtran.bst: No hyphenation pattern has been}%
\typeout{** loaded for the language `#1'. Using the pattern for}%
\typeout{** the default language instead.}%
\else
\language=\csname l@#1\endcsname
\fi
#2}}
\providecommand{\BIBdecl}{\relax}
\BIBdecl

\bibitem{elyasi2018cascade}
M.~Elyasi and S.~Mohajer, ``A cascade code construction for (n, k, d)
  distributed storage systems,'' in \emph{2018 IEEE International Symposium on
  Information Theory (ISIT)}.\hskip 1em plus 0.5em minus 0.4em\relax IEEE,
  2018, pp. 1241--1245.

\bibitem{ghemawat2003google}
S.~Ghemawat, H.~Gobioff, and S.-T. Leung, \emph{The Google file system}.\hskip
  1em plus 0.5em minus 0.4em\relax ACM, 2003, vol.~37, no.~5.

\bibitem{sathiamoorthy2013xoring}
M.~Sathiamoorthy, M.~Asteris, D.~Papailiopoulos, A.~G. Dimakis, R.~Vadali,
  S.~Chen, and D.~Borthakur, ``Xoring elephants: Novel erasure codes for big
  data,'' in \emph{Proceedings of the VLDB Endowment}, vol.~6, no.~5.\hskip 1em
  plus 0.5em minus 0.4em\relax VLDB Endowment, 2013, pp. 325--336.

\bibitem{huang2012erasure}
C.~Huang, H.~Simitci, Y.~Xu, A.~Ogus, B.~Calder, P.~Gopalan, J.~Li, S.~Yekhanin
  \emph{et~al.}, ``Erasure coding in windows azure storage.'' in \emph{Usenix
  annual technical conference}.\hskip 1em plus 0.5em minus 0.4em\relax Boston,
  MA, 2012, pp. 15--26.

\bibitem{dabek2004designing}
F.~Dabek, J.~Li, E.~Sit, J.~Robertson, M.~F. Kaashoek, and R.~Morris,
  ``Designing a {DHT} for low latency and high throughput.'' in \emph{NSDI},
  vol.~4, 2004, pp. 85--98.

\bibitem{rhea2001maintenance}
S.~Rhea, C.~Wells, P.~Eaton, D.~Geels, B.~Zhao, H.~Weatherspoon, and
  J.~Kubiatowicz, ``Maintenance-free global data storage,'' \emph{IEEE internet
  computing}, no.~5, pp. 40--49, 2001.

\bibitem{bhagwan2004total}
R.~Bhagwan, K.~Tati, Y.~Cheng, S.~Savage, and G.~M. Voelker, ``Total recall:
  System support for automated availability management.'' in \emph{Nsdi},
  vol.~4, 2004, pp. 25--25.

\bibitem{dimakis2010network}
A.~G. Dimakis, P.~B. Godfrey, Y.~Wu, M.~J. Wainwright, and K.~Ramchandran,
  ``Network coding for distributed storage systems,'' \emph{IEEE Trans. Inf.
  Theory}, vol.~56, no.~9, pp. 4539--4551, 2010.

\bibitem{ho2006random}
T.~Ho, M.~M{\'e}dard, R.~Koetter, D.~R. Karger, M.~Effros, J.~Shi, and
  B.~Leong, ``A random linear network coding approach to multicast,''
  \emph{IEEE Transactions on Information Theory}, vol.~52, no.~10, pp.
  4413--4430, 2006.

\bibitem{wu2010existence}
Y.~Wu, ``Existence and construction of capacity-achieving network codes for
  distributed storage,'' \emph{IEEE Journal on Selected Areas in
  Communications}, vol.~28, no.~2, 2010.

\bibitem{rashmi2011optimal}
K.~V. Rashmi, N.~B. Shah, and P.~V. Kumar, ``Optimal exact-regenerating codes
  for distributed storage at the {MSR} and {MBR} points via a product-matrix
  construction,'' \emph{IEEE Trans. Inf. Theory}, vol.~57, no.~8, pp.
  5227--5239, 2011.

\bibitem{shah2012distributed}
N.~B. Shah, K.~V. Rashmi, P.~V. Kumar, and K.~Ramchandran, ``Distributed
  storage codes with repair-by-transfer and nonachievability of interior points
  on the storage-bandwidth tradeoff,'' \emph{IEEE Transactions on Information
  Theory}, vol.~58, no.~3, pp. 1837--1852, 2012.

\bibitem{tian2014characterizing}
C.~Tian, ``Characterizing the rate region of the (4, 3, 3) exact-repair
  regenerating codes,'' \emph{IEEE Journal on Selected Areas in
  Communications}, vol.~32, no.~5, pp. 967--975, 2014.

\bibitem{yeung1997framework}
R.~W. Yeung, ``A framework for linear information inequalities,'' \emph{IEEE
  Transactions on Information Theory}, vol.~43, no.~6, pp. 1924--1934, 1997.

\bibitem{cadambe2010distributed}
V.~R. Cadambe, S.~A. Jafar, and H.~Maleki, ``Distributed data storage with
  minimum storage regenerating codes-exact and functional repair are
  asymptotically equally efficient,'' \emph{arXiv preprint arXiv:1004.4299},
  2010.

\bibitem{suh2010existence}
C.~Suh and K.~Ramchandran, ``On the existence of optimal exact-repair {MDS}
  codes for distributed storage,'' \emph{arXiv preprint arXiv:1004.4663}, 2010.

\bibitem{cullina2009searching}
D.~Cullina, A.~G. Dimakis, and T.~Ho, ``Searching for minimum storage
  regenerating codes,'' \emph{arXiv preprint arXiv:0910.2245}, 2009.

\bibitem{lin2015unified}
S.-J. Lin, W.-H. Chung, Y.~S. Han, and T.~Y. Al-Naffouri, ``A unified form of
  exact-{MSR} codes via product-matrix frameworks,'' \emph{IEEE Transactions on
  Information Theory}, vol.~61, no.~2, pp. 873--886, 2015.

\bibitem{goparaju2017minimum}
S.~Goparaju, A.~Fazeli, and A.~Vardy, ``Minimum storage regenerating codes for
  all parameters,'' \emph{IEEE Transactions on Information Theory}, vol.~63,
  no.~10, pp. 6318--6328, 2017.

\bibitem{tamo2013zigzag}
I.~Tamo, Z.~Wang, and J.~Bruck, ``Zigzag codes: {MDS} array codes with optimal
  rebuilding,'' \emph{IEEE Transactions on Information Theory}, vol.~59, no.~3,
  pp. 1597--1616, 2013.

\bibitem{raviv2017constructions}
N.~Raviv, N.~Silberstein, and T.~Etzion, ``Constructions of high-rate minimum
  storage regenerating codes over small fields,'' \emph{IEEE Transactions on
  Information Theory}, vol.~63, no.~4, pp. 2015--2038, 2017.

\bibitem{wang2016explicit}
Z.~Wang, I.~Tamo, and J.~Bruck, ``Explicit minimum storage regenerating
  codes,'' \emph{IEEE Transactions on Information Theory}, vol.~62, no.~8, pp.
  4466--4480, 2016.

\bibitem{li2015framework}
J.~Li, X.~Tang, and U.~Parampalli, ``A framework of constructions of minimal
  storage regenerating codes with the optimal access/update property,''
  \emph{IEEE Transactions on Information theory}, vol.~61, no.~4, pp.
  1920--1932, 2015.

\bibitem{papailiopoulos2013repair}
D.~S. Papailiopoulos, A.~G. Dimakis, and V.~R. Cadambe, ``Repair optimal
  erasure codes through hadamard designs,'' \emph{IEEE Transactions on
  Information Theory}, vol.~59, no.~5, pp. 3021--3037, 2013.

\bibitem{li2016optimal}
J.~Li and X.~Tang, ``Optimal exact repair strategy for the parity nodes of the
  $(k+ 2, k) $ zigzag code,'' \emph{IEEE Transactions on Information Theory},
  vol.~62, no.~9, pp. 4848--4856, 2016.

\bibitem{wang2011codes}
Z.~Wang, I.~Tamo, and J.~Bruck, ``On codes for optimal rebuilding access,'' in
  \emph{Communication, Control, and Computing (Allerton), 2011 49th Annual
  Allerton Conference on}.\hskip 1em plus 0.5em minus 0.4em\relax IEEE, 2011,
  pp. 1374--1381.

\bibitem{balaji2018tight}
S.~Balaji and P.~V. Kumar, ``A tight lower bound on the sub-packetization level
  of optimal-access {MSR} and {MDS} codes,'' in \emph{2018 IEEE International
  Symposium on Information Theory (ISIT)}.\hskip 1em plus 0.5em minus
  0.4em\relax IEEE, 2018, pp. 2381--2385.

\bibitem{sasidharan2015high}
B.~Sasidharan, G.~K. Agarwal, and P.~V. Kumar, ``A high-rate {MSR} code with
  polynomial sub-packetization level,'' in \emph{Information Theory (ISIT),
  2015 IEEE International Symposium on}.\hskip 1em plus 0.5em minus 0.4em\relax
  IEEE, 2015, pp. 2051--2055.

\bibitem{rawat2016progress}
A.~S. Rawat, O.~O. Koyluoglu, and S.~Vishwanath, ``Progress on high-rate {MSR}
  codes: Enabling arbitrary number of helper nodes,'' \emph{arXiv preprint
  arXiv:1601.06362}, 2016.

\bibitem{ye2017explicitnearly}
M.~Ye and A.~Barg, ``Explicit constructions of optimal-access {MDS} codes with
  nearly optimal sub-packetization,'' \emph{IEEE Transactions on Information
  Theory}, vol.~63, no.~10, pp. 6307--6317, 2017.

\bibitem{li2018generic}
J.~Li, X.~Tang, and C.~Tian, ``A generic transformation to enable optimal
  repair in {MDS} codes for distributed storage systems,'' \emph{IEEE
  Transactions on Information Theory}, vol.~64, no.~9, pp. 6257--6267, 2018.

\bibitem{li2018alternative}
------, ``An alternative generic transformation for optimal repair bandwidth
  and rebuilding access in {MDS} codes,'' in \emph{2018 IEEE International
  Symposium on Information Theory (ISIT)}.\hskip 1em plus 0.5em minus
  0.4em\relax IEEE, 2018, pp. 1894--1898.

\bibitem{sasidharan2016explicit}
B.~Sasidharan, M.~Vajha, and P.~V. Kumar, ``An explicit, coupled-layer
  construction of a high-rate {MSR} code with low sub-packetization level,
  small field size and all-node repair,'' \emph{arXiv preprint
  arXiv:1607.07335}, 2016.

\bibitem{ye2017explicit}
M.~Ye and A.~Barg, ``Explicit constructions of high-rate {MDS} array codes with
  optimal repair bandwidth,'' \emph{IEEE Trans. Inf. Theory}, vol.~63, no.~4,
  pp. 2001--2014, 2017.

\bibitem{tian2015layered}
C.~Tian, B.~Sasidharan, V.~Aggarwal, V.~A. Vaishampayan, and P.~V. Kumar,
  ``Layered exact-repair regenerating codes via embedded error correction and
  block designs,'' \emph{IEEE Trans. Inf. Theory}, vol.~61, no.~4, pp.
  1933--1947, 2015.

\bibitem{elyasi2015linear}
M.~Elyasi, S.~Mohajer, and R.~Tandon, ``Linear exact repair rate region of
  $(k+1, k, k)$ distributed storage systems: A new approach,'' in
  \emph{Information Theory Proceedings (ISIT), 2015 IEEE International
  Symposium on}.\hskip 1em plus 0.5em minus 0.4em\relax IEEE, 2015, pp.
  2061--2065.

\bibitem{prakash2015storage}
N.~Prakash and M.~N. Krishnan, ``The storage-repair-bandwidth trade-off of
  exact repair linear regenerating codes for the case $ d= k= n-1$,'' in
  \emph{Information Theory Proceedings (ISIT), 2015 IEEE International
  Symposium on}.\hskip 1em plus 0.5em minus 0.4em\relax IEEE, 2015, pp. 859 --
  863.

\bibitem{duursma2015shortened}
I.~M. Duursma, ``Shortened regenerating codes,'' \emph{arXiv preprint
  arXiv:1505.00178}, 2015.

\bibitem{elyasi2015probabilistic}
M.~Elyasi and S.~Mohajer, ``A probabilistic approach towards exact-repair
  regeneration codes,'' in \emph{Communication, Control, and Computing
  (Allerton), 2015 53rd Annual Allerton Conference on}.\hskip 1em plus 0.5em
  minus 0.4em\relax IEEE, 2015, pp. 865--872.

\bibitem{elyasi2016determinant}
------, ``Determinant coding: A novel framework for exact-repair regenerating
  codes,'' \emph{IEEE Transactions on Information Theory}, vol.~62, no.~12, pp.
  6683--6697, Dec 2016.

\bibitem{elyasi2018newndd}
------, ``Determinant codes with helper-independent repair for single and
  multiple failures,'' \emph{IEEE Transactions on Information Theory}, vol.~65,
  no.~9, pp. 5469--5483, Sep. 2019.

\bibitem{goparaju2014new}
S.~Goparaju, S.~El~Rouayheb, and R.~Calderbank, ``New codes and inner bounds
  for exact repair in distributed storage systems,'' in \emph{Information
  Theory (ISIT), 2014 IEEE International Symposium on}.\hskip 1em plus 0.5em
  minus 0.4em\relax IEEE, 2014, pp. 1036--1040.

\bibitem{senthoor2015improved}
K.~Senthoor, B.~Sasidharan, and P.~V. Kumar, ``Improved layered regenerating
  codes characterizing the exact-repair storage-repair bandwidth tradeoff for
  certain parameter sets,'' in \emph{Information Theory Workshop (ITW), 2015
  IEEE}.\hskip 1em plus 0.5em minus 0.4em\relax IEEE, 2015.

\bibitem{elyasi2017scalable}
M.~Elyasi and S.~Mohajer, ``Scalable (n, k, d) exact-repair regenerating codes
  with small repair bandwidth,'' in \emph{Communications (ICC), 2017 IEEE
  International Conference on}.\hskip 1em plus 0.5em minus 0.4em\relax IEEE,
  2017, pp. 1--7.

\bibitem{elyasi2017exact}
------, ``Exact-repair trade-off for (n, k= d- 1, d) regenerating codes,'' in
  \emph{Communication, Control, and Computing (Allerton), 2017 55th Annual
  Allerton Conference on}.\hskip 1em plus 0.5em minus 0.4em\relax IEEE, 2017,
  pp. 934--941.

\bibitem{horn1990matrix}
R.~A. Horn, R.~A. Horn, and C.~R. Johnson, \emph{Matrix analysis}.\hskip 1em
  plus 0.5em minus 0.4em\relax Cambridge university press, 1990.

\bibitem{schechter1959inversion}
S.~Schechter, ``On the inversion of certain matrices,'' \emph{Mathematical
  Tables and Other Aids to Computation}, vol.~13, no.~66, pp. 73--77, 1959.

\bibitem{shah2012interference}
N.~B. Shah, K.~Rashmi, P.~V. Kumar, and K.~Ramchandran, ``Interference
  alignment in regenerating codes for distributed storage: Necessity and code
  constructions,'' \emph{IEEE Transactions on Information Theory}, vol.~58,
  no.~4, pp. 2134--2158, 2012.

\bibitem{TianSCITL}
\BIBentryALTinterwordspacing
C.~Tian, ``Scitl: Solutions of computed information-theoretic limits.''
  [Online]. Available: \url{http://web.eecs.utk.edu/~ctian1/SCITL.html}
\BIBentrySTDinterwordspacing

\bibitem{mahdaviani2018bandwidth}
K.~Mahdaviani, A.~Khisti, and S.~Mohajer, ``Bandwidth adaptive error resilient
  {MBR} exact repair regenerating codes,'' \emph{IEEE Transactions on
  Information Theory}, 2018.

\bibitem{mahdaviani2018product}
K.~Mahdaviani, S.~Mohajer, and A.~Khisti, ``Product matrix {MSR} codes with
  bandwidth adaptive exact repair,'' \emph{IEEE Transactions on Information
  Theory}, vol.~64, no.~4, pp. 3121--3135, 2018.

\bibitem{sohn2018capacity}
J.-y. Sohn, B.~Choi, S.~W. Yoon, and J.~Moon, ``Capacity of clustered
  distributed storage,'' \emph{IEEE Transactions on Information Theory}, 2018.

\bibitem{sohn2018class}
J.-y. Sohn, B.~Choi, and J.~Moon, ``A class of {MSR} codes for clustered
  distributed storage,'' \emph{arXiv preprint arXiv:1801.02014}, 2018.

\bibitem{suh2011exact}
C.~Suh and K.~Ramchandran, ``Exact-repair {MDS} code construction using
  interference alignment,'' \emph{IEEE Transactions on Information Theory},
  vol.~57, no.~3, pp. 1425--1442, 2011.

\bibitem{mohajer2008transmission}
S.~Mohajer, S.~Diggavi, C.~Fragouli, and D.~Tse, ``Transmission techniques for
  relay-interference networks,'' in \emph{2008 46th Annual Allerton Conference
  on Communication, Control, and Computing}.\hskip 1em plus 0.5em minus
  0.4em\relax IEEE, 2008, pp. 467--474.

\bibitem{mohajer2011approximate}
------, ``Approximate capacity of a class of gaussian interference-relay
  networks,'' \emph{IEEE Transactions on Information Theory}, vol.~57, no.~5,
  pp. 2837--2864, 2011.

\bibitem{gou2012aligned}
T.~Gou, S.~A. Jafar, C.~Wang, S.-W. Jeon, and S.-Y. Chung, ``Aligned
  interference neutralization and the degrees of freedom of the 2x2
  interference channel,'' \emph{IEEE Transactions on Information Theory},
  vol.~58, no.~7, pp. 4381--4395, 2012.

\bibitem{Oppenheim}
A.~V. Oppenheim, A.~S. Willsky, and S.~H. Nawab, \emph{Signals \&Amp; Systems
  (2Nd Ed.)}.\hskip 1em plus 0.5em minus 0.4em\relax Upper Saddle River, NJ,
  USA: Prentice-Hall, Inc., 1996.

\bibitem{forouzan2016region}
A.~Forouzan, ``Region of convergence of derivative of z transform,''
  \emph{Electronics Letters}, vol.~52, no.~8, pp. 617--619, 2016.

\end{thebibliography}

	\appendices
	
	\section{Proof of Corollaries~\ref{cor:MSR} and ~\ref{cor:mu=k-1}}
	\label{app:proof:cor}
	\begin{proof}[Proof of Corollary~\ref{cor:MSR}]
		For the MBR point, corresponding to $\mu=1$,  we have
		\begin{align}
		\begin{split}
		\alpha_{\mathsf{MBR}} \triangleq\alpha(k,d;1) &= \sum_{m=0} ^{1} (d-k)^{1-m} \binom{k}{m} =(d-k)+k=d,\\
		\beta_{\mathsf{MBR}} \triangleq\beta(k,d;\mu) &= \sum_{m=0} ^{1} (d-k)^{1-m} \binom{k-1}{m-1}=1,\\
		F_{\mathsf{MBR}} \triangleq F(k,d;\mu) &= \sum_{m=0} ^{1} k(d-k)^{1-m} \binom{k}{m}-\binom{k}{1+1}=kd-\binom{k}{2}=\frac{k(2d-k+1)}{2}.
		\end{split}
		\end{align}
		This triple satisfies $\left(\frac{\alpha_{\mathsf{MBR}}}{F_{\mathsf{MBR}}}, \frac{\beta_{\mathsf{MBR}}}{F_{\mathsf{MBR}}}\right)=\left(\frac{2d}{k(2d-k+1)}, \frac{2}{k(2d-k+1)}\right)$, which is the characteristic of the MBR point \cite{dimakis2010network}. 
		
		Similarly, for $\mu=k$ we have
		\begin{align}
		\begin{split}
		\alpha_{\mathsf{MSR}} \triangleq \alpha(k,d;\mu) &= \sum_{m=0} ^{k} (d-k)^{k-m} \binom{k}{m} =(d-k+1)^k,\\
		\beta_{\mathsf{MSR}} \triangleq \beta(k,d;\mu) &= \sum_{m=0} ^{k} (d-k)^{k-m} \binom{k-1}{m-1} = \sum_{m=0} ^{k-1} (d-k)^{k-1-m} \binom{k-1}{m} = (d-k+1)^{k-1},\\
		F_{\mathsf{MSR}} \triangleq F(k,d;\mu) &= \sum_{m=0} ^{k} k(d-k)^{k-m} \binom{k}{m}-\binom{k}{k+1}=k(d-k+1)^k,\\
		\end{split}
		\end{align}
		where we shifted the summation variable in the  evaluation of $\beta_{\mathsf{MSR}} $, and  used the fact that $\binom{k}{k+1}=0$ 
		 to simplify $F_{\mathsf{MSR}}$.  These parameters satisfy $\left(\frac{\alpha_{\mathsf{MSR}}}{F_{\mathsf{MSR}}}, \frac{\beta_{\mathsf{MSR}}}{F_{\mathsf{MSR}}}\right)=\left(\frac{1}{k}, \frac{1}{k(d-k+1)}\right)$, which characterizes  the MSR point \cite{dimakis2010network}.
	\end{proof}
	
	\begin{proof}[Proof of Corollary~\ref{cor:mu=k-1}]
		The cut-set bound in~\eqref{eq:func:tradeoff} given by $F \leq \sum_{i=1}^{k} \min(\alpha,(d-i+1)\beta)$ reduces to 
		\[F \leq (k-1) \alpha + (d-k+1)\beta,\] 
		for $(d-k+1)\beta \leq \alpha \leq (d-k) \beta$.
		The latter bound is satisfied with equality by the parameters of a cascade code with mode $\mu=k-1$. To show this claim, we use the expressions for $\alpha$ and $\beta$ in~\eqref{eq:params}, and write
		\begin{align}
		(k-1)\alpha(k,d;k-1) +(d-k+1)\beta(k,d;k-1)
		&= k\alpha(k,d;k-1) +(d-k)\beta(k,d;k-1) \nonumber \\ 
		& \hspace{20pt} - \left[\alpha(k,d;k-1) - \beta(k,d;k-1)\right]\nonumber\\	
		&=\sum_{m=0}^{k-1} k(d-k)^{k-1-m}\binom{k}{m}+\sum_{m=0}^{k-1} (d-k)^{k-m} \binom{k-1}{m-1} \nonumber \\ &\hspace{20pt} - \sum_{m=0}^{k-1} (d-k)^{k-1-m}   \left[\binom{k}{m}-\binom{k-1}{m-1}\right] \nonumber\\
    	&=\sum_{m=0}^{k-1} k(d-k)^{k-1-m}\binom{k}{m}+\sum_{m=0}^{k-1} (d-k)^{k-m} \binom{k-1}{m-1} \nonumber\\ &\hspace{20pt}- \sum_{m=0}^{k-1} (d-k)^{k-1-m}   \binom{k-1}{m}\label{eq:cutset:1} \\
    	&=\sum_{m=0}^{k-1} k(d-k)^{k-1-m}\binom{k}{m}+\sum_{m=0}^{k-1} (d-k)^{k-m} \binom{k-1}{m-1} \nonumber\\ &\hspace{20pt} -\sum_{m=1}^{k} (d-k)^{k-m}\binom{k-1}{m-1}\nonumber\\
    	&=\sum_{m=0}^{k-1} k(d-k)^{k-1-m}\binom{k}{m}+(d-k)^{k-0} \binom{k-1}{0-1} \nonumber\\ 
    	&\hspace{20pt} - (d-k)^{k-k} \binom{k-1}{k-1}\nonumber\\
    	&= \sum_{m=0}^{k-1} k(d-k)^{k-1-m}\binom{k}{m} -1\nonumber\\
    	&= F(k,d;k-1), 
    	\end{align}
    	where we used Pascal's identity in~\eqref{eq:cutset:1}
		Hence, this point satisfies the cut-set bound and it is optimum.
	\end{proof}
	
		\section{Proof of Node Repairability for Singed Determinant Codes}
	\label{app:prf:prop:ndd:repair}
	In this section we present the proof of Proposition~\ref{prop:ndd:repair}. We start from the RHS of \eqref{eq:ndd:repair} and show it is equal to the LHS.
 	\begin{align}
	&\sum_{i\in \sen I} (-1)^{\s{\e{D}}{i}+\ind{\sen I}{i}} \left[\erep{f}{\e{D}}\right]_{i,\sen I \setminus \seq{i}}\nonumber\\
	&=\sum_{i \in \sen I} (-1)^{\s{\e{D}}{i} + \ind{\sen I}{i}} \left[ \e{D} \cdot \repMat{f}{m}\right]_{i,\sen I\setminus \set{i}}\nonumber\\
	&= \sum_{i \in \sen I} (-1)^{\s{\e{D}}{i}+ \ind{\sen I}{i}} \sum_{\substack{\sen L\subseteq \intv{d} \\ \size{\sen L}=m}}\e{D}_{i,\sen L}  \cdot \repMat{f}{m}_{\sen L ,\sen I \setminus \seq{i}} \nonumber\\ 
	&= \sum_{i \in \sen I} (-1)^{\s{\e{D}}{i} + \ind{\sen I}{i}} \sum_{ y\in \intv{d}\setminus{(\sen I \setminus{\set{i}})} } \e{D}_{i,\sen (\sen I \setminus{\set{i}}) \cup \set{y}}  \cdot \repMat{f}{m}_{\sen (\sen I \setminus{\set{i}}) \cup \set{y},\sen I \setminus \seq{i}} \label{eq:rep:iden:3}\\ 
	&= \sum_{i \in \sen I} (-1)^{\s{\e{D}}{i} + \ind{\sen I}{i}} \left[ \e{D}_{i,\sen I} \   \repMat{f}{m}_{\sen I ,\sen I \setminus \seq{i}} + \sum_{{y\in \intv{d}\setminus{\sen I}} }\e{D}_{i,\sen (\sen I \setminus{\set{i}}) \cup \set{y}}  \cdot\repMat{f}{m}_{\sen (\sen I \setminus{\set{i}}) \cup \set{y} ,\sen I \setminus \seq{i}}\right] \label{eq:rep:iden:4}\\ 
	&= \sum_{i \in \sen I} (-1)^{\s{\e{D}}{i} + \ind{\sen I}{i}} \left[ (-1)^{\s{\e{D}}{i}+\ind{\sen I}{i}} \psi_{f,i}\e{D}_{i,\sen I} + \sum_{y\in \intv{d}\setminus{\sen I} }(-1)^{\s{\e{D}}{y}+\ind{(\sen I \setminus \seq{i})\cup \set{y}}{y}}\psi_{f,y} \e{D}_{i,\sen (\sen I \setminus{\set{i}}) \cup \set{y}}\right] \label{eq:rep:iden:6}  \\ 
	&= \sum_{i \in \sen I} \psi_{f,i}\e{D}_{i,\sen I} +\sum_{y\in \intv{d}\setminus{\sen I} } \sum_{i \in \sen I}  (-1)^{\s{\e{D}}{i} + \s{\e{D}}{y} + \ind{\sen I}{i}+\ind{(\sen I \setminus \seq{i})\cup \set{y}}{y} }  \psi_{f,y} \e{D}_{i,\sen (\sen I \setminus{\set{i}}) \cup \set{y}}   \label{eq:rep:iden:7}\\ 
	&= \sum_{i \in \sen I} \psi_{f,i}\e{D}_{i,\sen I} + \sum_{y\in \intv{d}\setminus{\sen I}} \sum_{i \in \sen I}   (-1)^{\s{\e{D}}{i}+\s{\e{D}}{y} + \ind{\sen I \cup \set{y}}{y} +\ind{\sen I \cup \set{y}}{i} +1 }  \psi_{f,y} \e{D}_{i,\sen (\sen I \setminus{\set{i}}) \cup \set{y}} \label{eq:rep:iden:8}  \\ 
	&= \sum_{i \in \sen I} \psi_{f,i}\e{D}_{i,\sen I} +\sum_{y\in \intv{d}\setminus{\sen I} }(-1)^{\s{\e{D}}{y}+\ind{\sen I \cup \set{y}}{y}+1} \psi_{f,y}\sum_{i \in \sen I} (-1)^{\ind{\sen I \cup \set{y}}{i}} (-1)^{\s{\e{D}}{i}} \e{D}_{i,\sen (\sen I \setminus{\set{i}}) \cup \set{y}} \label{eq:rep:iden:9} \\
	&= \sum_{i \in \sen I} \psi_{f,i}\e{D}_{i,\sen I} +\sum_{y\in \intv{d}\setminus{\sen I} }(-1)^{\s{\e{D}}{y}+\ind{\sen I \cup \set{y}}{y}+1} \psi_{f,y}\sum_{i \in \sen I} (-1)^{\ind{\sen I \cup \set{y}}{i}}  w_{i,\sen I \cup \set{y}}  \label{eq:rep:iden:10} \\
	&= \sum_{i \in \sen I} \psi_{f,i}\e{D}_{i,\sen I} +\sum_{ y\in \intv{d}\setminus{\sen I} }(-1)^{\s{\e{D}}{y}+\ind{\sen I \cup \set{y}}{y}+1} \psi_{f,y} \left[ (-1)^{\ind{\sen I \cup \set{y}}{y}+1}  w_{y,\sen I \cup \set{y}} \right] \label{eq:rep:iden:11} \\
	&= \sum_{i \in \sen I} \psi_{f,i}\e{D}_{i,\sen I} +\sum_{y\in \intv{d}\setminus{\sen I} }(-1)^{\s{\e{D}}{y} + \ind{\sen I \cup \set{y}}{y}+1} \psi_{f,y} \left[
	(-1)^{\ind{\sen I \cup \set{y}}{y}+1} (-1)^{\s{\e{D}}{y}} \e{D}_{y,\sen I } \right]   \label{eq:rep:iden:12} \\
	&= \sum_{i \in \sen I}  \psi_{f,i}\e{D}_{i,\sen I} +\sum_{y\in \intv{d}\setminus{\sen I} } \psi_{f,y} \e{D}_{y,\sen I} \label{eq:rep:iden:13} \\
	&= \sum_{i \in \intv{d}} \psi_{f,i}\e{D}_{i,\sen I} =
	\left[\enc_{f,:} \cdot \e{D} \right]_{\sen I}, \label{eq:rep:iden:14}
	\end{align}
    In this proof we used the following facts:
	\begin{itemize}
		\item In \eqref{eq:rep:iden:3} we used the definition of $\repMat{f}{m}$ in \eqref{eq:rep:enc}, where the entry  $\repMat{f}{m}_{\sen L, \sen I \setminus\set{i}}$ is non-zero only if $\sen L$ includes $\sen I \setminus\set{i}$. This implies that for non zero $\repMat{f}{m}_{\sen L, \sen I \setminus\set{i}}$, $\sen L$ should satisfy $\sen L = (\sen I \setminus\set{i})\cup \set{y}$ for some $y\in \intv{d}\setminus(\sen I \setminus\set{i})$;
		\item In ~\eqref{eq:rep:iden:4}, we split the summation into two cases:  $y=i$ and  $ y \neq i$;
		\item In ~\eqref{eq:rep:iden:6}, we replaced  $\repMat{f}{m}_{\sen (\sen I \setminus{\set{i}}) \cup \set{y} ,\sen I \setminus \seq{i}}$ by $(-1)^{\s{\e{D}}{y} + \ind{(\sen I\setminus\set{i})\cup \set{y}}{y}} \psi_{f,y}$ from its definition in \eqref{eq:rep:enc}; 
		\item In ~\eqref{eq:rep:iden:7} the two summations over $i$ and $y$ are swapped;  
		\item In ~\eqref{eq:rep:iden:8}, we used the definition of $\ind{\cdot}{\cdot}$ function to write 
		\begin{align*}
		&\ind{\sen I \cup \set{y}}{i} +\ind{\sen I \cup \set{y}}{y} \\
		& = 
		\left|\left\{u\in \sen I \cup \set{y}: u\leq i\right\}\right| + \left|\left\{u\in \sen I \cup \set{y}: u\leq y\right\}\right|  \nonumber\\
		&= 
		\left|\left\{u\in \sen I : u\leq i\right\}\right| + \mathbbm{1}\left[y\leq i\right] + \left|\left\{u\in (\sen I\setminus\set{i}) \cup \set{y}: u\leq y\right\}\right| + \mathbbm{1}\left[i\leq y\right]\nonumber\\
		&= \left|\left\{u\in \sen I : u\leq i\right\}\right| + \mathbbm{1}\left[y< i\right] + \left|\left\{u\in (\sen I\setminus\set{i}) \cup \set{y}: u\leq y\right\}\right| + \mathbbm{1}\left[i<y\right]\nonumber\\  
		&=  \ind{\sen I }{i} + \ind{(\sen I\setminus\set{i}) \cup \set{y}}{y} +1.
		\end{align*}   
		 Here, the third equality holds since $i\in \sen I$ and $y\in \intv{d}\setminus \sen I$, which implies $i\neq y$. The last equality holds because $\mathbbm{1}\left[i < y\right] + \mathbbm{1}\left[y < i\right] =1$. This leads to $\ind{\sen I}{i} + \ind{(\sen I \setminus \seq{i})\cup \set{y}}{y}  \equiv \ind{\sen I \cup \set{y}}{y} + \ind{\sen I \cup \set{y}}{i} + 1$ modulo $2$.
		\item In~\eqref{eq:rep:iden:10}  we used the definition of $\e{D}$ in \eqref{eq:def:S}: since $i\notin (\sen I\setminus \set{i}) \cup \set{y}$ then $\e{D}_{i,(\sen I\setminus \set{i}) \cup \set{y}}=(-1)^{\sigma_{\e{D}}(i)}w_{i,\sen I \cup \set{y}}$. A similar argument is used in \eqref{eq:rep:iden:12}; 
		\item In~\eqref{eq:rep:iden:11} we used the parity equation~\eqref{eq:parityeq}. In particular, we have 
		\[\sum_{i \in \sen I \cup \set{y}}(-1)^{\ind{\sen I \cup \set{y}}{i}} w_{i,I \cup \set {y}} = 0,\]  which implies  
		\[\sum_{i \in \sen I } (-1)^{\ind{\sen I \cup \set{y}}{i}}w_{i,I \cup \set {y}}=- (-1)^{\ind{\sen I \cup \set{y}}{y}}w_{y,I \cup \set {y}}.\] 
		\item  In~\eqref{eq:rep:iden:13}, we notice that the overall sign of each term in the summation is positive. 
	\end{itemize}
	This leads to~\eqref{eq:rep:iden:14}, which is exactly the LHS of~\eqref{eq:ndd:repair}. This completes the proof of Proposition~\ref{prop:ndd:repair}.\hfill $\square$
	
	\begin{rem}
		Note that in the chain of equations above we aim to repair the coded symbol at position $\sen I$ of the failed node, which is a linear combination of symbols in column  $\sen I$ of the message matrix. However, the linear combination in \eqref{eq:rep:iden:9} \emph{misses} some of the symbols of column $\sen I$ (i.e., $\e{D}_{i,\sen I}$ when $i\notin \sen I$) and  \emph{includes} symbols from columns of the message matrix  (i.e., $\e{D}_{i,(\sen I \setminus \{i\})\cup\set{y}}$ with $y\neq i$). However, these two interference perfectly cancel each other due to the parity equation in \eqref{eq:parityeq}. This is identical to the notion of interference neutralization, which is well studied in multi-hop wireless networks \cite{mohajer2008transmission, mohajer2011approximate, gou2012aligned}. 
		\label{rem:IN}
	\end{rem}
	
	\section{Semi-Systematic Encoder Matrix}
	\label{app:semi-sys}
	Consider an $(n,k,d)$ regenerating code obtained using an encoder matrix $\enc$ that satisfies Conditions \ref{cond:G} and \ref{cond:P}. Here, we show that we can modify the encoder matrix such that the resulting code becomes semi-systematic, that is, the first $k$  nodes store pure symbols from the message matrix. Consider a general encoder matrix
	\begin{align*}
	\enc_{n \times d} = \left[\begin{array}{c|c} \mathbf{\Gamma}_{n \times k} & \mathbf{\Upsilon}_{n \times (d-k)} \end{array} \right] =\left[\begin{array}{c|c} \bA_{k \times k} & \bB_{k \times (d-k)} \\ \hline \bC_{(n-k) \times k} & \bD_{(n-k) \times (d-k)}  \end{array} \right].
	\end{align*}
	Recall Condition~\ref{cond:G}, that ensures any $k$ rows of $\mathbf{\Gamma}_{n \times k}$ are linearly independent. Thus, $\mathbf{A}_{k\times k}$ is a full-rank and invertible matrix. We can define 
	\begin{align*}
	\bX=\left[\begin{array}{c|c} \bA_{k \times k}^{-1} & -\bA_{k \times k}^{-1} \bB_{k \times (d-k)} \\ \hline \bO_{(d-k) \times k} & \bI_{(d-k) \times (d-k)}  \end{array} \right]
	\end{align*}
	It is easy to verify that $\bX$ is a full-rank matrix, and its inverse is given by 
	\begin{align*}
	\bX^{-1}=\left[\begin{array}{c|c} \bA_{k \times k}& \bB_{k \times (d-k)}\\ \hline \bO_{(d-k) \times k} & \bI_{(d-k) \times (d-k)}\end{array} \right].
	\end{align*}
	
	Therefore, we can modify the encoder matrix to 
	\begin{align*}
	\tilde{\enc}_{n \times d} = \enc_{n \times d} \cdot  \bX 
	&= \left[\begin{array}{c|c} \bI_{k \times k} & \bO_{k \times (d-k)} \\ \hline  \bC \bA^{-1} &\bD-\bC \bA^{-1}  \bB \end{array} \right] = \left[\begin{array}{c|c} \tilde{\mathbf{\Gamma}}_{n \times k} & \tilde{\mathbf{\Upsilon}}_{n \times (d-k)} \end{array} \right].
	\end{align*}
	It is easy to verify that $\tilde{\enc}$ satisfy both Conditions~\ref{cond:G} and \ref{cond:P}. To this end, let $\sen K$ be an arbitrary set of row indices with $\size{\sen K}=k$. We have 
	$ \tilde{\mathbf{\Gamma}}[{\sen K},:] = -\mathbf{\Gamma}[{\sen K},:] \bA_{k \times k}^{-1}$ which is a full-rank matrix, since both $\mathbf{\Gamma}[{\sen K},:]$ and  $\bA_{k \times k}^{-1}$ are full-rank. This shows Condition~\ref{cond:G} holds for $\tilde{\enc}$. Similarly, for an arbitrary set $\sen H\subseteq [n]$ with $\size{\sen H} = d$ we have 
	$\tilde{\mathbf{\Psi}}[{\sen H},:] = \enc[{\sen H},:] \bX$, which is again full-rank, because both $\enc[{\sen H},:]$ are  $\bX$ full-rank. Hence Condition~\ref{cond:P} is also satisfied for $\tilde{\enc}$. The code obtained using the encoder matrix  $\tilde{\enc}$ is semi-systematic, since the content of node $i$ is exactly the symbols in the $i$-th row of the super-message matrix, for $i\in\intv{k}$.

	\section{$\cZ$-Transform for  Evaluation of Code Parameters}
	\label{app:Z}
	\subsection{An Overview of $\cZ$-Transform}
	We will use the $\cZ$-transform to solve the recursive equation in \eqref{eq:ell_rec} for $\mdcnt_m$'s, and evaluate the code parameters in \eqref{eq:alpha-in-L}, \eqref{eq:beta-in-L}, and \eqref{eq:F-in-L}. For the sake of completeness, we start with the definition and some of the main properties of this transformation. We refer the reader to \cite{Oppenheim,forouzan2016region} for the details and proofs of the properties listed below.
	
	\begin{defi}
		The two-sided $\cZ$-transform of a sequences\footnote{With slightly abuse of notation, we use  $x_m$ to refer to the sequence $\{x_m\}_{m=-\infty}^{\infty}$ as well as the $m$-th element of this sequence.} $x_m$ is defined as
		\begin{align}
		X(z) = \cZ \set{x_m} = \sum_{m=-\infty}^{\infty} x_m z^{-m},
		\end{align}
		where $z$ is a complex number.  The region of convergence (ROC) of $X(z)$ is defined as the set of all points in the complex plane ($z\in \mathbb{C}$) for which  $X(z)$  converges, that is, 
		\begin{align}
		\text{ROC}_x = \set{ z:\size{ \sum_{m=-\infty}^{\infty} x_m z^{-m} } < \infty }.
		\end{align}
	\end{defi}	
	
	\begin{defi}
		The inverse   $\cZ$-transform of $X(z)$ is defined as a sequence $\{x_m\}_{m=-\infty}^{\infty}$ where
		\begin{align}
		x_m={\cZ}^{-1}\{X(z)\}= \frac{1}{2\pi j} \oint_{C} X(z) z^{m-1} dz, \qquad m\in\mathbb{Z},
		\end{align}
		where $C$ is a counterclockwise closed path encircling the origin and entirely located in the region of convergence, $\text{ROC}_x$.
	\end{defi}		
	
	For a given ROC, there is a one-to-one correspondence between the sequences $x_m$ and its $\cZ$-transform, $X(z)$. Some properties of the $\cZ$-transform as well as some pairs of sequences\footnote{Recall that, in this paper, we defined $\binom{\ell}{m}$ to be zero for $m<0$ and $m>\ell$.} and their $\cZ$-transforms are listed in Table~\ref{table:prop} and Table~\ref{table:pair}, respectively.  
	\begin {table*}[!t]
	    \begin{tabular}{|c|c|c|c|} 
    		\hline
    		& Time Domain & $\cZ$-Domain & ROC\\ \hline
    		Linearity & $w_m = a x_m + b y_m$ & $W(z) = a X(z) + b Y(z)$ & $\text{ROC}_x \cap \text{ROC}_y$\\ \hline
    		Convolution & $w_m = x_m * y_m$ & $W(z) = X(z) Y(z)$ & $\text{ROC}_x \cap \text{ROC}_y$\\ \hline
    		Differentiation & $w_m = m x_m$ & $W(z)= -z\frac{dX(z)}{dz}$ & $\text{ROC}_x$\\ \hline
    		Scaling in the z-domain & $a^{-m} x_m$ & $X(a\cdot z)$ & $\text{ROC}_x/|a|$\\ \hline
    		(Generalized) Accumulation & $w_m = \sum_{\ell = -\infty}^{m} a^{m-t} x_t$ & $W(z) = \frac{1}{1-a z^{-1}} X(z)$  & $\text{ROC}_x \cap \{z: |z|>|a|\}$\\ \hline
    		Time shifting & $w_m = x_{m-b}$ & $ W(z) = z^{-b} X(z)$  & $\text{ROC}_x$  \\ \hline
		\end {tabular}
		\centering
	    \caption{Properties of the $\cZ$-transform.}\label{table:prop}
	\end {table*}
		
		\begin {table*}[!t]
		    \begin{tabular}{|c|c|c|} 
    			\hline
    			Sequence & $\cZ$-Transform & ROC \\ \hline
    			$x_m = \delta(m)$ & $X(z) = 1$ & all $z\in \mathbb{C}$ \\ \hline
    			$x_m=  \binom{r}{m}$ & $X(z) = (1+z^{-1})^r $ & all $z\in \mathbb{C}$ \\ \hline
    			$x_= \binom{m+b-1}{m}a^{m}, \quad b\in \mathbb{Z}^+$ & $X(z) = \frac{1}{(1-az^{-1})^b}$ &  $ \size{z}> \size{a}$ \\ \hline
    			$x_m = \binom{b}{m}a^m, \quad b\in \mathbb{Z}^+$ & $X(z) = (1+az^{-1})^b$ &  all $z\in \mathbb{C}$ \\ \hline
			\end {tabular}
			\centering
		    \caption{Some useful pairs of $\cZ$-transform.} \label{table:pair}
	    \end {table*}
			
			\subsection{Proof of Lemma~\ref{lm:p_rec}}
			\label{sec:lm:p_rec}
			We start from the definition of $p_m$  and use \eqref{eq:ell_rec} to obtain a recursive equation. For any $m$ with $m \neq 0$,  we have
			\begin{align}
			p_m &= \mdcnt_{\mu-m} \nonumber\\&= \sum_{j=(\mu-m)+1}^{\mu} \mdcnt_j \cdot (j-(\mu-m)-1) \binom{d-k+1}{j-(\mu-m)}\label{eq:conv_m>0:1}\\
			&= \sum_{\ell=1}^{m} \mdcnt_{\ell+\mu-m} \cdot (\ell-1) \binom{d-k+1}{\ell}\label{eq:conv_m>0:2}\\
			&= \sum_{\ell=1}^{m} p_{m-\ell} \cdot (\ell-1) \binom{d-k+1}{\ell},\label{eq:conv_m>0:3}
			\end{align}
			where \eqref{eq:conv_m>0:1} is implied by~\eqref{eq:ell_rec}, and in~\eqref{eq:conv_m>0:2} we used a change of variable $\ell=j-\mu+m$. 
			
			Note that $p_m$ can be also written as 
			$ p_m = -p_{m-0} \cdot (0-1) \binom{d-k+1}{0}$, which is of the form of the summands in~\eqref{eq:conv_m>0:2}.   Hence, by including $\ell=0$ in the summation,  we get 
			\begin{align}
			\sum_{\ell=0}^{m} p_{m-\ell}\cdot  (\ell-1) \binom{d-k+1}{\ell} = 0, \qquad m\neq 0.
			\label{eq:conv_m>0}
			\end{align}
			Finally, for $m=0$, we have 	
			\begin{align}
			\sum_{\ell=0}^{0} p_{0-\ell} \cdot (\ell-1) \binom{d-k+1}{\ell} = -p_0 = -\mdcnt_{\mu-0}= -1. 
			\label{eq:conv_m=0}
			\end{align}
			Putting \eqref{eq:conv_m>0} and \eqref{eq:conv_m=0} together, we have 
			\begin{align}
			\sum_{\ell=0}^{m} p_{m-\ell}\cdot  (\ell-1) \binom{d-k+1}{\ell}=-\delta_m.\qquad \forall m\in \mathbb{Z}.
			\label{eq:conv}
			\end{align}
			Next, define a sequence $q_{\ell}=(\ell-1) \binom{d-k+1}{\ell}$ for every integer $\ell$. 
			Then \eqref{eq:conv} can be rewritten as
			\begin{align}
			-\delta_m &= \sum_{\ell=0}^{m} p_{m-\ell}\cdot (\ell-1) \binom{d-k+1}{\ell} \nonumber\\
			&=\sum_{\ell=0}^{m} p_{m-\ell} \cdot q_{\ell}  \nonumber\\			&=\sum_{\ell=-\infty}^{\infty} p_{m-\ell} \cdot q_{\ell}  \label{eq:param:p:rec:3}\\
			&= p_m * q_m,
			\label{eq:param:p:rec}
			\end{align}
			where   \eqref{eq:param:p:rec:3} holds since $q_{\ell}=0$  for $\ell<0$, and  $p_{m-\ell}=\mdcnt_{\mu + (\ell-m)}=0$ is zero for $\ell>m$ (see definition of $\mdcnt_m$ in  \eqref{eq:ell_rec_2}). Here, the operator $*$ in~\eqref{eq:param:p:rec} denotes the convolution between sequences $p_m $ and $q_m$. 
			We can take the $\cZ$-transform from both sides of \eqref{eq:param:p:rec}. Denoting the $\cZ$-transforms of $p_m$ and $q_m$ by $P(z)$ and $Q(z)$, respectively, and using Table~\ref{table:prop} and Table~\ref{table:pair}, we can write 
			\begin{align}
			P(z) Q(z)= -1.
			\label{eq:z:product}
			\end{align}
			The $\cZ$-transform of $q_m$ can be easily found using   Table~\ref{table:prop} and Table~\ref{table:pair} as follows. 
			\begin{align}
			Q(z) &= \cZ \set{q_m}\nonumber \\
			&= \cZ \set{(m-1)\binom{d-k+1}{m}}\nonumber\\
			&= \cZ \set{m\binom{d-k+1}{m}}-\cZ \set{\binom{d-k+1}{m}} \label{eq:Q-Linearity}\\
			&= -z \frac{d\cZ \set{\binom{d-k+1}{m}}}{dz}-\cZ \set{\binom{d-k+1}{m}} \label{eq:Q-Differentiation}\\
			&= -z \frac{d}{dz}  (1+z^{-1})^{d-k+1}-(1+z^{-1})^{d-k+1} \label{eq:Q-Table}\\
			&= -z  (d-k+1) (-z^{-2})(1+z^{-1})^{d-k} - (1+z^{-1})^{d-k+1}\nonumber \\
			&= (1+z^{-1})^{d-k} \left[(d-k)z^{-1}-1\right]. \label{eq:Q}
			\end{align}
			where \eqref{eq:Q-Linearity} holds due to  linearity of the $\cZ$-transform, in \eqref{eq:Q-Differentiation} we used the differentiation effect, and in~\eqref{eq:Q-Table} we used the fourth pair in Table~\ref{table:pair} with  $a=1$ and $b=d-k+1$ . 
			Plugging \eqref{eq:Q} into \eqref{eq:z:product}, we get
			\begin{align}
			P(z) = \frac{-1}{Q(z)}= \frac{1}{1-(d-k)z^{-1}} \left(\frac{1}{1+z^{-1}}\right)^{d-k}, 
			\label{eq:ptransform}
			\end{align}
			with the region of convergence  $\text{ROC}_p=\set{z: \size{z} > \size{d-k}}$. It remains to find $p_m$ from $P(z)$ by computing its inverse $\cZ$-transform.  We have
			\begin{align}
			p_m &= \cZ^{-1}\set{P(z)} \nonumber\\
			&= \sum_{t=-\infty}^{m} (d-k)^{m-t} \cZ^{-1} \set{\left(\frac{1}{1+z^{-1}}\right)^{d-k} }\label{eq:P-GenAcc}\\
			&= \sum_{t=-\infty}^{m} (d-k)^{m-t} \cdot (-1)^t \binom{t+d-k-1}{t} \label{eq:P-Table}\\
			&= \sum_{t=0}^{m}(-1)^t(d-k)^{m-t}\binom{t+d-k-1}{t}. \label{eq:P}
			\end{align}
			where in \eqref{eq:P-GenAcc} we used the generalized accumulation rule in Table~\ref{table:prop} for $a=d-k$. It is worth mentioning that the inverse $\cZ$-transform of $\left(\frac{1}{1+z^{-1}}\right)^{d-k}$ should be taken with respect to variable $t$. To this end, in \eqref{eq:P-Table} we have used the third pair in Table~\ref{table:pair} with $a=-1$ and $b=d-k$.  Finally, in \eqref{eq:P} we have limited the range of $t$ by noticing the fact that the binomial coefficient is zero for $t<0$. 
			This shows the desired identity and completes the proof. \hfill $\square$

			\subsection{Proof of Lemma~\ref{lm:helping:iden}}
			\label{sec:lm:helping:iden}
			Let us define 
			\begin{align*}
			u_\mu 
			&\hspace{-2pt}=\hspace{-5pt} \sum_{m=-\infty}^{\infty} p_{\mu-m} \binom{d+a}{m+b}
			,\nonumber\\
			v_\mu &\hspace{-2pt}=\hspace{-5pt} \sum_{m=-b}^{\mu}\hspace{-2pt}(d-k)^{\mu-m} \binom{k+a}{m+b} \hspace{-3pt}=\hspace{-3pt} \sum_{m=0}^{\mu+b}(d-k)^{\mu+b-m} \binom{k+a}{m},
			\end{align*}
			for every integer $\mu$. 
			The claim of this lemma is equivalent to $u_\mu = v_{\mu}$ for all $\mu\in \mathbb{Z}$. Instead of directly showing in the $\mu$-domain, we will prove that the two sequences are identical in the $z$-domain, and have the same ROCs. We start with sequence $\{u_\mu\}$ and write
			\begin{align}
			U(z) 
			&= \cZ \set{\sum_{m=-\infty}^{\infty} p_{\mu-m} \binom{d+a}{m+b}} \nonumber\\
			&= \cZ  \set{p_{\mu} * \binom{d+a}{\mu+b}} \nonumber\\
			&= \cZ  \set{p_{\mu}} \cdot  \cZ  \set{ \binom{d+a}{\mu+b}} \label{eq:lm:a-b:1}\\
			&= P(z) \cdot     z^{b} \cZ \set{ \binom{d+a}{\mu}} \label{eq:lm:a-b:2}\\
			&=  \frac{1}{1-(d-k)z^{-1}} \left(\frac{1}{1+z^{-1}}\right)^{d-k} \cdot  z^{b}  \cdot  \left(1+z^{-1}\right)^{d+a}
			\label{eq:lm:a-b:3}\\
			&=z^{b} \frac{1}{1-(d-k)z^{-1}} \left(1+z^{-1}\right)^{k+a},
			\end{align}
			where in \eqref{eq:lm:a-b:1} and \eqref{eq:lm:a-b:2} we used the convolution and time-shift properties from Table~\ref{table:prop}, respectively. Moreover, in~\eqref{eq:lm:a-b:3}  we have used  \eqref{eq:ptransform} and  Table~\ref{table:pair} to evaluate the $\cZ$-transforms. Note that $\text{ROC}_u = \text{ROC}_p=\set{z: \size{z} > \size{d-k}}$. 
			
			Similarly, for sequence $\{v_\mu\}$ we have
			\begin{align}
			V(z) 
			&= \cZ \set{\sum_{m=0}^{\mu+b}(d-k)^{\mu+b-m} \binom{k+a}{m}} \nonumber\\
			&= z^{-(-b)} \cdot \cZ \set{\sum_{m=0}^{\mu}(d-k)^{\mu-m} \binom{k+a}{m}} \label{eq:lm:a-b:4}\\
			&= z^{b} \cdot \cZ \set{\sum_{m=-\infty}^{\mu}(d-k)^{\mu-m} \binom{k+a}{m}} \label{eq:lm:a-b:5}\\
			&= z^{b} \cdot \frac{1}{1-(d-k)z^{-1}} \cZ \set{ \binom{k+a}{m}} \label{eq:lm:a-b:6}\\
			&= z^{b} \cdot \frac{1}{1-(d-k)z^{-1}} \left(1+z^{-1}\right)^{k+a},\label{eq:lm:a-b:7}
			\end{align}
			where in \eqref{eq:lm:a-b:4} and \eqref{eq:lm:a-b:6} we used time-shift property and generalized accumulation property from Table~\ref{table:prop}, respectively. 
			Moreover, \eqref{eq:lm:a-b:5} holds because $\binom{k+a}{m}$ is zero for $m<0$, and we used  Table~\ref{table:pair} to evaluate the $\cZ$-transform in~\eqref{eq:lm:a-b:7}. It is worth noting that the ROC of $V(z)$ is given by $\text{ROC}_v = \set{z: \size{z} > \size{d-k}}$, due to the step in \eqref{eq:lm:a-b:6}. Comparing \eqref{eq:lm:a-b:3} and \eqref{eq:lm:a-b:7} we find that $U(z) = V(z)$. Since the two functions in $z$-domain have identical ROCs, their corresponding  sequences $\{u_\mu\}$ and $\{v_\mu\}$ should be also identical. This completes the proof of the lemma. 	\hfill $\square$	
\end{document}